\newcommand{\numberOfIslaLitmusTests}{119}
\newcommand{\numberOfpKVMTests}{14}

\newif\ifworkingdraft
\workingdrafttrue
\workingdraftfalse

\newif\iflinenum
\linenumtrue
\linenumfalse

\newif\ifvmsaTestIndexStandalone
\vmsaTestIndexStandalonefalse

\newif\ifinternaldiscussion
\ifdefined\DocVersion
\internaldiscussiontrue
\else
\internaldiscussionfalse
\fi

\newif\ifTODO
\TODOfalse

\documentclass[11pt]{article}

\usepackage{graphicx}
\usepackage{caption}
\usepackage{hyperref}
\usepackage{float}
\usepackage[svgnames]{xcolor}
\hypersetup{colorlinks=true,linkcolor=DarkGreen,urlcolor=DarkGreen,citecolor=DarkGreen}
\usepackage{etoolbox}
\usepackage{adjustbox}
\usepackage{tikz}
\usetikzlibrary{decorations.pathreplacing}
\usetikzlibrary{arrows.meta}
\makeatletter
\AtBeginEnvironment{thebibliography}{%
\interlinepenalty=10000
}
\makeatother

\usepackage[a4paper,lmargin=25.4mm,rmargin=25.4mm,tmargin=30.4mm,bmargin=20mm]{geometry}
\makeatletter
\renewcommand{\l@section}{\@dottedtocline{1}{1.5em}{2.6em}}
\renewcommand{\l@subsection}{\@dottedtocline{2}{4.0em}{3.6em}}
\renewcommand{\l@subsubsection}{\@dottedtocline{3}{7.4em}{4.5em}}
\renewcommand{\l@paragraph}{\@dottedtocline{3}{7.4em}{4.5em}}
\usepackage{theorem}
\newtheorem{theorem}{Theorem}
\newtheorem{lemma}{Lemma}
\newtheorem{definition}{Definition}
\newenvironment{proof}{\smallskip\noindent\textbf{Proof:}}{\hfill$\square$}
\usepackage{fancyhdr}
\makeatother

\usepackage{multicol}
\usepackage{xcolor,colortbl}
\usepackage{multirow}

\iflinenum
\usepackage[left]{lineno}
\linenumbers

\fi%

\newcounter{rgq}

\definecolor[named]{ACMBlue}{cmyk}{1,0.1,0,0.1}
\definecolor[named]{ACMYellow}{cmyk}{0,0.16,1,0}
\definecolor[named]{ACMOrange}{cmyk}{0,0.42,1,0.01}
\definecolor[named]{ACMRed}{cmyk}{0,0.90,0.86,0}
\definecolor[named]{ACMLightBlue}{cmyk}{0.49,0.01,0,0}
\definecolor[named]{ACMGreen}{cmyk}{0.20,0,1,0.19}
\definecolor[named]{ACMPurple}{cmyk}{0.55,1,0,0.15}
\definecolor[named]{ACMDarkBlue}{cmyk}{1,0.58,0,0.21}

\ifTODO
\newcommand{\TODO}[1]{{\color{blue}TODO: #1}}
\newcommand{\marginTODO}[1]{\marginpar{\color{blue}TODO: #1}}
\newcommand\todojp[1]{\textcolor{cyan}{TODO:jp: #1}}
\else
\newcommand{\TODO}[1]{}
\newcommand{\marginTODO}[1]{}
\newcommand\todojp[1]{}
\fi

\usepackage{listings}
\usepackage{wrapfig}
\usepackage{alltt}

\newcommand{\testpara}[1]{\paragraph{#1}: }
\newcommand{\testPara}[1]{\paragraph{#1}: }

\newcommand{\mysubsubsection}[1]{\paragraph*{\upshape\textbf{#1}:} }
\usepackage[T1]{fontenc}
\usepackage[scaled=0.79]{beramono}

\newcommand{\PKVMTEST}[1]{\texttt{#1}}
\usepackage{ifthen}
\usepackage{amssymb}
\usepackage{longtable}

\definecolor{IslaInitialState}{rgb}{0.92,1,1}
\definecolor{Forbid}{rgb}{1,0.92,0.92}
\definecolor{darkblue}{rgb}{0,0.0,0.5}

\lstdefinelanguage{none}{
  identifierstyle=
}

\lstdefinelanguage{cat}{
  basicstyle=\linespread{0.8}\ttfamily,
  keywordstyle=\bfseries,
  morekeywords={acyclic, let, irreflexive, as, empty},
  escapechar=@,
  morecomment=[f][\color{DarkGreen}]{//},
  morecomment=[n][\color{DarkGreen}]{(*}{*)},
  moredelim=**[is][\color{red}]{>}{<},
  moredelim=**[is][\color{gray}]{[*}{*]},
  showstringspaces=false,
  keepspaces=true,
  breaklines=true,
  mathescape
}

\lstdefinelanguage{smt2}{
  basicstyle=\linespread{0.8}\ttfamily,
  keywordstyle=\bfseries,
  escapechar=@,
  morecomment=[f][\color{DarkGreen}]{//},
  morecomment=[n][\color{DarkGreen}]{(*}{*)},
  moredelim=**[is][\color{red}]{>}{<},
  moredelim=**[is][\color{gray}]{[*}{*]},
  showstringspaces=false,
  keepspaces=true,
  breaklines=true,
  mathescape
}

\lstdefinelanguage{AArch64}{
  basicstyle=\ttfamily\small,
  keywordstyle=[1]\bfseries,
  morekeywords = [1]{str, dmb, dsb, tlbi, ldr, isb, mov, mrs, msr, eret, add, cbz, and, cmp, csel, eor, orr, svc, lsr, sub, cbnz,
                  STR, STLR, DMB, DSB, TLBI, LDR, ISB, MOV, MRS, MSR, ERET, ADD, CBZ, AND, CMP, CSEL, EOR, ORR, SVC, LSR, SUB, CBNZ, LDAR},
  keywordstyle=[2]{\color{darkblue}},
  morekeywords=[2]{X0, X1, X2, X3, X4, X5, X6, X7, X8, X9, X10, X11, X12, X13, X14, X15, X16, X17, X18, X19, X20, X21, X22, X23},
  morecomment = [l]{//},
  commentstyle=\sffamily\itshape\footnotesize,
}

\lstdefinelanguage{IslaLitmusName}{
  basicstyle=\rmfamily
}

\lstdefinelanguage{IslaLitmusExp}{
  keywordstyle=\bfseries,
  basicstyle=\ttfamily\footnotesize,
  morekeywords = {pte2, pte3, desc2, desc3, pa, extz, exts, bvlshr, bvshl, ttbr, mkdesc2, mkdesc3, page},
  alsodigit={0x,0b},
}

\lstdefinelanguage{IslaPageTableSetup}{
  morekeywords = {with, code, identity, s1table, s2table, physical, intermediate, virtual, option},
  keywordstyle=\bfseries,
  basicstyle=\ttfamily\footnotesize,
  literate= {|->}{$\mapsto$}1,
  mathescape=true
}

\usepackage{hyperref}
\hypersetup{
    colorlinks=true,
    linkcolor=blue,
    filecolor=magenta,
    urlcolor=cyan,
}

\newcommand{\asm}[1]{\texttt{#1}}

\newcommand{\cc}[1]{\lstinline{#1}}
\newcommand{\cat}[1]{\texttt{#1}}
\newcommand{\herd}[1]{\texttt{#1}}
\newcommand{\herdevent}[1]{\textbf{\texttt{#1}}}

\makeatletter
\renewcommand\paragraph{\@startsection{paragraph}{4}{\z@}%
                       {-12\p@ \@plus -4\p@ \@minus -4\p@}%
                       {-0.5em \@plus -0.22em \@minus -0.1em}%
                       {\normalfont\normalsize\bfseries}}
\renewcommand\subparagraph{\@startsection{subparagraph}{5}{\z@}%
                       {-12\p@ \@plus -4\p@ \@minus -4\p@}%
                       {-0.5em \@plus -0.22em \@minus -0.1em}%
                       {\normalfont\normalsize\itshape}}
\makeatother

\newcommand{\Mysubsection}[1]{\subsection{%SUBSECTION:
#1}}
\newcommand{\Mysubsubsection}[1]{\subsubsection{%SUBSUBSECTION:
#1}}
\newenvironment{TESTGROUP}[1]{%
#1-shaped tests%
}{}

\newcommand{\TESTPARAGRAPH}[1]{\paragraph{%TESTPARAGRAPH
Test: #1}{\label{test:#1}\global\def\currenttestname{#1}}}

\newcommand{\TEST}[1]{\hyperref[{test:#1}]{%TEST:
#1}}

\newcommand{\testlisting}[1]{%
\ifvmsaTestIndexStandalone
\IfFileExists{vmsa-test-index/tests/tex/#1.tex}%
{\input{vmsa-test-index/tests/tex/#1.tex}}%
{\TODO{file not found tests/tex/#1.tex}}%
\else
\IfFileExists{vmsa-test-index/tests/tex/#1.tex}%
{\input{vmsa-test-index/tests/tex/#1.tex}}%
{\TODO{file not found vmsa-test-index/tests/tex/#1.tex}}%
\fi
}

\newcommand{\testdiagram}[2]{%
\ifvmsaTestIndexStandalone
\IfFileExists{vmsa-test-index/model_results/diagrams/pdfs/#1/#1#2.pdf}%
{\includegraphics[width=0.99\textwidth]{{vmsa-test-index/model_results/diagrams/pdfs/#1/#1#2}.pdf}}%
{\TODO{vmsa-test-index/model\_results/diagrams/pdfs/#1/#1#2.pdf}}%
\else
\IfFileExists{vmsa-test-index/model_results/diagrams/pdfs/#1/#1#2.pdf}%
{\includegraphics[width=0.99\textwidth]{{vmsa-test-index/model_results/diagrams/pdfs/#1/#1#2}.pdf}}%
{\TODO{vmsa-test-index/model\_results/diagrams/pdfs/#1/#1#2.pdf}}%
\fi
}

\newcommand{\testresult}[2]{%
\ifvmsaTestIndexStandalone
\IfFileExists{vmsa-test-index/model_results/results_#1/#2.result}%
{\lstinputlisting{vmsa-test-index/model_results/results_#1/#2.result}}%
{no result for #2}
\else
\IfFileExists{vmsa-test-index/model_results/results_#1/#2.result}%
{\lstinputlisting{vmsa-test-index/model_results/results_#1/#2.result}}%
{no result for #2}
\fi
}

\newcommand{\testfigfor}[2]{%
\vspace*{2mm}
\hspace*{-20mm}\adjustbox{max width=\textwidth}{\input{vmsa-test-index/tests/tex/#1.tex}}

{\begin{center}\testdiagram{#1}{#2}\end{center}}

\begin{tabular}{l l}
Model & {\hfill Result \hfill} \\ \hline
Base & \testresult{strong}{#1} \\
ETS & \testresult{ETS}{#1} 
\end{tabular}
}

\newcommand{\testfignumbered}[1]{%
\testfigfor{\currenttestname}{_allow_#1}
}

\newcommand{\testfig}{%
\testfigfor{\currenttestname}{}
}

\newcommand{\testref}[1]{\ref{test:#1}}
\newcommand{\testrefs}[1]{\marginpar{#1}}

\newcommand{\TODOLATER}[1]{}

\newif\ifwiptests
\wiptestsfalse

\AtBeginDocument{%
  \providecommand\BibTeX{{%
    \normalfont B\kern-0.5em{\scshape i\kern-0.25em b}\kern-0.8em\TeX}}}

\begin{document}
\sloppy
\newcommand{\mytitle}{Relaxed virtual memory in Armv8-A (extended version)}
\title{\mytitle%
}

\author{
  Ben Simner$^{1}$ \ 
  Alasdair Armstrong$^{1}$ \ 
  Jean Pichon-Pharabod$^{2}$\\
  Christopher Pulte$^{1}$ \ 
  {Richard Grisenthwaite}$^{3}$ \ 
  Peter Sewell$^{1}$
  \\
  \ \\
{}$^{1}$  University of Cambridge, UK   \url{first.last@cl.cam.ac.uk}\\
{}$^{2}$  Aarhus University, Denmark   \url{jean.pichon@cs.au.dk}\\
{}$^{3}$    {Arm Ltd., UK}
\\
\ \\
\ifworkingdraft {\color{red}\ -- Draft --}\fi
}

\maketitle

\begin{abstract}
  \normalsize
  
  Virtual memory is an essential mechanism for enforcing security boundaries,
but its relaxed-memory concurrency semantics has not previously been investigated in detail.  The concurrent systems code managing virtual memory has been left on an entirely informal basis, and 
OS and hypervisor verification has had to make major simplifying assumptions.
  
\hspace*{3mm}  We explore the design space for relaxed virtual memory semantics in the \mbox{Armv8-A} architecture, to support future system-software verification. We identify many design questions, in discussion with Arm; develop a test suite, including use cases from the pKVM production hypervisor under development by Google; delimit the design space with axiomatic-style concurrency models; prove that under simple stable configurations our architectural model collapses to previous ``user'' models; develop tooling to compute allowed behaviours in the model integrated with the full Armv8-A ISA semantics; and develop a hardware test harness.
  
\hspace*{3mm}  This lays out some of the main issues in relaxed virtual memory  bringing these security-critical systems phenomena into the domain of programming-language semantics and verification with foundational architecture semantics.%

\hspace*{3mm}  This document is an extended version of a paper in ESOP 2022, with additional explanation and examples in the main body, and appendices detailing our litmus tests, models, proofs, and test results. 

  \end{abstract}

\newpage
\fancyhead[EOL]{\nouppercase{\leftmark}}
\tableofcontents
\newpage
\fancyhead[EOL]{\raisebox{-0.0\baselineskip}{\parbox{0.7\textwidth}{\nouppercase{\leftmark}\\\nouppercase{\rightmark}}}}

\section{Introduction}\label{intro}

Computing relies on virtual memory to enforce security boundaries:
hypervisors and operating systems manage mappings from virtual to
physical addresses to restrict access 
to physical memory and memory-mapped devices,
and thereby to ensure that processes and virtual machines cannot interfere with each other, or with the parent OS or hypervisor.
In
a world with endemic use of memory-unsafe languages for critical infrastructure, and of hardware that does not enforce fine-grained protection, virtual memory is one of the few mechanisms one has to enforce strong security guarantees.
This has driven interest in hypervisors and virtual machines, and it
provides a compelling motivation for verification of the OS-kernel and hypervisor code that manages virtual memory to provide security.

However, any such verification requires a semantics for the protection mechanisms provided by the underlying hardware architecture. %
There are two major challenges in establishing such a semantics.
First, there is its \emph{sequential intricacy}: virtual memory is one of the most complex aspects of a modern
general-purpose
architecture.
For 64-bit Armv8-A (AArch64) it is described in a 166-page chapter of the prose reference manual~\cite[Ch.D5]{G.a} and includes a host of features and options.
Second, and more fundamentally, there is its \emph{relaxed memory behaviour}.
Hardware implementations of virtual memory use in-memory representations of the virtual-to-physical address mappings, represented as hierarchical page tables.
For performance, there are dedicated cache structures for commonly used mapping data, in Translation Lookaside Buffers (TLBs).  
Translations are used often -- a single load instruction might need 40 or more page-table entries to translate its fetch and access addresses -- but they are changed only rarely, and by systems code not user code. Architectures therefore require manual management of TLB caching, e.g.~with specific instructions to invalidate old TLB entries that should no longer be used, instead of providing the simpler coherent memory abstraction that they do for normal accesses.
All this gives rise to new relaxed-memory effects, with subtle constraints determining when translations are required or forbidden to read from specific writes to the page tables, and systems code has to handle these appropriately to provide the desired virtual-memory abstraction and its security properties. 

Previous work has developed hand-written sequential semantics for some aspects of address translation in Arm~\cite{%
DBLP:conf/lpar/SyedaK17,DBLP:conf/itp/SyedaK18,DBLP:phd/basesearch/Syeda19,DBLP:journals/jar/SyedaK20,%
Li2021,%
DBLP:journals/jcs/GuancialeNDB16,%
DBLP:phd/basesearch/Kolanski11%
}
and x86~%
\cite{%
GoelPhD,DBLP:books/sp/17/GoelHK17,%
DBLP:conf/birthday/DegenbaevPS09,%
DBLP:journals/jar/TewsVW09%
}, 
but these are at best lightly validated formalisations, and 
there is no well-validated relaxed-memory concurrency semantics of virtual memory. 
In the absence of that (and of proof techniques above it), previous OS and hypervisor verification work,
e.g.~on seL4, CertiKOS, KCore, Hyper-V, the PROSPER hypervisor, and SeKVM~\cite{seL4-the-proof,Klein_AEMSKH_14,DBLP:conf/osdi/GuSCWKSC16,Li2021,DBLP:conf/vstte/AlkassarCKP12,DBLP:journals/jcs/GuancialeNDB16,DBLP:conf/uss/LiLGNH21,Gu2021}
has had to make major simplifying assumptions,
either assuming correctness of TLB management and a single-threaded setting (seL4), or assuming sequentially consistent concurrency with one of those hand-written sequential semantics, or assuming an extended notion of data-race-freedom %
(we return to the related work in \S\ref{related}).
We explore the design space for Armv8-A relaxed virtual memory semantics,
to support future systems-software verification.  We contribute:
\begin{itemize}
\item A description of the current Arm architectural intent as we understand it,
      and a set of design questions and issues arising from its relaxed virtual memory semantics~(\S\ref{questions}).
\item A relaxed virtual memory test suite, comprising of a set of hand-written litmus tests
      which illustrate the aforementioned design questions and capture key use cases from
      pKVM, a production hypervisor under development by Google~(\S\ref{sec:tests}).  
\item An axiomatic-style concurrency model for relaxed virtual memory in Armv8~(\S\ref{sec:models}),
      which to the best of our knowledge and ability captures the architectural intent described
      in \S\ref{questions}.
      We also define a weaker model, motivated by the properties pKVM relies on. %

\item  We prove that,
for stable injective page-tables, the first model collapses to the previous Armv8-A user-mode concurrency model~(\S\ref{proof}).

\item  We extend our Isla tool~\cite{isla-cav}, enabling it to compute the allowed behaviours of virtual memory litmus tests with respect to arbitrary axiomatic models,
       using the authoritative Arm ASL definition of the intra-instruction semantics including pagetable walks~(\S\ref{isla}).

\item  We develop a test harness that lets us run virtual-memory litmus tests bare-metal, albeit currently only for Stage 1 tests, and report results from running these on hardware~(\S\ref{harness}).
\end{itemize}

We begin in \S\ref{background} with an informal introduction to virtual memory in a simple sequential setting, to make this as self-contained as possible, but familiarity with virtual memory from a systems perspective, and with previous work on user-space relaxed memory, will be helpful. %

Mainstream industrial architecture specifications evolve over many years, balancing
hardware-implementation and systems-software concerns.
Experience with ``user'' relaxed-memory concurrency has shown that the process of developing rigorous semantics for arbitrary code provides a useful third input into this process, 
leading one to ask questions which help clarify the architectural intent.
The architects, hardware designers, and system-software authors typically have a deep understanding of the area, but there is usually not, \emph{a priori}, a well-understood informal specification that just needs to be formalised; instead that needs to be iteratively and collaboratively developed.
{Our} \S\ref{questions} {is based on
detailed discussion with the Arm Chief Architect
(a co-author of this paper);}
 the current Arm prose documentation~\cite{G.a};
 discussion with the pKVM development team;
and our experimental testing. 
{To the best of our knowledge, our models provide a reasonable basis for software development and for verification,  but this paper is
surely not the last word on the subject,
and it does not give an authoritative definition of the Armv8-A architecture.}
The history of relaxed-memory models shows that it typically takes multiple years, and gradual refinement of models, to converge on something reasonably stable for a production architecture or language, and even then they continue to change as new knowledge or features arise; with hindsight, few are definitive. 
Our goal here is rather to lay out some of the main issues,
bringing this security-critical systems code into the domain of programming-language semantics and verification, above foundational architecture semantics.
This document is an extended version of a paper in ESOP 2022~\cite{esop2022}, with additional explanation and examples in the main body, and appendices detailing our
litmus tests (\ref{app:vmsa}), %
models (\ref{app:models}),
proofs (\ref{app:modelsrelation}), and test results (\ref{app:results}). 
Further details are at \url{https://www.cl.cam.ac.uk/users/pes20/RelaxedVM-Arm/}.

\subsubsection*{Scope and non-goals}
\label{sec:scope}
Our scope is Armv8-A virtual memory for the 64-bit (AArch64) architecture,
aiming especially to support aspects relevant to hypervisors such as pKVM.
Accordingly, we consider translation with multiple stages (for both hypervisor and OS), multiple levels, and the full Armv8-A intra-instruction semantics and translation walk behaviour (as defined by Arm in ASL and auto-translated to Sail~\cite{sail-popl2019}).
Our models cover the Armv8-A ETS option as work in progress.
We discuss some mixed-size aspects, but our models do not currently cover them. 
To keep things manageable,
we do not consider hardware management of access flags or dirty bits, conflict aborts, \textsc{feat\_bbm, feat\_cnp, feat\_xs}, the interactions between virtual memory and instruction-fetch, or all the relaxed behaviour of exceptions, and we handle only some of the many varieties of the TLBI instruction.
We focus on the specification of the architecturally allowed envelope of functional behaviour, not on side-channel phenomena.
We include some experimental testing, as a sanity check of our models, but our principal goal is to capture the architectural intent, and our principal validation is from discussion with Arm.  %
Many of the issues should also be relevant to other architectures, but here we address only Armv8-A. 

\section{Background: A crash course on virtual memory}\label{background}

\subsection{Virtualising addressing}
In conventional computer systems, the underlying memory is indexed by \emph{physical addresses} (PAs), as are memory-mapped devices. 
For a small microcontroller running trusted code, accessing resources directly via physical addresses may suffice. 
Larger systems rely heavily on virtual addressing: they interpose one or more layers of indirection between \emph{virtual addresses} (VAs) used by instructions and the underlying physical addresses.   This lets them:
\begin{enumerate}
\item partition resources among different programs, giving each access only to those it needs;
\item provide convenient numeric ranges of virtual addresses to each program; and
\item dynamically extend and change the mapping from virtual to physical addresses, e.g.~to support copy-on-write or %
swapping, or shared buffers.
\end{enumerate}
A simple system might have many processes managed by an operating system, each of which (including the OS) has a partial function that gives the physical address and permissions for the virtual addresses it can use, roughly:
\[\begin{array}{l}
\texttt{translate} : \texttt{VirtualAddress} \rightharpoonup \texttt{PhysicalAddress} \times 2^{\{\texttt{Read}, \texttt{Write}, \texttt{Execute}\}}
\end{array}\]
Typically each process would have access to a subset of the physical addresses (the range of its translate function), disjoint from those of the other processes and from that of the OS, while the OS would have sole access to its own working memory and also access to that of the processes.
This is implemented with a combination of hardware and system software. 
The hardware memory management unit (MMU) automatically translates virtual
to physical addresses when doing an access needed to execute an instruction.  If the function is undefined, the instruction traps with a page fault; if it is defined but does not have the appropriate accesses, it traps with a permission fault; and if it is defined with the right permissions, the hardware performs the required access using the resulting physical address. 
The OS has to set up the translate functions, ensure that the appropriate function is used when switching to a new process, and handle those faults.
In general translation functions are not necessarily injective,
and includes not just access permissions (which can moreover vary between exception levels),
but also additional fields for
cacheability, shareability, security, contiguity, and other aspects which we elide for simplicity here.

\subsection{The translation-table walk}
The current translate function for execution is determined by a system register, a \emph{translation table base register} or \texttt{TTBR}, that contains the physical address of a lookup-tree data structure in memory.
The details of this structure are (in Armv8-A) highly configurable, e.g.~for different page sizes, controlled by various system registers.  
In a common configuration used by Linux, it maps 4096-byte pages and has a tree up to four levels (0--3) deep.
We assume this configuration for the remainder of this section.

Each node in the tree is a 4096-byte block of memory
made up of 512 64-bit entries (called ``descriptors'' by Arm).

These descriptors are of various types, either: \emph{invalid}, indicating that this part of the domain is unmapped; 
a \emph{block} or \emph{page} descriptor, defining a fixed-size mapping to a range of output addresses;
or a \emph{table} descriptor which points
to another level of table for this part of the domain.

The least significant two bits of the descriptor define what type the descriptor is,
and the other bits are partitioned into various fields depending on the type:
\begin{itemize}
  \item Output address (OA): the page the final output (IPA or PA) address is in.
  \item Table pointer: a 4k-aligned pointer to the next-level translation table.
  \item Attrs: encoding of the access permissions, memory attributes, shareability, access bits and dirty flags.
\end{itemize}

\mysubsubsection{Invalid descriptors}

\begin{center}
\begin{tikzpicture}[x=-\linewidth/74,y=2.5ex,every path/.style={draw=black,semithick}]
  \node at (63-.5,1.5) {63};
  \node at (33-.5,1.5) {\ldots};
  \node at (1+.5,1.5) {1};
  \node at (0+.5,1.5) {0};

  \draw (1,0) rectangle (63,1);
  \node at (33,0.5) {ignored};

  \draw (0,0) rectangle (1,1);
  \node at (0+.5,0.5) {0};
\end{tikzpicture}
\end{center}

\mysubsubsection{Block or page descriptors}

\begin{center}
\begin{tikzpicture}[x=-\linewidth/74,y=2.5ex,every path/.style={draw=black,semithick}]
  \node at (63-.5,1.5) {63};
  \node at (50+1,1.5) {50};
  \node at (47,1.5) {47};
  \node at (24+.5,1.5) {n};
  \node at (22-.5,1.5) {(n-1)};
  \node at (13,1.5) {12};
  \node at (11,1.5) {11};
  \node at (2+.5,1.5) {2};
  \node at (1+.5,1.5) {1};
  \node at (0+.5,1.5) {0};

  \draw (0,0) rectangle (1,1);
  \node at (0+.5,0.5) {1};

  \draw (1,0) rectangle (2,1);
  \node at (1+.5,0.5) {x};

  \draw (2,0) rectangle (12,1);
  \node at (6+.5,0.5) {attrs};

  \draw (12,0) rectangle (24,1);
  \node at (17+.5,0.5) {ignored};

  \draw (24,0) rectangle (48,1);
  \node at (36+.5,0.5) {output address};

  \draw (48,0) rectangle (49,1);
  \draw (49,0) rectangle (50,1);
  \node at (48+.5,0.5) {0};
  \node at (49+.5,0.5) {0};

  \draw (50,0) rectangle (64,1);
  \node at (56+.5,0.5) {attrs};
\end{tikzpicture}
\end{center}

Here \emph{n} depends on how deep in the table this entry is:
for a level~1 block descriptor $n==30$,
for a level~2 it is 21,
and for level~3 it is 12.
Note that bit 1 should be set when at level~3 (a page descriptor),
otherwise it is 0 for block descriptors.

\mysubsubsection{Table descriptors}

\begin{center}
  \begin{tikzpicture}[x=-\linewidth/74,y=2.5ex,every path/.style={draw=black,semithick}]
    \node at (63-.5,1.5) {63};
    \node at (50+1,1.5) {50};
    \node at (47,1.5) {47};
    \node at (13,1.5) {12};
    \node at (11,1.5) {11};
    \node at (2+.5,1.5) {2};
    \node at (1+.5,1.5) {1};
    \node at (0+.5,1.5) {0};

    \draw (0,0) rectangle (1,1);
    \node at (0+.5,0.5) {1};

    \draw (1,0) rectangle (2,1);
    \node at (1+.5,0.5) {1};

    \draw (2,0) rectangle (12,1);
    \node at (6+.5,0.5) {attrs};

    \draw (12,0) rectangle (48,1);
    \node at (30+.5,0.5) {table pointer};

    \draw (48,0) rectangle (49,1);
    \draw (49,0) rectangle (50,1);
    \node at (48+.5,0.5) {0};
    \node at (49+.5,0.5) {0};

    \draw (50,0) rectangle (64,1);
    \node at (56+.5,0.5) {attrs};
  \end{tikzpicture}
  \end{center}
Table descriptors are allowed only at levels 0--2.

\mysubsubsection{Arm's translation-table walk}

The sequential behaviour of Arm's \emph{translation-table walk} function is fully defined in the Arm ASL language.

The hardware walker first splits up the input virtual address into chunks:
the upper 16 bits are typically ignored;
fields \emph{a-d} are used for indexing into the tables; and
field \emph{e} is added to the final result to get the physical address.

\begin{center}
  \begin{tikzpicture}[x=-\linewidth/74,y=2.5ex,every path/.style={draw=black,semithick}]
    \draw [decorate,decoration={brace,amplitude=10pt}]
      (63+.5,2.5) -- (0-.5,2.5) node [black,midway,yshift=0.5cm,xshift=-0.1cm] {VA};
  
    \node at (63-.5,1.5) {63};
    \node at (48+1,1.5) {48};
  
    \node at (47-.5,1.5) {47};
    \node at (39+1,1.5) {39};
  
    \node at (38-.5,1.5) {38};
    \node at (30+1,1.5) {30};
  
    \node at (29-.5,1.5) {29};
    \node at (21+1,1.5) {21};
  
    \node at (20-.5,1.5) {20};
    \node at (12+1,1.5) {12};
  
    \node at (11-.5,1.5) {11};
    \node at (0+.5,1.5) {0};

    \draw (0,0) rectangle (12,1);
    \node at (5+.5,0.5) {e};

    \draw (12,0) rectangle (21,1);
    \node at (16+.5,0.5) {d};

    \draw (21,0) rectangle (30,1);
    \node at (25+.5,0.5) {c};

    \draw (30,0) rectangle (39,1);
    \node at (34+.5,0.5) {b};

    \draw (39,0) rectangle (48,1);
    \node at (43+.5,0.5) {a};

    \draw (48,0) rectangle (63,1);
    \node at (55,0.5) {ignored};
  \end{tikzpicture}
\end{center}

A pointer to the initial level of translation is obtained by reading the relevant translation table base register (\asm{TTBR}).
The fields \asm{a}-\asm{d} are then used to indirect into each table in turn,
until a block (or page) mapping is found.

Each level of the tree maps a different size block.
e.g. in a common configuration
each level~1 entry maps a 1GiB region,
each level~2 a 2MiB region,
and each level~3 a single 4KiB page.

\begin{center}  
  \begin{tikzpicture}[x=\linewidth/74,y=2.5ex,every path/.style={draw=black,semithick}]
      \def\boxstartoffsx{12}
      \def\boxstartoffsy{1}
      \def\boxoffsx{15}
      \def\boxoffsy{2}
      \def\boxwidth{9}
      \def\boxheight{6}
  
      \newcommand{\tablebox}[1]{%
      (\boxstartoffsx+\boxoffsx*#1, \boxstartoffsy+\boxoffsy*#1) rectangle (\boxstartoffsx+\boxoffsx*#1+\boxwidth, \boxstartoffsy+\boxoffsy*#1+\boxheight) node[xshift=-22,yshift=5] {Level #1}%
      }
  
      \newcommand{\tableentry}[3]{%
      (\boxstartoffsx+\boxoffsx*#1, \boxstartoffsy+\boxoffsy*#1+1.5+#2) rectangle (\boxstartoffsx+\boxoffsx*#1+\boxwidth, \boxstartoffsy+\boxoffsy*#1+2.5+#2)  node[xshift=-25,yshift=-5] {#3}%
      }
      \newcommand{\tableentryarrow}[4]{%
      (\boxstartoffsx+\boxoffsx*#1+\boxwidth, \boxstartoffsy+\boxoffsy*#1+\boxoffsy+#2) -- (\boxstartoffsx+\boxoffsx*#1+\boxoffsx+#3, \boxstartoffsy+\boxoffsy*#1+\boxoffsy+#2) node[yshift=5] {#4}%
      }
      \newcommand{\tableentryoffsetarrow}[2]{%
      (\boxstartoffsx+\boxoffsx*#1-1, \boxstartoffsy+\boxoffsy*#1) -- (\boxstartoffsx+\boxoffsx*#1-1, \boxstartoffsy+\boxoffsy*#1+1.5) node[yshift=-7.5,xshift=-5] {#2}%
      }
  
      \draw (0, 0) rectangle (8,2);
      \node at (4, 1) {\textsc{ttbr}};
  
      \foreach \i in {0,1,2,3}
          \draw \tablebox{\i};
  
      \draw[->] (8,1) -- (12,1);
  
      \draw \tableentry{0}{0}{table};
      \draw \tableentry{1}{0}{table};
      \draw \tableentry{2}{0}{table};
      \draw \tableentry{3}{0}{\emph{y}: page};
  
      \draw[->] \tableentryarrow{0}{0}{0}{};
      \draw[->] \tableentryarrow{1}{0}{0}{};
      \draw[->] \tableentryarrow{2}{0}{0}{};
      \draw[->] \tableentryarrow{3}{0}{-4}{\tiny 4KiB};
  
      \draw \tableentry{1}{2}{\emph{x}: block};
      \draw[->] \tableentryarrow{1}{2}{-4}{\tiny 1GiB};
  
      \draw[<->] \tableentryoffsetarrow{0}{a};
      \draw[<->] \tableentryoffsetarrow{1}{b};
      \draw[<->] \tableentryoffsetarrow{2}{c};
      \draw[<->] \tableentryoffsetarrow{3}{d};
  \end{tikzpicture}
\end{center}

The final physical address is the output address field of the page or block mapping,
with the remaining bits of the VA appended.

For example, if the VA was translated using the 4KiB page entry \emph{y} from above:
\begin{center}
  \begin{tikzpicture}[x=-\linewidth/74,y=2.5ex,every path/.style={draw=black,semithick}]
    \draw [decorate,decoration={brace,amplitude=10pt}]
      (63+.5,2.5) -- (0-.5,2.5) node [black,midway,yshift=0.5cm,xshift=-0.1cm] {PA from \emph{y}};
  
    \node at (63-.5,1.5) {63};
    \node at (48+1,1.5) {48};

    \node at (11-.5,1.5) {11};
    \node at (0+.5,1.5) {0};

    \draw (0,0) rectangle (12,1);
    \node at (5+.5,0.5) {e};

    \draw (12,0) rectangle (48,1);
    \node at (29,0.5) {OA from page descriptor};

    \draw (48,0) rectangle (63,1);
    \node at (55,0.5) {ignored};
  \end{tikzpicture}
\end{center}

\vspace*{0.5cm}
\noindent
Or for the 1GiB level 1 block entry \emph{x}:
\vspace*{-0.5cm}
\begin{center}
  \begin{tikzpicture}[x=-\linewidth/74,y=2.5ex,every path/.style={draw=black,semithick}]
  \draw [decorate,decoration={brace,amplitude=10pt}]
  (63+.5,2.5) -- (0-.5,2.5) node [black,midway,yshift=0.5cm,xshift=-0.1cm] {PA from \emph{x}};

  \node at (63-.5,1.5) {63};
  \node at (48+1,1.5) {48};

  \node at (47-.5,1.5) {47};
  \node at (30+1,1.5) {30};

  \node at (29-.5,1.5) {29};
  \node at (21+1,1.5) {21};

  \node at (20-.5,1.5) {20};
  \node at (12+1,1.5) {12};

  \node at (11-.5,1.5) {11};
  \node at (0+.5,1.5) {0};

  \draw (0,0) rectangle (12,1);
  \node at (5+.5,0.5) {e};

  \draw (12,0) rectangle (21,1);
  \node at (16+.5,0.5) {d};

  \draw (21,0) rectangle (30,1);
  \node at (25+.5,0.5) {c};

  \draw (30,0) rectangle (48,1);
  \node at (39,0.5) {OA from block};

  \draw (48,0) rectangle (63,1);
  \node at (55,0.5) {ignored};
 \end{tikzpicture}
\end{center}

Note that,
as mentioned before,
the architecture is highly configurable
and the above diagrams give just a common configuration.
Various system registers allows the user to configure:
the size of the input addresses (e.g. 48 bit, 52 bit);
the number of translation table levels (e.g. 3, 4, 5);
the size of a page (e.g. 4K, 16K, 64K),
and much more that we elide here for brevity.

\subsection{Multiple stages of translation}
The above suffices for an operating system isolating multiple processes from each other, but one often wants to isolate multiple operating systems (or other guests), managed by a hypervisor.
To support this, the architecture provides a second layer of indirection: instead of going straight from virtual to physical addresses, with a single \emph{stage} of mapping controlled by the OS,
one can have two stages, with the OS managing a Stage~1 table which maps virtual addresses to an \emph{intermediate physical addresses} (IPAs),
composed with a hypervisor-managed Stage~2 table, mapping IPAs to PAs.
The full translation composes the two, intersecting %
their permissions.
\[ \begin{array}{l}
\texttt{translate\_stage1} : \texttt{VirtualAddress} \rightharpoonup \texttt{IPA} \times 2^{\{\texttt{Read}, \texttt{Write}, \texttt{Execute}\}}\\
\texttt{translate\_stage2} : \texttt{IPA} \rightharpoonup \texttt{PhysicalAddress} \times 2^{\{\texttt{Read}, \texttt{Write}, \texttt{Execute}\}}
\end{array}
\]
Armv8-A has various \emph{exception levels} (ELs), determined by the \asm{PSTATE.CurrentEL} register, 
including EL0 (for user processes), EL1 (for OSs or other
guests), and EL2 (for a hypervisor).
These each have associated translation-table base registers:
\begin{itemize}
  \item \asm{TTBR0\_EL1}: contains a pointer (IPA) to the Stage~1 table for EL1\&0, lower VA range (process addresses), producing IPAs, controlled by OS at EL1
  \item \asm{TTBR1\_EL1}: contains a pointer (IPA) to the Stage~1 table for EL1\&0, upper VA range (OS kernel addresses), producing IPAs, controlled by OS at EL1
  \item \asm{VTTBR\_EL2}: contains a pointer (PA) to the Stage~2 table (second stage for IPAs translated at EL1\&0), producing PAs, controlled by hypervisor at EL2
  \item \asm{TTBR0\_EL2}: contains a pointer (PA) to the single-stage table for EL2 (hypervisor's own addresses), producing PAs, controlled by hypervisor at EL2
\end{itemize}
Each hardware thread has its own base registers (and other system registers), and so different hardware threads can be using different address spaces (for example, for different processes)
at the same time.

\subsection{Caching translations in TLBs}
A naive hardware implementation of address translation would need many translation memory reads -- with four levels, 
up to 24 with both stages enabled, for every instruction-fetch, read, or write.
This would have unacceptable performance, so processors have
specialised caches for translation-table wal k reads called
\emph{translation lookaside buffers} (or TLBs).
Under normal operation the TLBs are invisible to user code, %
but systems code has to manage them explicitly,
to change which translation table is currently in use (e.g.~when context switching),
or to make changes to the tables for one process or guest.
Without correct management a TLB could hold incorrect (stale) data,
breaking the protection that the address translation is intended to provide.

The architecture supports explicit TLB maintenance with various
flavours of the \asm{TLBI} instruction
(TLB invalidate),
to invalidate old entries for specific ranges of virtual or intermediate-physical addresses,
or even whole ASIDs or VMIDs at once.
The \emph{memory management unit} (MMU) is responsible for performing these translations.
It does this by looking at the TLB and, if the TLB does not contain an entry for the given address
(called a \emph{miss}),
it performs the translation table walk function as described earlier and caches the result in the TLB (a \emph{fill}).

TLB maintenance and TLB misses are expensive,
and one would not want the cost of TLB invalidation on every context switch,
so the architecture provides
\emph{address space identifiers} (ASIDs). 
The translation table base registers include an ASID in addition to the table base address, 
and when translation data is cached in a TLB it is tagged with the current ASID,
giving the illusion of separate TLBs per ASID, and
allowing switching from one to another without TLB maintenance.
Eventually the system will need to reclaim and reuse a previously used ASID,
and then TLB maintenance is required to clean that ASID's old entries.
There are similar identifiers for Stage~2 intermediate physical memory,
known as virtual-machine identifiers or VMIDs.

\section{Concurrency architecture design questions}\label{questions}

Now we will
introduce the main concurrency architecture design questions that arise for Armv8-A virtual memory, within the scope laid out in the introduction.
As usual, the architecture has to define an envelope of behaviour that provides the guarantees needed by software, while admitting the relaxed behaviour of the microarchitectural techniques necessary for performance.
That means we have to discuss both, including just enough microarchitecture to understand the possible programmer-visible behaviour, before we abstract it in the semantic models we give in \S\ref{sec:models}.
The discussion includes points of several kinds:
some that are clear in the current Arm documentation,
some where Arm have a change in flight,
some that are not documented but where the semantics is (after discussion) obviously constrained by existing hardware or software practice,
and some where there is a tentative Arm intent but it is not yet fixed upon;
our modelling raised a number of questions of the latter two.
To make this as coherent as possible, we discuss all these in a logical order, laying out the design principles. %
We have developed a suite comprised of \numberOfIslaLitmusTests{} hand-written Isla-compatible virtual-memory litmus tests that illustrate the issues,
but to keep this concise we just give the main ideas here.  For ease of reference, we give the actual tests in App.~\ref{app:vmsa}, with links in the margin. 
As a sample, we explain one pKVM test in detail in \S\ref{sec:pkvm}. %

\subsection{Coherence with respect to physical or virtual addresses}\label{subsec:questions_coherence}

For normal memory accesses, the most fundamental guarantee that architectures provide is \emph{coherence}:  in any execution, for each memory location, there is a total order of the accesses to that location, consistent with the program order of each thread, with reads reading from the most recent write in that order. 
Hardware implementations provide this, despite their elaborate cache hierarchies and out-of-order pipelines, by %
coherent cache protocols and pipeline hazard checking, identifying and restarting instructions when possible coherence violations are detected. 
Previous work on relaxed-memory semantics for architectures has
taken virtual addresses as primitive, implicitly considering only execution with well-formed, constant, and injective address translation mappings. 

Now, we have to consider whether coherence is with respect to virtual or physical addresses, for non-injective mappings.
For Arm, coherence is w.r.t.~physical addresses~\cite[D5.11.1 (p2812)]{G.a}. \testrefs{\testref{CoRR0.alias+po}\\\testref{CoRR2.alias+po}}%
This means %
that if two virtual addresses alias to the same physical address, then (still assuming well-formed and constant translation):
a load from one virtual address cannot ignore a program-order (\herd{po}) previous store to the other; and \testrefs{\testref{CoWR.alias}}%
a load from one virtual address can have its value forwarded from a store to the other,   and similarly on a speculative branch. \testrefs{\testref{MP.alias3+rfi-data+dmb}\\\testref{PPOCA.alias}}%

\subsection{Relaxed behaviour from TLB caching}\label{subsec:questions:tlb}

There are two main aspects of the concurrency semantics of virtual memory:
the relaxed behaviour arising directly from TLB caching,
and the relaxed behaviour of the \emph{not-from-TLB} (\emph{non-TLB}) memory accesses for translation reads that read from memory or by forwarding from \herd{po}-previous writes, and that might supply TLB cache fills.  We discuss them in this and the following subsection respectively.

\mysubsubsection{What can be cached}\label{subsubsec:questions:tlb:what_cached}
The MMU can cache information from successful translations, and also from translations that result in permission faults, but it is architecturally forbidden from caching information from attempted translations that result in translation faults.%
This ensures that the handlers of those faults do not need to do TLB maintenance to remove the faulting entry~\cite[D5.8.1 (p2780)]{G.a}, and makes the potential behaviour for page-table updates from invalid-to-valid and
valid-to-any
 quite different, as we shall see.

TLB implementations might cache any combination of individual page-table entries and partial or complete translations, e.g.~from the virtual address and context to the physical address of the last-level page.  Conceptually, however, we can simply view a TLB as containing a set of cached page-table-entry writes (i.e., writes that have been read from for a translation), including at least:
\begin{itemize}
\item the context information of the translation: the VMID, ASID
 (or a ``global indicator''),
,     and the originating exception level;
\item the virtual address, intermediate physical address, and/or physical address of the translation;
\item the translation stage and level at which the write was used;
\item the system register values used in the translation (those which can be cached); and
\item for an entry used for a Stage 1 translation, whether it has been invalidated at both stages.
\end{itemize}
That additional information allows the various TLBI instructions to target specific entries.\testrefs{\testref{CoWinvT+dsb-isb}\\\testref{CoWinvT.EL1+dsb-tlbiis-dsb-isb}\\\testref{ROT.inv+dsb}}
A translation walk can arbitrarily use either a cached write (if one exists) or do a non-TLB read, either from memory or by forwarding from a po-previous write, for any stage or level.

\mysubsubsection{Caching of multiple entries for the same virtual address and context}\label{subsubsec:questions:tlb:same_context}

High-performance hardware implementations may have elaborate TLB structures, including multiple ``micro TLBs'' per thread.
These can be seen as a conceptual single per-thread TLB that can hold zero, one, or more entries for each combination of input address and the other information above.
If zero, a translation will necessarily read from memory (with ordering constrained as discussed below).
If one or more, a translation may use any of those entries or read from memory (and the write read from might or might not be cached). \testrefs{\testref{CoTfT+dsb-isb}}
However, in some cases multiple entries constitute a \emph{break-before-make} failure, leading to relatively unconstrained behaviour; we return to this below.

\mysubsubsection{When can page-table entries be cached}\label{subsubsec:questions:tlb:when_cached}

Any %
memory read by a translation can be cached.
Any
thread can spontaneously do a translation for any virtual address at
any program point, with respect to its context at that point
(though this interacts with the %
system-register write/read semantics). Spontaneous translations model
hardware
prefetching, speculative execution, and branch
prediction.
They mean that, in the absence of cache maintenance,
translations may use TLB entries from arbitrarily old writes. \testrefs{\testref{CoWinvT+dsb-isb}}
Additionally, any thread may do a spontaneous translation at any point using
the configuration from
any exception level higher than the current one, but not for lower levels.
Preventing spontaneous walks at lower EL is essential, as during an EL2 hypervisor switch between VMs, the EL1 control registers will be in an inconsistent state.
Allowing spontaneous walks at higher EL models arbitrary interrupts to the higher level and then doing a spontaneous walk there. %

Each virtual-memory access by a thread involves a non-spontaneous translation which
is constrained by the normal inter-instruction constraints on
out-of-order and speculative execution by the thread.
These
constraints are especially important in order to understand when a translation
must fault: as invalid entries cannot be cached, a translation that gives rise to such a fault must be at least in part from a non-TLB read, subject to these ordering constraints. \testrefs{\testref{CoTfRpte+dsb-isb}}

\mysubsubsection{Coherence of translations}\label{subsubsec:questions:tlb:coherence}

Due to the TLB caching as described above
translations of the same virtual address by the same thread need not see a coherent view of page-table memory.
This is in sharp contrast to normal accesses, but analogous to instruction-fetch reads~\cite{iflat-esop2020-extended}
and reads from persistent memory~\cite{persistence-tso}. \testrefs{\testref{CoTfT+dsb-isb}}

\mysubsubsection{Removing cached entries}\label{subsubsec:questions:tlb:dropping}

TLBs may spontaneously forget any cached information at any point. 
To \emph{ensure} that a cached entry is removed,
software must ensure that it will not be spontaneously re-cached.
It can do this with a write of an invalid entry and then a DSB instruction (data synchronization barrier) to ensure that it is visible across the system,
followed by a TLBI.%
\testrefs{\testref{CoWinvT+dsb-isb}\\\testref{CoRpteT.EL1+dsb-tlbi-dsb-isb}\\\testref{CoWinvT.EL1+dsb-tlbiis-dsb}}

\mysubsubsection{Break-before-make failures}\label{subsubsec:questions:tlb:bbm}

When changing an existing translation mapping, from one valid entry to another valid entry, Arm require in many cases the use of a \emph{break-before-make (BBM)} sequence:  breaking the old mapping with a write of an invalid entry;
a DSB to ensure that is visible across the system;
and a broadcast TLBI %
to invalidate any cached entries for all relevant threads;
a DSB to wait for the TLBI to finish;
then making the new mapping with a write of the new entry, and additional synchronisation  to ensure that it is visible to translations (specifically, to translation-walk non-TLB reads).
\testrefs{\testref{BBM+dsb-tlbiis-dsb}\\\testref{BBM.Tf+dsb-tlbiis-dsb}\\\testref{MP.BBM1+dsb-tlbiis-dsb-dsb+dsb-isb}}%
The current Arm text~\cite[D5.10.1 (p2795)]{G.a} identifies six cases of page-table updates that without such a sequence constitute \emph{BBM failures}, and gives very severe architectural consequences:  failures of coherency, single-copy atomicity, ordering, or uniprocessor semantics.
Note that these consequences are architecturally allowed if there could exist a break-before-make-failure change to the translation tables for some virtual address, irrespective of whether the program architecturally accesses it.

This severity is because, in some of the six cases, hardware implementations could give rather arbitrary behaviour, e.g.~an amalgamation of old and new entries.
 From a software point of view, it seems that one must treat such cases more-or-less as fatal errors. This is analogous to the Data-race-free-or-catch-fire semantics underlying the C/C++
relaxed memory model~\cite{Adve:1990:WON:325164.325100,DBLP:journals/jpdc/GharachorlooAGHH92,bohemadvec,cpp-popl}, in which any program with a consistent execution that includes a race between nonatomic accesses is deemed to have undefined behaviour, and the C/C++ standards do not constrain implementation behaviour for such programs in any way.
This makes many potential litmus tests that change between valid entries uninteresting, as they simply exhibit BBM failures (though changes of permissions do not necessitate a BBM sequence).

However, for a processor architecture that supports virtualisation, one cannot regard BBM failures as allowing completely arbitrary behaviour for the entire machine: if one guest virtual machine (at EL1) changes one of its own translation mappings without correctly following the BBM sequence, either mistakenly or maliciously, that should not impact security of the hypervisor (at EL2) or other guests.  Instead, one has to bound the arbitrary behaviour to that virtual machine, allowing arbitrary memory and register accesses that are possible within its context.
In our exhaustively executable semantics, to keep litmus-test executions finite, we currently simply detect BBM failures; we do not explicitly model that arbitrary behaviour.

In reality, these six BBM failure cases include some where hardware may give such weakly constrained behaviour and others where, because coherence is over physical addresses and the mapping may be temporarily indeterminate, software might see well-defined but nondeterministic or surprising results. These were architected as a guide for system software to produce predictable behaviour, and future versions of the architecture might refine this.

When a hypervisor installs a new guest, it has to be able to reset to a clean state, in case a previous guest has failed to follow the BBM sequence (which in general the hypervisor cannot know).  
It can do so with a TLBI %
covering all the previous guest's processes address space.
There seems to be no need or support for finer-grain cleanup.

\subsection{Relaxed behaviour of translation-walk non-TLB reads}\label{subsec:questions:non_tlb}

Now we turn to the semantics of translation-walk \emph{non-TLB} reads, those that are satisfied from memory or by forwarding, not from a TLB, and which might then be cached in a TLB.
This matters especially when one knows that there are no relevant cached TLB entries,
e.g.~when an invalid entry has been written and a TLBI performed,
so one knows that translation walks will do such a non-TLB read.
\testrefs{\testref{CoRpteTf.inv+dsb-isb}\\\testref{CoTTf.inv+dsb-isb}}%

\mysubsubsection{Ordering among the translation-walk reads of an access}
Each translation-table walk for a virtual-memory access can involve many memory reads,
one for each level of the table for each stage of translation.

\begin{wrapfigure}{r}{35mm}
  \vspace*{-5mm}
  \includegraphics[width=34mm]{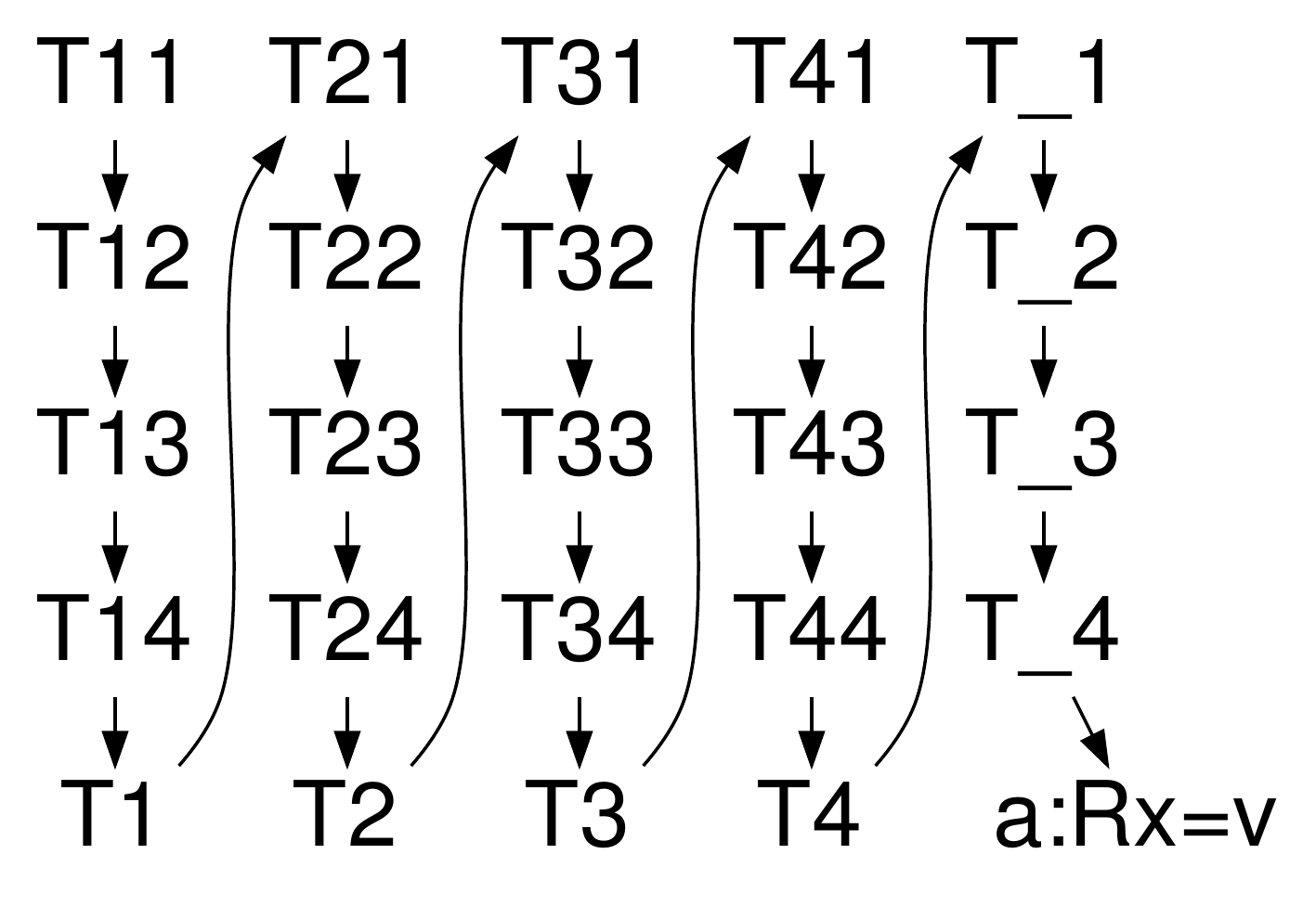}
  \vspace*{-7mm}
\end{wrapfigure}

The diagram on the right is an example walk, where each \texttt{Tn} is read of level~\texttt{n} of the Stage~1 table.
Each of those Stage~1 reads must first be translated to get the PA (as the table contains IPAs)
and so each \texttt{Tnk} is a read of level~\texttt{k} of the Stage~2 table for the address of the Stage~1 table at level~\texttt{n}.
Once the full Stage~1 walk has been completed the final output IPA must be translated to the final PA, and those are the final 4
\texttt{T\_n} reads, of the Stage~2 table at level~\texttt{n}.
The reads are ordered one after another in the order they appear in the ASL walk function.
This ordering must be respected by hardware
as software relies on it when building the tables bottom-up.

For example, if one starts with a Level 2 invalid entry, one might first create a Level 3 table (then a barrier to keep writes ordered).\testrefs{\testref{ROT.inv+dmbst}\\\testref{ROT.inv+dsb}} In other words, the architecture has to prohibit value speculation of page table entries.

\mysubsubsection{Dependencies into translation-walk non-TLB reads}\label{subsubsec:questions:non_tlb:dependencies}
Address dependencies into a memory-access instruction in classic ``user'' models are now explainable as dataflow dependencies
to the translation reads of those accesses, as the address has to be available before a walk can start.  %
These are virtual-address dataflow dependencies (contrasting with physical-address coherence).%
\testrefs{\testref{MP.RTf.inv+dmb+addr}\\\testref{WRC.RRTf.inv+addrs}}

\mysubsubsection{Translation-walk non-TLB reads from non-speculative same-thread writes}\label{subsubsec:questions:non_tlb:non-speculative-same-thread}\label{nonspec-same} \ 

\smallskip

\noindent\textit{PO-past} A translation-walk non-TLB read might read from a po-previous page-table-entry write, but it is only guaranteed to see such a write if there is enough intervening synchronisation.
Arm have recently introduced  \emph{Enhanced Translation
Synchronization} (ETS), optional in Armv8.0 and mandatory from Armv8.7.
Armv8-A implementations without ETS require both a DSB, to make the write visible to translation-walk non-TLB reads, and an ISB, to ensure that any translations for later instructions that were done out-of-order, before the write, are restarted.
With ETS, only the DSB is required for a translation-walk non-TLB read to definitely see the write, though one might still need an ISB if the new translation enables new instruction fetch. \testrefs{\testref{CoRpteTf.inv+dsb}\\\testref{MP.TTf.inv+dsbs}}

 Because invalid entries cannot be cached, this means that if an entry is initially invalid, then after a write of a valid entry and a DSB;ISB/DSB, translations will use that valid entry.  However, the DSB;ISB/DSB does not remove cached entries, so an initially valid entry might be cached by a spontaneous walk, so even after a write (of an invalid or non-BBM-failure valid entry) and a DSB;ISB/DSB, the old entry could still be used by translations.  One would need a TLBI sequence to remove old cached entries, which we return to below. %

\smallskip

\noindent\textit{PO-future}
The Armv8-A architecture allows load-store reordering, but it does not allow writes to become visible to other threads while they are still speculative. In the same vein, translation-walk non-TLB reads cannot read from po-later page-table-entry writes~\cite[D5.2.5 (p2683)]{G.a}. Before the po-earlier translation is complete, one cannot know that it is not going to fault, so the later write has to be considered speculative. This prevents a thread-local self-satisfying translation cycle, analogous to the prevention of load-store cycles with dependencies. \testrefs{\testref{CoTW1.inv}}

\smallskip

\noindent\textit{PO-present} On the margin, can a translation-walk non-TLB read for a write access see that write, or a distinct write from the same instruction?  The second case could arise from a store-pair or misaligned store that does two writes, with one to a page-table-entry that could be used by the other, though real code would typically not do this intentionally.\
This is explicitly allowed by the current architecture text~\cite[D5.2.5 (p2683)]{G.a}.
However that text does not specify whether the translations for those two writes could \emph{both} read from the other,
a self-satisfying translation cycle where the writes write each others translations.
In general such self-satsifying cycles give rise to \emph{thin air} behaviours and are universally forbidden by the architecture.

\mysubsubsection{Translation-walk non-TLB reads from speculative same-thread writes}

Speculative execution requires translation walks, which might result in additional page-table entries being cached, but in most cases this is indistinguishable from the effects of a non-speculative spontaneous walk.
However, one has to ask whether a translation-walk non-TLB read can see a po-previous write that is still speculative, e.g.~while both instructions follow an as-yet-unresolved conditional branch.  
It is clear that the result of such a walk should not be persistently cached, or made visible to other threads (via a shared TLB), while it remains speculative.  Moreover, such translations could lead to arbitrary reads of read-sensitive device locations, which one normally relies on the MMU to prevent.  The conclusion is therefore that this must be forbidden.
\testrefs{\testref{MP.RT.inv+dmb+ctrl-trfi}}

\mysubsubsection{Translation-walk non-TLB reads from same-thread writes, forbidden past (same-thread TLBI completion)}\label{forbidden-past-single}
To remove an existing mapping on a single thread,
one needs first to write an invalid entry, then a DSB to ensure that has reached memory and thus is visible to translation-walk non-TLB reads (to prevent spontaneous re-caching), then a TLBI to invalidate any cached entries, then a DSB to wait for TLBI completion. 
Without ETS, one
also needs an ISB to ensure that po-later translations that have been done early are restarted.
With ETS, the ISB is not always necessary,
though might still be needed for
its instruction-cache effects if the change of mapping affects instruction fetch.
After all that, an attempted access by that thread is guaranteed to fault. 
\testrefs{\testref{CoWinvT+dsb-isb}\\\testref{CoWinvT.EL1+dsb-tlbi-dsb}\\\testref{CoWinvT.EL1+dsb-tlbiis-dsb-isb}}

\mysubsubsection{Translation-walk non-TLB reads from other-thread writes, guaranteed past, initially invalid}\label{twoets}\testrefs{\testref{MP.RTf.inv+dmb+po}\\\testref{MP.RTf.inv+dmb+addr}}
Now consider when a translation-walk non-TLB read is guaranteed to see a write by another thread of a new entry, assuming that the entry was previously invalid and any cached entries for it
invalidated.  Consider a two-thread message-passing case, where a producer P0 writes a new
\begin{wrapfigure}{r}{80mm}
  \vspace*{-9mm}
\newsavebox{\flagthreadone}
\begin{lrbox}{\flagthreadone}
\begin{lstlisting}[language=AArch64,showlines=true,mathescape=true]
a:W pte(x)=pte_valid
<Producer ordering>
b:W flag=1
\end{lstlisting}
\end{lrbox}
\newsavebox{\flagthreadtwo}
\begin{lrbox}{\flagthreadtwo}
\begin{lstlisting}[language=AArch64,showlines=true,mathescape=true]
c:R flag=1
<$\textnormal{Receiver ordering}$>
d:Tx, for a Rx or Wx
\end{lstlisting}
\end{lrbox}
\begin{center}
\begin{tabular}{|l|l|}
  \hline
  P0 & P1\\
  \hline
  \usebox{\flagthreadone} & \usebox{\flagthreadtwo}\\
  \hline
\end{tabular}
\end{center}
  \vspace*{-11mm}
\end{wrapfigure}
valid
page table entry (\texttt{pte\_valid}), then has some ordering before a write of a flag, while a consumer P1 reads the flag, then has some ordering before an access \texttt{Rx} or \texttt{Wx} that needs that entry for a translation \texttt{Tx}    of  virtual address \texttt{x}. %

On some Armv8-A implementations that do not support ETS, some ``obvious'' combinations of ordering on P0 and P1 could lead to an abort of the translation of (d), which some OS software would find difficult to handle.  This was the main motivation for %
ETS:
implementations without it can have weak behaviour, requiring strong synchronisation to prevent the abort, while with ETS the architecture is stronger, requiring only weaker ordering to prevent the abort. 

Without ETS, two combinations of ordering are architected as sufficient to ensure that the translation (d) sees the new valid entry:
\begin{enumerate}
\item\label{noetsone} P0 has any ordered-before relationship, and P1 has DSB+ISB.\testrefs{\testref{MP.RTf.inv+dmb+dsb-isb}}
\item\label{noetstwo} P0 has DSB; TLBI; DSB, and P1 has any ordered-before relationship.\testrefs{\testref{MP.RTf.inv.EL1+dsb-tlbiis-dsb+dsb-isb}}
\end{enumerate}
In Case \ref{noetsone}, the message-passing is enough to ensure the write (a) is in main memory, the P1 ISB ensures that any out-of-order translation of (d) is restarted, and the P1 DSB keeps the read (c) and that ISB in order.
In Case~\ref{noetstwo}, the first DSB ensures the write is visible to all threads, the TLBI (broadcast, for the virtual address \texttt{x}) invalidates any older cached entry on P1, and the second DSB waits for that TLBI to be complete, after which any new translation on P1 will have to see the new entry.
However, it appears that the probability of an unhandleable abort in practice, where one usually does not have these operations immediately adjacent, and where in many cases the abort could be handled, has been judged low enough that OS code is not necessarily using either of these.

With ETS, the architecture says~\cite[D5.2.5,p2683]{G.a}
that \emph{``if a memory access RW1 is Ordered-before a second memory access RW2, then
RW1 is also Ordered-before any translation table walk generated by RW2 that generates a Translation fault, Address size fault, or Access flag fault.''}
Microarchitecturally, the intuition here is that with ETS any translation done while speculative that leads to such a fault will have to be reconfirmed as faulting when execution is no longer speculative, so an early faulting translation of (d) would have to be restarted after the ordered-before edges have ensured that (a) is visible. 
However, in the case that the RW2 instruction faults, there is no read or write event, and if the fault is a translation fault, there is no physical address.  One therefore has to ask what the meaning of ordered-before edges into RW2 is, especially for the parts of ordered-before dependent on physical addresses, such as coherence. 
The conclusion is that this should be only the non-physical-address parts of ordered-before into RW2, and in modelling one needs a ``ghost'' event to properly record what the dependencies would have been if it had succeeded.
Note that this includes ordered-before to RW2 that ends with a data dependency into a write, even though that data would not normally be necessary for the translation. %

Even with ETS, one might need an ISB on P1 if the new translation affects instruction fetch.%

\mysubsubsection{Translation-walk non-TLB reads from other-thread writes, guaranteed past, initially valid (other-thread TLBI completion)}

The following test
has a read-only mapping for
\begin{wrapfigure}{r}{80mm}
  \vspace*{-9mm}
\newsavebox{\twthreadone}
\begin{lrbox}{\twthreadone}
\begin{lstlisting}[language=AArch64,showlines=true,mathescape=true]
  STR pte_writeable,[pte($x$)]
  DSB SY
  TLBI VAAE1IS,[page($x$)]
  DSB SY
  MOV X7,#1
  STR X7,[$y$]
\end{lstlisting}
\end{lrbox}
\newsavebox{\twthreadtwo}
\begin{lrbox}{\twthreadtwo}
\begin{lstlisting}[language=AArch64,showlines=true,mathescape=true]
  LDR X0,[$y$]
  DMB SY
  MOV X1,#1
LO:
  STR X1,[$x$]

\end{lstlisting}
\end{lrbox}
\begin{center}
\begin{tabular}{|l|l|}
  \hline
  P0 & P1\\
  \hline
  \usebox{\twthreadone} & \usebox{\twthreadtwo}\\
  \hline
  \multicolumn{2}{|l|}{\cellcolor{Forbid}{Forbid: \mbox{\lstinline[mathescape=true]|1:X0=1 & permission_fault(L0, $x$)?|}}}\\
  \hline
\end{tabular}
\end{center}
  \vspace*{-9mm}
\end{wrapfigure}
some physical address
that is updated with a new
writeable mapping to the same physical
address, followed by a message-pass to another thread that attempts to write.
There is no requirement for break-before-make here, as the output address has not changed,
but TLB maintenance is required to ensure that the new writeable entry is guaranteed to be used by later translation-reads.

Arm forbid the outcome where the \asm{STR} faults due to a permission check.
This is because the \asm{TLBI} only completes once all instructions using any old translations which would be invalidated by the TLBI,\testrefs{\testref{RBS+dsb-tlbiis-dsb}}
on all other threads that the TLBI affects, have also completed, and the following \asm{DSB} waits for that (the same-thread case is different; see \S\ref{forbidden-past-single}).
In practice this means that once the \asm{TLBI} completes, one of the following holds:
either the final \asm{STR} has not performed its translation of $\textit{x}$ yet and will be required to see the writeable mapping for its page table entry (pte);
or the \asm{STR} has translated using the new writeable mapping;
or the \asm{STR} has already translated using the old read-only mapping,
in which case  we know that the \asm{STR} has finished and performed its write,
since the \asm{TLBI} could not complete while it was still in-progress.
In that case if the \asm{STR} has completed,
then so must have the locally-ordered-before \asm{LDR},
and that must have read 0.
This explanation also covers the make-after-break case above, %
for non-ETS Case~\ref{noetstwo}.

This is reflected in text to be included in future versions of the Arm ARM:
{\em
A TLB maintenance operation
[without nXS] generated by a TLB maintenance instruction is finished for a PE when:
\begin{enumerate}
\item all memory accesses generated by that PE using in-scope old translation information are complete.
 \item all memory accesses RWx generated by that PE are complete. RWx is the set of all memory accesses generated by instructions for that PE that appear in program order before an instruction (I1) executed by that PE where:
\begin{enumerate}
\item I1 uses the in-scope old translation information, and
     \item the use of the in-scope old translation information generates a synchronous data abort, and
     \item if I1 did not generate an abort from use of the in-scope old translation information, I1 would generate a memory access that RWx would be locally-ordered-before.
     \end{enumerate}
     \end{enumerate}
}

\mysubsubsection{Translation-walk reads from same- and other-thread writes, forbidden past (break-before-make)}\label{forbidden-past}

Now we can finally return to the break-before-make sequence.
Normal reads cannot read from the coherence-predecessors of the most coherence-recent write that is visible to them, but translation reads can read old (non-invalid) values from a TLB. 
To prevent this, and to ensure that a translation read sees a new page-table entry,  one  has to both ensure that any old TLB entries are invalidated, with a suitable TLBI, and that the new entry is visible to translation-walk non-TLB reads.

Armv8-A says~\cite[D5.10.1 (p2795)]{G.a} \emph{``A break-before-make sequence on changing from an old translation table entry to a new translation table entry
requires the following steps: (1) Replace the old translation table entry with an invalid entry, and execute a DSB instruction.  (2) Invalidate the translation table entry with a broadcast TLB invalidation instruction, and execute a DSB
instruction to ensure the completion of that invalidation.
(3) Write the new translation table entry, and execute a DSB instruction to ensure that the new entry is visible.''}.

Typically the write of an invalid entry and TLBI would be on the same thread, 
but more generally, any shape as below should be forbidden,\testrefs{\testref{CoRpteT.EL1+dsb-tlbi-dsb-isb}}
where \texttt{Tx} is a translation-walk read for an
\begin{wrapfigure}{r}{55mm}
  \vspace*{-5mm}
  \includegraphics[width=54mm]{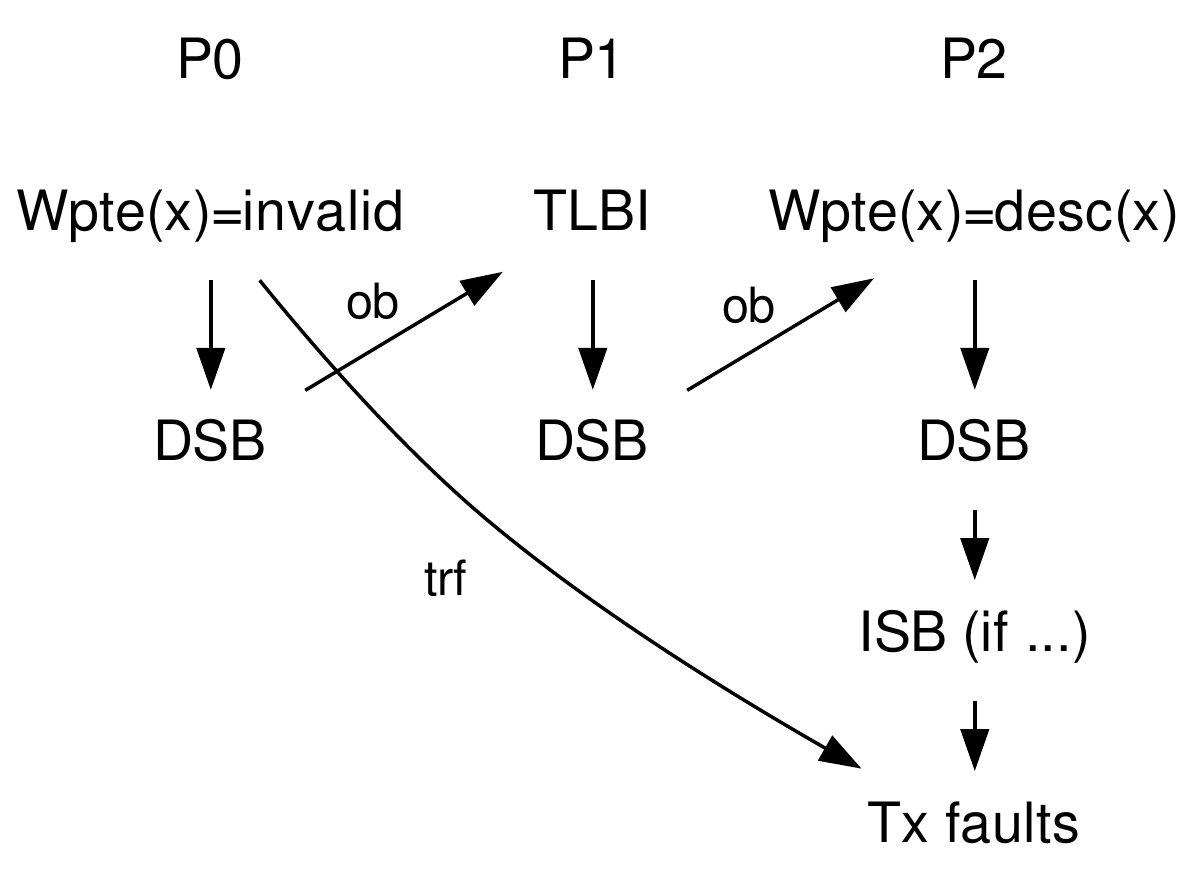}
  \vspace*{-5mm}
\end{wrapfigure}
access of \texttt{x} and the \texttt{trf} relation shows the page-table write it reads from.
In other words, the sequence ensures that the write of the
invalid entry, and of any co-predecessor writes, are hidden behind the new page-table entry  as far as new translations are concerned. 
Here the P0 DSB and P0-to-P1 \herd{ob} ensure the P0 write has propagated to memory before the P1 TLBI starts; the P1 DSB waits
for that TLBI to have finished on all threads; the P1-to-P2 \herd{ob} ensures that has happened before the new page-table-entry write starts; and the DSB ensures the new write has reached memory and so is visible to translation before subsequent instructions.  The P2 ISB is needed if on non-ETS hardware, to force restarts of any out-of-order translations for po-later instructions, or (on any hardware) if P2=P1, to ensure any later translations on the TLBI thread are restarted, or if the new mapping affects instruction fetch. 

This generalisation seems necessary, as a TLBI might be performed by a virtual CPU at EL1 which is interrupted and rescheduled by an EL2 hypervisor.  One should be able to rely on the hypervisor doing a DSB on the same hardware thread as part of the context switch, and that has to suffice.  It is sound because the DSBs and TLBI are all broadcast, though note that the DSB waiting for TLBI completion has to be on the same hardware thread as it. 

\mysubsubsection{Translation-walk non-TLB reads from other-thread writes, forbidden future}\label{other-forbidden-future}

Above
we saw that translation-walk non-TLB reads should not read from po-later writes.  How should
\begin{wrapfigure}{r}{55mm}
  \vspace*{-5mm}
  \includegraphics[width=54mm]{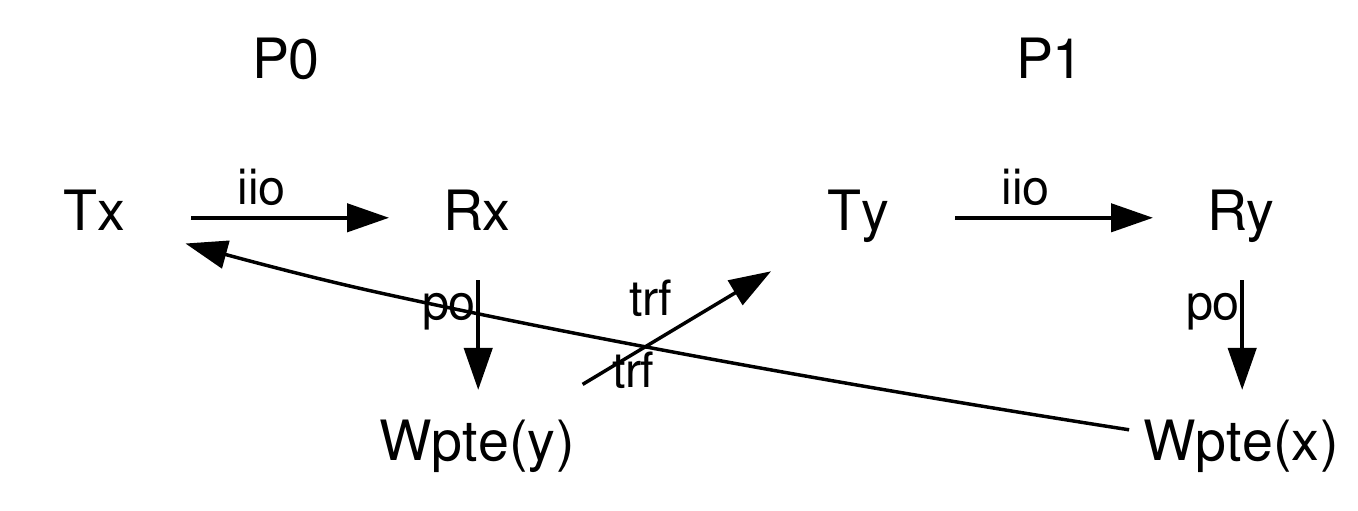}
  \vspace*{-5mm}
\end{wrapfigure}
that be
generalised to multiple threads?  \testrefs{\testref{LB.TT.inv+pos}}
For
the simplest example, consider the translation version of the LB test on the right, in which two threads translation-read from each other's po-future (\texttt{iio} relates translation reads to their accesses).  Standard LB shapes for normal accesses without dependencies are allowed in Armv8-A, but this example should be forbidden: until each translation is done, one cannot know that the first instruction on each thread will not abort, so one could not make the po-later write visible to the other thread without inter-thread roll-back. In other words, the possibility of translation aborts creates ordering rather like a control dependency from translation reads to po-later writes.\par

\mysubsubsection{Multicopy atomicity of translation-walk non-TLB reads}\label{wrctestsec}\testrefs{\testref{WRC.RRTf.inv+addrs}\\\testref{WRC.TfRR+dsb-isb+dsb}\\\testref{WRC.TTTf.inv+addrs}}
The ARMv7 and early Armv8-A architectures for normal accesses were \emph{non-multicopy-atomic}: a write could become
visible to some other threads before becoming visible to all threads, broadly similar in this respect
to the IBM POWER architecture~\cite{Power2.07,pldi105}.
This is one of the most fundamental choices for a relaxed memory model. 
In 2017 Arm revised their Armv8-A architecture to be \emph{multicopy-atomic} (\emph{other multicopy-atomic}, or OMCA, in their terminology), a considerable simplification~\cite{armv8-mca,B.a}.  However, there was no consideration at the time of whether this should also apply to the visibility of writes by translation-walk non-TLB reads, or of the force of the ARM statement that \emph{a translation table walk is considered to be a separate observer}~\cite[D5.10.2 (p2808)]{G.a}. 

For example, consider the following translation-read analogue of the classic WRC+addrs test, which would be forbidden in OMCA Armv8-A for normal reads.
Suppose one has ETS,
\begin{wrapfigure}{r}{70mm}
  \vspace*{-5mm}
  \includegraphics[width=70mm]{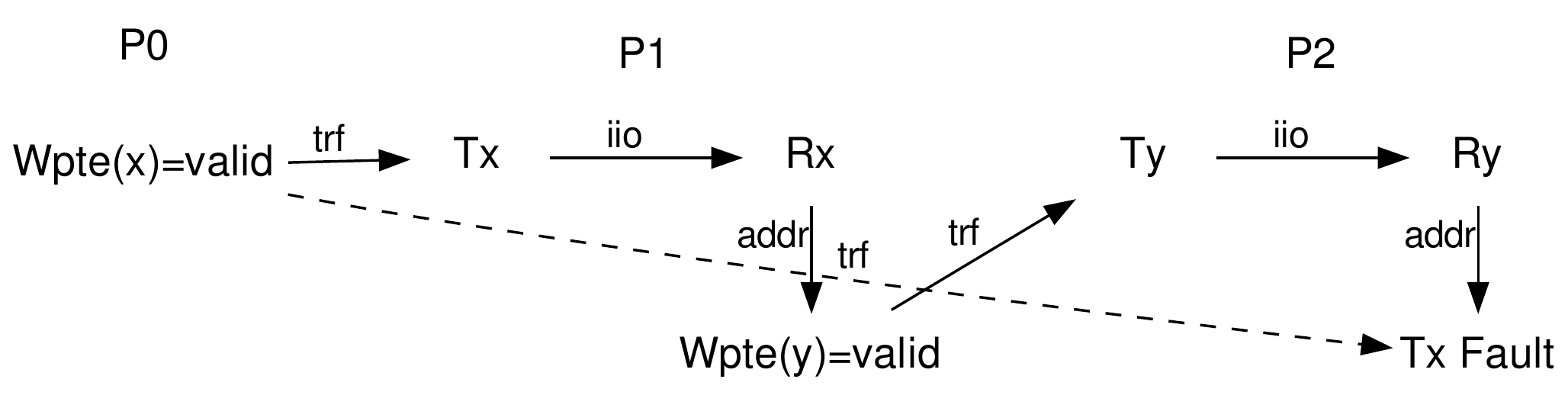}
  \vspace*{-5mm}
\end{wrapfigure}
the last-level
page-table entries for \asm{x} and \asm{y} are
initially invalid and not cached in any TLB, P0 writes a valid entry for \asm{x}, P1 does a translation that sees that entry and then (via an address dependency) writes a valid entry for \asm{y}, then P2 does a translation that sees that entry and then (via an address dependency) tries a translation for \asm{x}, is that last guaranteed to see the valid entry instead of faulting?  This might be exhibited by a microarchitecture with a shared TLB between P0 and P1 (e.g.~if they are SMT threads on the same core, or have a shared TLB for a subcluster).
The tentative Arm conclusion is that this should be forbidden, to avoid software issues with unexpected aborts similar to those motivating ETS.

Now consider the above translation version of LB, generalising from %
po-future writes to other ob-future writes. For transitive combinations of reads-from and dependencies, it 
should clearly still be forbidden, to avoid needing inter-thread roll-back, but for \cat{ob} including coherence edges (\cat{coe}) one can imagine that a translate read could see a write before the coherence relationships are established, analogous to the weakness of coherence in the Power
non-MCA
model.
Discussion of these and other cases with Arm led to the 
tentative conclusion for Armv8-A that translation-walk non-TLB reads (like normal reads) do not see any non-OMCA behaviour.  In other words, there is no programmer-visible caching observable to some non-singleton subsets of threads' translations but not others.

\subsection{Further issues}
Our discussions with Arm identified and clarified various other architectural choices, though we do not discuss them fully here, and our models do not cover them at present. %
To give a flavour:
(1) Misaligned or load/store-pair instructions give rise to multiple accesses, which might be to different pages. Each has their own translation; not ordered w.r.t.~each other, and with no prioritisation of faults between them.  As noted in \S\ref{nonspec-same}, one might translate-read from the other, but not both simultaneously.
(2) Normal registers act like a per-thread sequential memory, with reads reading from the most recent po-previous write, but 
the system registers that control translations can have more relaxed behaviour, requiring \asm{ISB}s to enforce sequential behaviour.
(3) The architecture requires, and OSs rely on, the fact that turning on
the MMU does not need TLB maintenance.  However, in a two-stage world, if
Stage~1 is off, one is still using the TLB for Stage~2, so entries do get
added to the TLB.  When one later turns on Stage~1, it is
essential that the entries added from those earlier Stage~2 translations are not
 used, so one has to
regard them as from a 257'th ASID.

%
%
%
%
%
%
%
%
%
%
%
%
%

%
%

%

%

% auto-generated by cutlines; do not edit

\section{Virtual memory in the pKVM production hypervisor}\label{sec:tests}\label{sec:pkvm}%
Protected KVM, or pKVM~\cite{pKVM-lwn,pKVM-linux-talk,pKVM-src}, is currently being developed by Google to provide a common hypervisor for Android, to provide improved compartmentalisation by a small trusted computing base (TCB) between the Linux kernel and other services.
pKVM is built as a component of Linux. During boot, the Linux kernel hands over control of EL2 to the pKVM code, which constructs a memory map for itself and a Stage~2 memory map to encapsulate the Linux kernel. The Linux kernel thereafter runs only at EL1 (managing EL1\&0 Stage~1 memory maps for itself and for user processes), as the \emph{principal guest}, also known as the \emph{host} (not to be confused with the host hardware).   Other services can run as other guests, which are protected from the kernel and vice versa.  The kernel remains responsible for scheduling, but context switching %
and inter-guest communication is done by hypervisor calls to the pKVM code at EL2.
This gives us an ideal setting in which to examine the management of virtual memory by production code for Armv8-A relaxed-memory-concurrency, with both one and two stages of translation (for EL2 and EL1\&0 respectively).
The pKVM codebase is small, so it is feasible to examine all uses of TLB management, and we benefit from discussions with
the pKVM development team.
We have manually abstracted the main pKVM relaxed-virtual-memory scenarios into \numberOfpKVMTests{} tests.

To give a flavour for these we will explore just a few of the most fundamental for
hypervisor control of translation tables:
\cc{__pkvm_init} performs first-time per-CPU initialisation of pKVM,
which includes setting up the \cc{vmemmap},
a large array storing a struct entry for each physical page in memory
with ownership information;
\cc{__kvm_vcpu_run} switches to a different guest;
\cc{__kvm_flush_vm_context} flushes all entries in all TLB caches that relate to Stage~1 translations;
\cc{__kvm_tlb_flush_vmid_ipa} flushes entries in caches that relate to a particular guest for a given address;
\cc{__kvm_tlb_flush_vmid} flushes all entries in caches that relate to a particular guest;
\cc{__kvm_flush_cpu_context} flushes all entries in caches at EL1 for a given guest;
\cc{__pkvm_cpu_set_vector} sets the vector base address for handling exceptions;
\cc{__pkvm_create_mappings} creates entries in the Stage~1 page table that pKVM uses when it executes;
and
\cc{__pkvm_prot_finalize} which enables Stage 2 translations during boot.

In the remainder of this section we describe a selection of litmus tests
extracted from the current pKVM source code~\cite{pKVM-src} as of Dec~2021.

\subsection{Switching to another guest}\label{simpletest}
The most basic task that pKVM is in charge of is the actual swapping from one vCPU's context to another
to execute a particular VM on the physical CPU.
Deciding \emph{when} and to \emph{which} vCPU to switch is done by the underlying host Linux guest.
pKVM's role is to safely save the current guest's state,
and load the EL1 state for the target vCPU,
and manage the Stage~2 translation tables for the vCPU's VM.

Aside from the fundamental shape of switching from one VM to another on the same core,
there are a few interesting cases to consider:
\begin{itemize}
  \item Switching from one vCPU in a VM to another vCPU in the same VM.
  \item Switching to a new vCPU, re-using a previous VMID.
  \item Re-using an old VMID with a concurrently executing vCPU on another core.
\end{itemize}

When the host guest Linux kernel wishes to perform the switch from one vCPU to another,
it writes to some general-purpose registers and then performs an \asm{HVC} instruction
(in this case, to call the \cc{\_\_kvm\_vcpu\_run} hypercall).

\clearpage
\testPara{pKVM.vcpu\_run}
In the simplest case
where pKVM is just switching from one vCPU to another vCPU
in a different VM,
pKVM restores the per-CPU register state
and sets the \asm{VTTBR} with the new VMID.
So long as the two vCPUs are using disjoint VMIDs
there is no requirement for TLB maintenance.

This test, \cc{pKVM.vcpu\_run}, is below, %
\begin{figure}[h]
\noindent\hspace*{-20mm}\scalebox{0.9}{\newlength{\pKVMdotvcpurunH}
\setlength{\pKVMdotvcpurunH}{0cm}
\newsavebox{\pKVMdotvcpurunI}
\begin{lrbox}{\pKVMdotvcpurunI}
\begin{lstlisting}[language=AArch64,showlines=true]
msr ttbr0_el1, x0  // kvm/hyp/sysreg-sr.h:96
msr vttbr_el2, x1  // include/asm/kvm_mmu.h:276
eret               // kvm/hyp/nvhe/host.S
L0:
ldr x2, [x3]       // in guest
\end{lstlisting}
\end{lrbox}
\newlength{\pKVMdotvcpurunJ}
\settowidth{\pKVMdotvcpurunJ}{\usebox{\pKVMdotvcpurunI}}
\addtolength{\pKVMdotvcpurunH}{\pKVMdotvcpurunJ}
\newsavebox{\pKVMdotvcpurunK}
\begin{lrbox}{\pKVMdotvcpurunK}
\begin{lstlisting}[language=AArch64,showlines=true]
0x1400:
mov x2, #0
\end{lstlisting}
\end{lrbox}
\newlength{\pKVMdotvcpurunL}
\settowidth{\pKVMdotvcpurunL}{\usebox{\pKVMdotvcpurunK}}
\addtolength{\pKVMdotvcpurunH}{\pKVMdotvcpurunL}
\newsavebox{\pKVMdotvcpurunM}
\begin{lrbox}{\pKVMdotvcpurunM}
\begin{lstlisting}[language=IslaPageTableSetup,showlines=true]
option default_tables = false;
virtual x;
physical pa1 pa2;
intermediate ipa1 ipa2;
s1table hyp_map 0x200000 {
  identity 0x1000 with code;
  x |-> invalid; }
s1table vm1_stage1 0x2C0000 {
  x |-> ipa1; }
s1table vm2_stage1 0x300000 {
  x |-> ipa2; }
s2table vm1_stage2 0x240000 {
  ipa1 |-> pa1;
  ipa2 |-> invalid;
  s1table vm1_stage1; }
s2table vm2_stage2 0x280000 {
  ipa1 |-> invalid;
  ipa2 |-> pa2;
  s1table vm2_stage1; }
*pa2 = 1;
\end{lstlisting}
\end{lrbox}
\newlength{\pKVMdotvcpurunN}
\settowidth{\pKVMdotvcpurunN}{\usebox{\pKVMdotvcpurunM}}
\newsavebox{\pKVMdotvcpurunO}
\begin{lrbox}{\pKVMdotvcpurunO}
\begin{tabular}{l}
\vphantom{$\vcenter{\hbox{\rule{0pt}{1.8em}}}$}Initial state:\\
\lstinline[language=IslaLitmusExp]|PSTATE.EL=0b10|  // initial exception level is EL2\\
\lstinline[language=IslaLitmusExp]|VBAR_EL2=0x1000| // exception vector base address\\
\lstinline[language=IslaLitmusExp]|ELR_EL2=L0:|  // exception link register, to return to from EL2\\
\lstinline[language=IslaLitmusExp]|SPSR_EL2=0b00101| // saved program status\\
\lstinline[language=IslaLitmusExp]|TTBR0_EL1=ttbr(asid=0x00,base=vm1_stage1)| //
%translation table base 0, for
EL1 Stage 1\\
\lstinline[language=IslaLitmusExp]|VTTBR_EL2=ttbr(vmid=0x0001,base=vm1_stage2)| //
%virtualization translation table base, for
Stage 2\\
\lstinline[language=IslaLitmusExp]|TTBR0_EL2=ttbr(base=hyp_map,asid=0x00)| //
%translation table base 0, for
EL2\\
\lstinline[language=IslaLitmusExp]|x0=ttbr(asid=0x00,base=vm2_stage1)|\\
\lstinline[language=IslaLitmusExp]|x1=ttbr(base=vm2_stage2,vmid=0x0002)|\\
\lstinline[language=IslaLitmusExp]|x3=x|\\
\end{tabular}
\end{lrbox}
\newlength{\pKVMdotvcpurunJJ}
\settowidth{\pKVMdotvcpurunJJ}{\usebox{\pKVMdotvcpurunO}}
\begin{tabular}{|l|l|}
  \multicolumn{2}{l}{\textbf{AArch64} \lstinline[language=IslaLitmusName]|pKVM.vcpu_run|}\\
  \hline
  \multirow{6}{*}{\begin{minipage}{\pKVMdotvcpurunN}%
  %\vphantom{$\vcenter{\hbox{\rule{0pt}{1.8em}}}$}
  \vspace*{-5\baselineskip}%
  Page table setup:\\\usebox{\pKVMdotvcpurunM}\end{minipage}} & \cellcolor{IslaInitialState}{\usebox{\pKVMdotvcpurunO}}\\
  \cline{2-2}
  & Thread 0 (with pKVM source lines)\\
  \cline{2-2}
  & \usebox{\pKVMdotvcpurunI}\\
  \cline{2-2}
  & Thread 0 EL2 handler\\
  \cline{2-2}
  & \usebox{\pKVMdotvcpurunK}\\
  \cline{2-2}
  & Final state: \lstinline[language=IslaLitmusExp]|0:x2=0|\\
  \hline
\end{tabular}
}%
\end{figure}
typeset (lightly hand-edited) from the TOML input format of our Isla tool (\S\ref{isla}).
Here there is a single physical CPU, initially running a virtual machine VM1, with VMID \texttt{0x0001}, at EL1.
The %
section on the left defines the initial and all potential states of the page tables, and any other memory state.   
This test sets up separate translation tables for pKVM at EL2 (which has just a single stage) and for two VMs (each with two stages, Stage 2 controlled by pKVM and Stage 1 controlled by the VM). 
pKVM's own mapping \texttt{hyp\_map} maps its code. 
VM1's own Stage~1 mapping \texttt{vm1\_stage1} maps virtual address \texttt{x} to \texttt{ipa1},
and the initial pKVM-managed Stage 2 mapping \texttt{vm1\_stage2} maps that \texttt{ipa1} to \texttt{pa1}, which implicitly initially holds 0.
These page tables are described concisely %
by a small declarative language we developed, %
determining the page-table memory (here $\sim$30k)
required for the %
Armv8-A page-table walks. %

The %
top-right block gives the initial Thread~0 register values, including the various page-table base registers.
The bottom-right blocks give the code of the test. This starts running at EL2, as one can see from the \texttt{PSTATE.EL} register value. 
The key assembly lines are annotated with the pKVM source line numbers they correspond to.
To switch to run another virtual machine VM2, with VMID \texttt{0x0002}, on this same physical CPU, pKVM changes  \texttt{VTTBR\_EL2} to the new \texttt{vm2\_stage2} mapping and, as part of the context-switch register-file changes, restores \texttt{TTBR0\_EL1} to the VM2's own Stage~1 mapping \texttt{vm2\_stage1}.
The code then executes an \asm{ERET} (``exception-return'') instruction to return to EL1, and then tries to read \texttt{x}.
The test includes a final assertion %
of the relaxed outcome
that register \texttt{x2=0}, which could occur
if the \texttt{ldr} translation used the old VM1 mapping instead of VM2's mapping.  In this case that should not be allowed.
Other tests capture more elaborate scenarios. For example,
currently the host kernel manages VMIDs and assigns each VM its own VMID.
If the host runs out of VMIDs to allocate to new vCPUs,
it currently revokes all previously allocated VMIDs and re-allocates from the beginning, during which
pKVM has to ensure that any old vCPUs' translations using that VMID are expelled from any TLBs (\cc{pKVM.vcpu\_run.update\_vmid}).
If there is a concurrently executing vCPU using that VMID,
that vCPU must be paused
until after the new VMID generation (and hence any required TLB maintenance),
before continuing
with the freshly allocated VMID (\cc{pKVM.vcpu\_run.update\_vmid.concurrent}).
For another example, for pKVM to maintain the illusion that each vCPU is on its own core,
the per-core state must be cleaned between running different vCPUs,
including ensuring that translations for one vCPU are not cached and visible
to another, even if they happen to be in the same VM (and using the same VMID) (\cc{pKVM.vcpu\_run.same\_vm}).

The most basic task that pKVM is in charge of is the actual swapping from one vCPU's context to another
to execute a particular VM on the physical CPU.
Deciding \emph{when} and \emph{to which} vCPU to switch is done by the underlying host Linux guest.
pKVM's role is to safely save off the current guest's state,
and load the EL1 state for the target vCPU,
and manage the Stage~2 translation tables for the vCPU's VM.

Aside from the fundamental shape of switching from one VM to another on the same core,
there are a few interesting cases to consider:
\begin{itemize}
  \item Switching from one vCPU in a VM to another vCPU in the same VM.
  \item Switching to a new vCPU, re-using a previous VMID.
  \item Re-using an old VMID with a concurrently executing vCPU on another core.
\end{itemize}

When the host guest Linux kernel wishes to perform the switch from one vCPU to another,
it writes to some general-purpose registers and then performs an \asm{HVC} instruction
(in this case, to call the \cc{\_\_kvm\_vcpu\_run} hypercall).

\clearpage
\testpara{pKVM.vcpu\_run.update\_vmid}

Since pKVM's host Linux kernel is responsible for all scheduling,
it is responsible for the selection of VMIDs.
Typically, each VM is given its own VMID,
and all vCPUs within that VM share that VMID.
By giving distinct VMs different VMIDs, the host Linux kernel
can switch between vCPUs in different VMs without requiring TLB maintenance,
as was seen in the previous test.
Eventually, the host kernel runs out of VMIDs to allocate,
either from the naive incremental-allocation or because the
number of VMs is larger than the available VMID space
(8 or 16 bits in the worst case).
Currently, pKVM's host Linux kernel will simply revoke all previously allocated VMIDs,
and when pKVM goes to switch to a vCPU for a VM without an allocated VMID, it will be allocated there and then.

The \PKVMTEST{pKVM.vcpu\_run.update\_vmid} test in Figure~\ref{fig:vcpurunupdatevmid_code} is the simplest case, where
the new vCPU is the only vCPU of that VM that is running, and no other physical CPUs are executing.
Here, we have two VMs and their associated tables in the initial state:
\herd{vm1\_stage2} is the root of the Stage~2 table for the first VM,
and \herd{vm2\_stage2} is the root of the Stage~2 table for the second VM.
Initially, the \asm{VTTBR} points to VM1's table,
with VMID 1.
When switching to VM2's table,
with the same VMID,
the host Linux kernel begins a sequence of calls to the hypervisor:
first it must flush the CPU context,
to make sure any old cached translations for the old VMID are removed.
From there, pKVM can perform the switch to the new vCPU by restoring the EL1 system register
state,
then pointing the \asm{VTTBR} to the new table with the new ASID,
before doing an exception-return to the guest.
If the guest then accesses a location, it should be guaranteed to use the
new Stage~2 mapping with its own Stage~1 tables.
Figure~\ref{fig:vcpurunupdatevmid_code} contains the code listing and initial conditions for this test;
execution begins at EL2 with \cc{vm1}'s vCPU state, with VMID 1.
The test performs the pKVM \cc{vcpu\_run} sequence to clean up the TLB,
switch to \cc{vm2}'s EL1 state and then switching to \cc{vm2}'s Stage~2 mapping
with VMID 1 (the same as what \cc{vm1} had been using)
before returning to the guest.
The test then asserts that the guest's access uses \cc{vm2}'s own restored Stage~1 translation tables,
and does not see the old Stage~2 entries of \cc{vm1}.

Note that this test \emph{begins} with TLB invalidation.
To really capture the full sequence that is microarchitecturally interesting,
the test should really begin from execution at EL1 inside the guest
allowing Stage~2 TLB fills.
The test beginning at EL2 from a clean machine state would not give the TLB time to
actually fill with stale entries from \cc{vm1}'s Stage~2 translation tables
in an operational model.
Also note that, currently, Isla does not produce candidates with re-ordering of system register reads,
and so there are no `bad' candidate executions to show from the Isla-generated executions currently.
This is work-in-progress to find the correct semantics for such re-orderings.

\begin{figure}
  \centering
  \scalebox{0.8}{\newlength{\vcpurunupdatevmidH}
\setlength{\vcpurunupdatevmidH}{0cm}
\newsavebox{\vcpurunupdatevmidI}
\begin{lrbox}{\vcpurunupdatevmidI}
\begin{lstlisting}[language=AArch64,showlines=true]
// kvm/hyp/nvhe/tlb.c:145
dsb ishst
// kvm/hyp/nvhe/tlb.c:146
tlbi alle1is
// kvm/hyp/nvhe/tlb.c:160
dsb sy
// kvm/hyp/sysreg-sr.h:96
MSR TTBR0_EL1,X0
// include/asm/kvm_mmu.h:276
MSR VTTBR_EL2,X1
// kvm/hyp/nvhe/host.S
ERET
L0:
// in guest
LDR X2,[X3]
\end{lstlisting}
\end{lrbox}
\newlength{\vcpurunupdatevmidJ}
\settowidth{\vcpurunupdatevmidJ}{\usebox{\vcpurunupdatevmidI}}
\addtolength{\vcpurunupdatevmidH}{\vcpurunupdatevmidJ}
\newsavebox{\vcpurunupdatevmidK}
\begin{lrbox}{\vcpurunupdatevmidK}
\begin{lstlisting}[language=AArch64,showlines=true]
0x1200:
mov x2, #0
// data abort preferred-return-address is itself
// so jump to next instr instead
mrs x20,elr_el2
add x20,x20,#4
msr elr_el2,x20
eret
\end{lstlisting}
\end{lrbox}
\newlength{\vcpurunupdatevmidL}
\settowidth{\vcpurunupdatevmidL}{\usebox{\vcpurunupdatevmidK}}
\addtolength{\vcpurunupdatevmidH}{\vcpurunupdatevmidL}
\newsavebox{\vcpurunupdatevmidM}
\begin{lrbox}{\vcpurunupdatevmidM}
\begin{lstlisting}[language=IslaPageTableSetup,showlines=true]
option default_tables = false;
physical pa1 pa2;
intermediate ipa1 ipa2;

s1table hyp_map 0x200000 {
    x |-> invalid;
}

s2table vm1_stage2 0x300000 {
    ipa1 |-> pa1;
    ipa2 |-> invalid;

    s1table vm1_stage1 0x280000 {
        x |-> ipa1;
    }
}

s2table vm2_stage2 0x380000 {
    ipa1 |-> invalid;
    ipa2 |-> pa2;

    s1table vm2_stage1 0x2C0000 {
        x |-> ipa2;
    }
}

*pa2 = 1;
\end{lstlisting}
\end{lrbox}
\newlength{\vcpurunupdatevmidN}
\settowidth{\vcpurunupdatevmidN}{\usebox{\vcpurunupdatevmidM}}
\newsavebox{\vcpurunupdatevmidO}
\begin{lrbox}{\vcpurunupdatevmidO}
\begin{minipage}{0.5\vcpurunupdatevmidH}
\vphantom{$\vcenter{\hbox{\rule{0pt}{1.8em}}}$}Initial state:\\
\lstinline[language=IslaLitmusExp]|R1=ttbr(base=vm2_stage2,vmid=0x0001)|\\
\lstinline[language=IslaLitmusExp]|PSTATE.SP=0b1|\\
\lstinline[language=IslaLitmusExp]|TTBR0_EL1=ttbr(base=vm1_stage1,asid=0x0000)|\\
\lstinline[language=IslaLitmusExp]|R3=x|\\
\lstinline[language=IslaLitmusExp]|VBAR_EL2=0x1000|\\
\lstinline[language=IslaLitmusExp]|VTTBR_EL2=ttbr(vmid=0x0001,base=vm1_stage2)|\\
\lstinline[language=IslaLitmusExp]|PSTATE.EL=0b10|\\
\lstinline[language=IslaLitmusExp]|TTBR0_EL2=ttbr(asid=0x0000,base=hyp_map)|\\
\lstinline[language=IslaLitmusExp]|ELR_EL2=L0:|\\
\lstinline[language=IslaLitmusExp]|SPSR_EL2=0b00101|\\
\lstinline[language=IslaLitmusExp]|R0=ttbr(base=vm2_stage1,asid=0x0000)|\\
\end{minipage}
\end{lrbox}
\begin{tabular}{|l|l|}
  \multicolumn{2}{l}{\textbf{AArch64} \lstinline[language=IslaLitmusName]|pKVM.vcpu_run.update_vmid|}\\
  \hline
  \multirow{6}{*}{\begin{minipage}{\vcpurunupdatevmidN}\vphantom{$\vcenter{\hbox{\rule{0pt}{1.8em}}}$}Page table setup:\\\usebox{\vcpurunupdatevmidM}\end{minipage}} & \cellcolor{IslaInitialState}{{\usebox{\vcpurunupdatevmidO}}}\\
  \cline{2-2}
  & Thread 0\\
  \cline{2-2}
  & {\usebox{\vcpurunupdatevmidI}}\\
  \cline{2-2}
  & thread0 el2 handler\\
  \cline{2-2}
  & {\usebox{\vcpurunupdatevmidK}}\\
  \cline{2-2}
  & Final state: \lstinline[language=IslaLitmusExp]|0:R2=0|\\
  \hline
\end{tabular}
}
\caption{}\label{fig:vcpurunupdatevmid_code}
\end{figure}

\clearpage
\testpara{pKVM.vcpu\_run.update\_vmid.concurrent}

The concurrent case of the previous \PKVMTEST{pKVM.vcpu\_run.update\_vmid} test is critical:
while switching from one vCPU to another on the same core is typically a
thread-local event,
care must be taken in the case where another CPU is executing a VM whose current VMID should be revoked, as in Fig.~\ref{fig:vcpurunupdatevmidconcurrent}.

In this case, pKVM must interrupt the other core, so that it is not concurrently executing while the new VMIDs are being
allocated; otherwise, it might pollute the address space of the VM that gets allocated that VMID.

pKVM does this by sending an IPI to all the other cores to break out of their current vCPUs.
When it does this, each CPU will attempt to switch back to the host, and, in doing so,
will be forced to take a lock when acquiring the VMID to use.
This lock prevents those CPUs from executing at EL1,
and the architecture prevents the hardware from performing TLB fills while at EL2.
These guarantees ensure that while VMIDs are being re-allocated and TLB maintenance performed
on the original core,
no stale entries can find their way back into the TLB.

\begin{figure}
\noindent\hspace*{-19mm}
  \resizebox*{!}{0.95\textheight}{\newlength{\vcpurunupdatevmidconcurrentH}
\setlength{\vcpurunupdatevmidconcurrentH}{0cm}
\newsavebox{\vcpurunupdatevmidconcurrentI}
\begin{lrbox}{\vcpurunupdatevmidconcurrentI}
\begin{lstlisting}[language=AArch64,showlines=true]
// kvm/arm.c:551 force_vm_exit(cpu_all_mask);
STR X2,[X3]
// kvm/arm.c:551 force_vm_exit(cpu_all_mask);
LDR X4,[X5]
// kvm/hyp/nvhe/tlb.c:145
dsb ishst
// kvm/hyp/nvhe/tlb.c:146
tlbi alle1is
// kvm/hyp/nvhe/tlb.c:160
dsb sy
// kvm/arm.c:567 spin_unlock(&kvm_vmid_lock);
STR X6,[X7]
\end{lstlisting}
\end{lrbox}
\newlength{\vcpurunupdatevmidconcurrentJ}
\settowidth{\vcpurunupdatevmidconcurrentJ}{\usebox{\vcpurunupdatevmidconcurrentI}}
\addtolength{\vcpurunupdatevmidconcurrentH}{\vcpurunupdatevmidconcurrentJ}
\newsavebox{\vcpurunupdatevmidconcurrentK}
\begin{lrbox}{\vcpurunupdatevmidconcurrentK}
\begin{lstlisting}[language=AArch64,showlines=true]
// in guest, read X
LDR X0,[X1]
DSB SY
ISB
// fake IPI
HVC #0
// try again
LDR X2,[X3]
\end{lstlisting}
\end{lrbox}
\newlength{\vcpurunupdatevmidconcurrentL}
\settowidth{\vcpurunupdatevmidconcurrentL}{\usebox{\vcpurunupdatevmidconcurrentK}}
\addtolength{\vcpurunupdatevmidconcurrentH}{\vcpurunupdatevmidconcurrentL}
\newsavebox{\vcpurunupdatevmidconcurrentM}
\begin{lrbox}{\vcpurunupdatevmidconcurrentM}
\begin{lstlisting}[language=AArch64,showlines=true]
0x1400:
// smb recieve
LDR X4,[X5]
// smb reply
STR X6,[X7]
// kvm/hyp/include/nvhe/spinlock.h
LDAR X8,[X9]
// kvm/hyp/sysreg-sr.h:96
MSR TTBR0_EL1,X10
// include/asm/kvm_mmu.h:276
MSR VTTBR_EL2,X11
// kvm/hyp/nvhe/host.S
ERET
\end{lstlisting}
\end{lrbox}
\newlength{\vcpurunupdatevmidconcurrentN}
\settowidth{\vcpurunupdatevmidconcurrentN}{\usebox{\vcpurunupdatevmidconcurrentM}}
\addtolength{\vcpurunupdatevmidconcurrentH}{\vcpurunupdatevmidconcurrentN}
\newsavebox{\vcpurunupdatevmidconcurrentO}
\begin{lrbox}{\vcpurunupdatevmidconcurrentO}
\begin{lstlisting}[language=IslaPageTableSetup,showlines=true]
option default_tables = false;
physical pa1 pa2 pa_ipi pa_kvm_vmid_lock;
intermediate ipa1 ipa2;

s1table hyp_map 0x200000  {
    identity 0x1000 with code;
    x |-> invalid;
    ipi |-> pa_ipi;
    kvm_vmid_lock |-> pa_kvm_vmid_lock;
}

s2table vm1_stage2 0x300000  {
    ipa1 |-> pa1;
    ipa1 ?-> invalid;
    ipa2 |-> invalid;
    ipa2 ?-> pa2;

    s1table vm1_stage1 0x280000 {
        x |-> ipa1;
    }
}

s2table vm2_stage2 0x380000  {
    ipa1 |-> invalid;
    ipa1 ?-> pa1;
    ipa2 |-> pa2;
    ipa2 ?-> invalid;

    s1table vm2_stage1 0x2C0000 {
        x |-> ipa2;
    }
}

*pa2 = 1;
*pa_kvm_vmid_lock = 1;
\end{lstlisting}
\end{lrbox}
\newlength{\vcpurunupdatevmidconcurrentP}
\settowidth{\vcpurunupdatevmidconcurrentP}{\usebox{\vcpurunupdatevmidconcurrentO}}
\newsavebox{\vcpurunupdatevmidconcurrentA}
\begin{lrbox}{\vcpurunupdatevmidconcurrentA}
\begin{minipage}{0.6\vcpurunupdatevmidconcurrentH}
\vphantom{$\vcenter{\hbox{\rule{0pt}{1.8em}}}$}Initial state:\\
\lstinline[language=IslaLitmusExp]|0:R7=kvm_vmid_lock|\\
\lstinline[language=IslaLitmusExp]|0:PSTATE.SP=0b1|\\
\lstinline[language=IslaLitmusExp]|0:R6=0b0|\\
\lstinline[language=IslaLitmusExp]|0:R2=0b1|\\
\lstinline[language=IslaLitmusExp]|0:TTBR0_EL2=ttbr(asid=0x0000,base=hyp_map)|\\
\lstinline[language=IslaLitmusExp]|0:PSTATE.EL=0b10|\\
\lstinline[language=IslaLitmusExp]|0:R3=ipi|\\
\lstinline[language=IslaLitmusExp]|0:R5=ipi|\\
\lstinline[language=IslaLitmusExp]|0:R0=0b1|\\
\lstinline[language=IslaLitmusExp]|0:R1=kvm_vmid_lock|\\
\lstinline[language=IslaLitmusExp]|1:TTBR0_EL2=ttbr(base=hyp_map,asid=0x0000)|\\
\lstinline[language=IslaLitmusExp]|1:TTBR0_EL1=ttbr(base=vm1_stage1,asid=0x0000)|\\
\lstinline[language=IslaLitmusExp]|1:VBAR_EL2=0x1000|\\
\lstinline[language=IslaLitmusExp]|1:VTTBR_EL2=ttbr(base=vm1_stage2,vmid=0x0001)|\\
\lstinline[language=IslaLitmusExp]|1:R1=x|\\
\lstinline[language=IslaLitmusExp]|1:R3=x|\\
\lstinline[language=IslaLitmusExp]|1:PSTATE.EL=0b01|\\
\lstinline[language=IslaLitmusExp]|1:R6=0b10|\\
\lstinline[language=IslaLitmusExp]|1:R10=ttbr(asid=0x0000,base=vm2_stage1)|\\
\lstinline[language=IslaLitmusExp]|1:R9=kvm_vmid_lock|\\
\lstinline[language=IslaLitmusExp]|1:R5=ipi|\\
\lstinline[language=IslaLitmusExp]|1:R7=ipi|\\
\lstinline[language=IslaLitmusExp]|1:R11=ttbr(vmid=0x0001,base=vm2_stage2)|\\
\end{minipage}
\end{lrbox}
\begin{tabular}{|l|l|}
  \multicolumn{2}{l}{\textbf{AArch64} \lstinline[language=IslaLitmusName]|pKVM.vcpu_run.update_vmid.concurrent|}\\
  \hline
  \multirow{6}{*}{\begin{minipage}{\vcpurunupdatevmidconcurrentP}\vphantom{$\vcenter{\hbox{\rule{0pt}{1.8em}}}$}Page table setup:\\\usebox{\vcpurunupdatevmidconcurrentO}\end{minipage}} & \cellcolor{IslaInitialState}{\usebox{\vcpurunupdatevmidconcurrentA}}\\
  \cline{2-2}
  & Thread 0\\
  \cline{2-2}
  & \usebox{\vcpurunupdatevmidconcurrentI}\\
  \cline{2-2}
  & Thread 1\\
  \cline{2-2}
  & \usebox{\vcpurunupdatevmidconcurrentK}\\
  \cline{2-2}
  & thread1 el2 handler\\
  \cline{2-2}
  & \usebox{\vcpurunupdatevmidconcurrentM}\\
  \cline{2-2}
  & Final state: \lstinline[language=IslaLitmusExp]|0:R4=2 & 1:R0=0 & 1:R4=1 & 1:R8=0 & 1:R2=0|\\
  \hline
\end{tabular}
}
\caption{}\label{fig:vcpurunupdatevmidconcurrent}
\end{figure}

\clearpage
\testpara{pKVM.vcpu\_run.same\_vm}

For another example, for pKVM to maintain the illusion that each vCPU is on its own core,
the per-core state must be cleaned between running different vCPUs,
including ensuring that translations for one vCPU are not cached and visible
to another, even if they happen to be in the same VM (and using the same VMID) (\cc{pKVM.vcpu\_run.same\_vm} in Figure~\ref{fig:vcpurunsamevm}).

\begin{figure}
  \noindent\hspace*{-19mm}
    \resizebox*{!}{0.95\textheight}{\newlength{\vcpurunsamevmH}
\setlength{\vcpurunsamevmH}{0cm}
\newsavebox{\vcpurunsamevmI}
\begin{lrbox}{\vcpurunsamevmI}
\begin{lstlisting}[language=AArch64,showlines=true]
// arm64/include/asm/kvm_mmu.h:276
MSR VTTBR_EL2,X4
// kvm/hyp/nvhe/tlb.c:43
isb
// kvm/hyp/nvhe/tlb.c:135
tlbi vmalle1
// kvm/hyp/nvhe/tlb.c:137
dsb nsh
// kvm/hyp/nvhe/tlb.c:138
isb
// arm64/include/asm/kvm_mmu.h:276
MSR VTTBR_EL2,X5
// kvm/hyp/nvhe/tlb.c:52
isb
// kvm/hyp/sysreg-sr.h:96
MSR TTBR0_EL1,X0
// include/asm/kvm_mmu.h:276
MSR VTTBR_EL2,X1
// kvm/hyp/nvhe/host.S
ERET
L0:
// in guest
LDR X2,[X3]
\end{lstlisting}
\end{lrbox}
\newlength{\vcpurunsamevmJ}
\settowidth{\vcpurunsamevmJ}{\usebox{\vcpurunsamevmI}}
\addtolength{\vcpurunsamevmH}{\vcpurunsamevmJ}
\newsavebox{\vcpurunsamevmK}
\begin{lrbox}{\vcpurunsamevmK}
\begin{lstlisting}[language=AArch64,showlines=true]
0x1200:
mov x2, #0
// data abort preferred-return-address is itself
// so jump to next instr instead
mrs x20,elr_el2
add x20,x20,#4
msr elr_el2,x20
eret
\end{lstlisting}
\end{lrbox}
\newlength{\vcpurunsamevmL}
\settowidth{\vcpurunsamevmL}{\usebox{\vcpurunsamevmK}}
\addtolength{\vcpurunsamevmH}{\vcpurunsamevmL}
\newsavebox{\vcpurunsamevmM}
\begin{lrbox}{\vcpurunsamevmM}
\begin{lstlisting}[language=IslaPageTableSetup,showlines=true]
option default_tables = false;
physical pa1 pa2;
intermediate ipa1 ipa2;

s1table hyp_map 0x200000  {
    identity 0x1000 with code;
    x |-> invalid;
}

s2table vm1_stage2 0x300000  {
    ipa1 |-> pa1;
    ipa1 ?-> invalid;
    ipa2 |-> invalid;
    ipa2 ?-> pa2;

    s1table vm1_stage1 0x260000 {
        x |-> ipa1;
    }
}

s2table vm2_stage2 0x380000  {
    ipa1 |-> invalid;
    ipa1 ?-> pa1;
    ipa2 |-> pa2;
    ipa2 ?-> invalid;

    s1table vm2_stage1 0x2C0000 {
        x |-> ipa2;
    }
}

*pa2 = 1;
\end{lstlisting}
\end{lrbox}
\newlength{\vcpurunsamevmN}
\settowidth{\vcpurunsamevmN}{\usebox{\vcpurunsamevmM}}
\newsavebox{\vcpurunsamevmO}
\begin{lrbox}{\vcpurunsamevmO}
\begin{minipage}{0.7\vcpurunsamevmH}
\vphantom{$\vcenter{\hbox{\rule{0pt}{1.8em}}}$}Initial state:\\
\lstinline[language=IslaLitmusExp]|R3=x|\\
\lstinline[language=IslaLitmusExp]|R5=ttbr(base=0b0,vmid=0x0000)|\\
\lstinline[language=IslaLitmusExp]|TTBR0_EL1=ttbr(base=vm1_stage1,asid=0x0000)|\\
\lstinline[language=IslaLitmusExp]|SPSR_EL2=0b00101|\\
\lstinline[language=IslaLitmusExp]|TTBR0_EL2=ttbr(asid=0x0000,base=hyp_map)|\\
\lstinline[language=IslaLitmusExp]|ELR_EL2=L0:|\\
\lstinline[language=IslaLitmusExp]|R1=ttbr(vmid=0x0001,base=vm2_stage2)|\\
\lstinline[language=IslaLitmusExp]|R0=ttbr(asid=0x0000,base=vm2_stage1)|\\
\lstinline[language=IslaLitmusExp]|R4=ttbr(base=0b0,vmid=0x0001)|\\
\lstinline[language=IslaLitmusExp]|VBAR_EL2=0x1000|\\
\lstinline[language=IslaLitmusExp]|VTTBR_EL2=ttbr(vmid=0x0001,base=vm1_stage2)|\\
\lstinline[language=IslaLitmusExp]|PSTATE.EL=0b10|\\
\end{minipage}
\end{lrbox}
\begin{tabular}{|l|l|}
  \multicolumn{2}{l}{\textbf{AArch64} \lstinline[language=IslaLitmusName]|pKVM.vcpu_run.same_vm|}\\
  \hline
  \multirow{6}{*}{\begin{minipage}{\vcpurunsamevmN}\vphantom{$\vcenter{\hbox{\rule{0pt}{1.8em}}}$}Page table setup:\\\usebox{\vcpurunsamevmM}\end{minipage}} & \cellcolor{IslaInitialState}{\usebox{\vcpurunsamevmO}}\\
  \cline{2-2}
  & Thread 0\\
  \cline{2-2}
  & \usebox{\vcpurunsamevmI}\\
  \cline{2-2}
  & thread0 el2 handler\\
  \cline{2-2}
  & \usebox{\vcpurunsamevmK}\\
  \cline{2-2}
  & Final state: \lstinline[language=IslaLitmusExp]|0:R2=0|\\
  \hline
\end{tabular}
}
  \caption{}\label{fig:vcpurunsamevm}
\end{figure}

\clearpage
\subsection{Data Aborts}

As mentioned earlier,
a vCPU accessing a location that it does not have permissions for
or which is not mapped,
results in an Abort which is handled at EL2.

\testpara{pKVM.host\_handle\_trap.stage2\_idmap.l3}
If the vCPU accesses a location which is currently un-mapped,
but should be mapped on-demand, then pKVM will make a new mapping
for that vCPU, install it into its VM's Stage~2 translation tables,
and return to the vCPU to re-try the access
(Figures~\ref{fig:hosthandletrapstagetwoidmaplthree_code} and \ref{fig:hosthandletrapstagetwoidmaplthree_exec}).
This action typically requires no TLB invalidation, as the previously unmapped entries
could not have been stored in any TLB.

An added complexity here is that if pKVM wishes
to map only a single page then it \emph{must} install this mapping
with a Level~3 entry.
If the mapping is currently invalid at Level~2 or Level~1 then
care must be taken to not create entries out-of-order.
pKVM manages this by producing a fresh table of invalid entries first,
then installing that new table into the translation tables,
and recursing down until it reaches the level it needs to install the valid mapping at.
Otherwise, a concurrent translation could see the new table entries before
the leaf entries had been installed into the table.

\begin{figure}
    \noindent\hspace*{-19mm}
{\newlength{\hosthandletrapstagetwoidmaplthreeH}
\setlength{\hosthandletrapstagetwoidmaplthreeH}{0cm}
\newsavebox{\hosthandletrapstagetwoidmaplthreeI}
\begin{lrbox}{\hosthandletrapstagetwoidmaplthreeI}
\begin{lstlisting}[language=AArch64,showlines=true]
LDR X0,[X1]
LDR X2,[X3]
\end{lstlisting}
\end{lrbox}
\newlength{\hosthandletrapstagetwoidmaplthreeJ}
\settowidth{\hosthandletrapstagetwoidmaplthreeJ}{\usebox{\hosthandletrapstagetwoidmaplthreeI}}
\addtolength{\hosthandletrapstagetwoidmaplthreeH}{\hosthandletrapstagetwoidmaplthreeJ}
\newsavebox{\hosthandletrapstagetwoidmaplthreeK}
\begin{lrbox}{\hosthandletrapstagetwoidmaplthreeK}
\begin{lstlisting}[language=AArch64,showlines=true]
0x1400:
// count number of exceptions
add x10,x10,#1
// remember which IPA failed
mrs x9,hpfar_el2
lsl x9,x9,#8
// pkvm code
stlr x11,[x12]
dsb ishst
1:
// return to next instruction
// pKVM doesn't really do this, it just tries the same instr again
// but without this the test can loop forever ...
mrs x20,elr_el2
add x20,x20,#4
msr elr_el2,x20
// return from handle_trap
eret
\end{lstlisting}
\end{lrbox}
\newlength{\hosthandletrapstagetwoidmaplthreeL}
\settowidth{\hosthandletrapstagetwoidmaplthreeL}{\usebox{\hosthandletrapstagetwoidmaplthreeK}}
\addtolength{\hosthandletrapstagetwoidmaplthreeH}{\hosthandletrapstagetwoidmaplthreeL}
\newsavebox{\hosthandletrapstagetwoidmaplthreeM}
\begin{lrbox}{\hosthandletrapstagetwoidmaplthreeM}
\begin{lstlisting}[language=IslaPageTableSetup,showlines=true]
option default_tables = false;
physical pa1;
intermediate ipa1;

s2table vm_stage2 0x260000 {
    ipa1 |-> invalid;
    ipa1 ?-> pa1;

    s1table host_s1 0x2C0000 {
        x |-> ipa1;
    }
}

s1table hyp_map 0x200000 {
    x |-> invalid;
    s1table host_s1;
    s2table vm_stage2;
    identity 0x1000 with code;
}

*pa1 = 1;
\end{lstlisting}
\end{lrbox}
\newlength{\hosthandletrapstagetwoidmaplthreeN}
\settowidth{\hosthandletrapstagetwoidmaplthreeN}{\usebox{\hosthandletrapstagetwoidmaplthreeM}}
\newsavebox{\hosthandletrapstagetwoidmaplthreeO}
\begin{lrbox}{\hosthandletrapstagetwoidmaplthreeO}
\begin{minipage}{0.8\hosthandletrapstagetwoidmaplthreeH}
\vphantom{$\vcenter{\hbox{\rule{0pt}{1.8em}}}$}Initial state:\\
\lstinline[language=IslaLitmusExp]|R1=x|\\
\lstinline[language=IslaLitmusExp]|R12=pte3(ipa1,vm_stage2)|\\
\lstinline[language=IslaLitmusExp]|R3=x|\\
\lstinline[language=IslaLitmusExp]|TTBR0_EL1=ttbr(asid=0x0000,base=host_s1)|\\
\lstinline[language=IslaLitmusExp]|TTBR0_EL2=ttbr(base=hyp_map,asid=0x0000)|\\
\lstinline[language=IslaLitmusExp]|R11=mkdesc3(oa=pa1)|\\
\lstinline[language=IslaLitmusExp]|R10=0b0|\\
\lstinline[language=IslaLitmusExp]|PSTATE.EL=0b01|\\
\lstinline[language=IslaLitmusExp]|VTTBR_EL2=ttbr(base=vm_stage2,vmid=0x0000)|\\
\lstinline[language=IslaLitmusExp]|VBAR_EL2=0x1000|\\
\end{minipage}
\end{lrbox}
\begin{tabular}{|l|l|}
  \multicolumn{2}{l}{\textbf{AArch64} \lstinline[language=IslaLitmusName]|pKVM.host_handle_trap.stage2_idmap.l3|}\\
  \hline
  \multirow{6}{*}{\begin{minipage}{\hosthandletrapstagetwoidmaplthreeN}\vphantom{$\vcenter{\hbox{\rule{0pt}{1.8em}}}$}Page table setup:\\\usebox{\hosthandletrapstagetwoidmaplthreeM}\end{minipage}} & \cellcolor{IslaInitialState}{\usebox{\hosthandletrapstagetwoidmaplthreeO}}\\
  \cline{2-2}
  & Thread 0\\
  \cline{2-2}
  & \usebox{\hosthandletrapstagetwoidmaplthreeI}\\
  \cline{2-2}
  & thread0 el2 handler\\
  \cline{2-2}
  & \usebox{\hosthandletrapstagetwoidmaplthreeK}\\
  \cline{2-2}
  & Final state: \lstinline[language=IslaLitmusExp]|0:R10=2|\\
  \hline
\end{tabular}
}\
  \caption{\PKVMTEST{pKVM.host\_handle\_trap.stage2\_idmap.l3}: code listing}\label{fig:hosthandletrapstagetwoidmaplthree_code}
\end{figure}

\begin{figure}
  \centering

  \includegraphics[trim=10mm 0 0 0, clip,width=140mm,height=120mm]{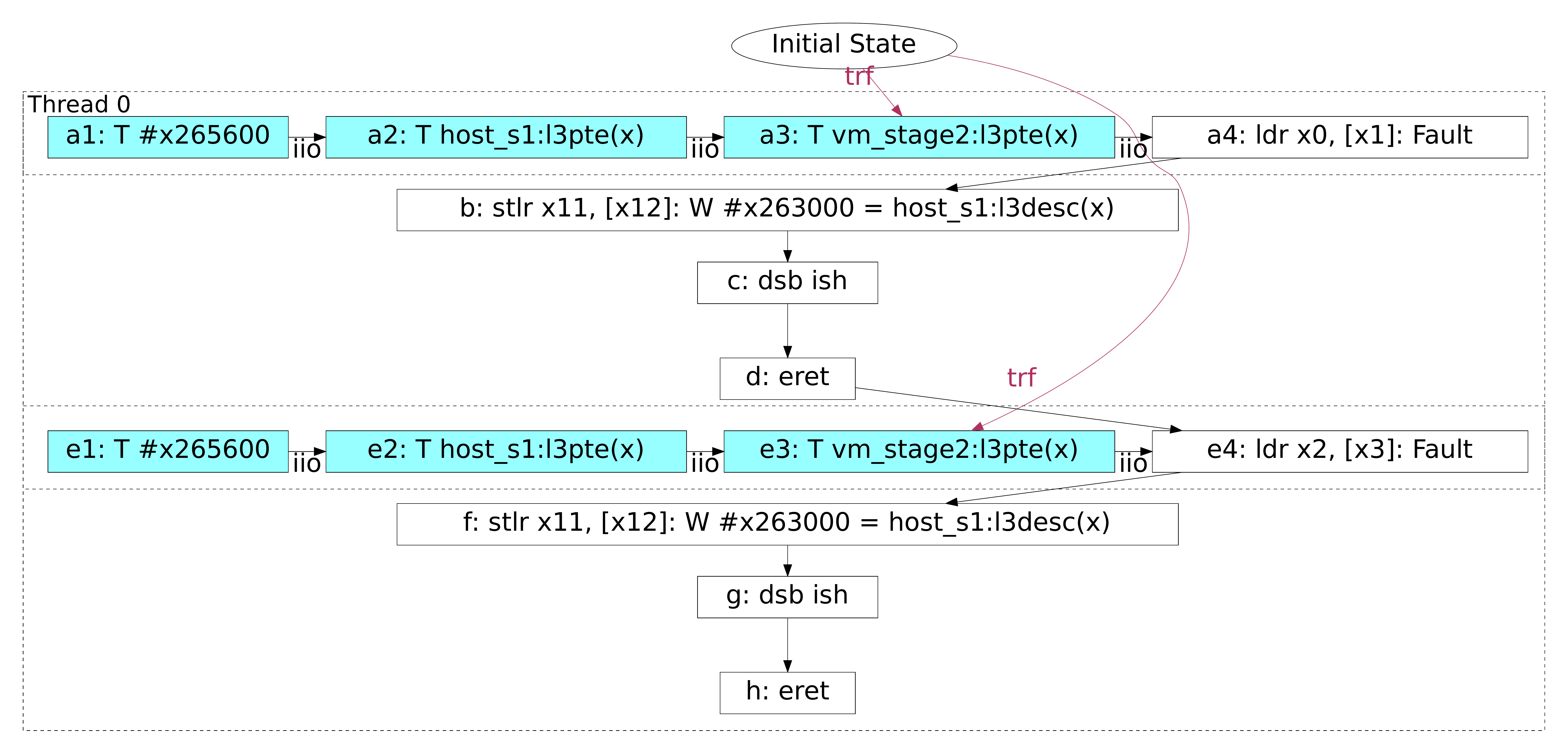}
  \caption{\PKVMTEST{pKVM.host\_handle\_trap.stage2\_idmap.l3}: forbidden candidate execution}\label{fig:hosthandletrapstagetwoidmaplthree_exec}
\end{figure}

\clearpage
\testpara{pKVM.host\_handle\_trap.stage2\_idmap.already\_exists}\label{test:pKVM.host_handle_trap.stage2_idmap.already_exists}
If the vCPU tries to access a location that is mapped but that it lacks the necessary permissions for,
then there may be an existing entry already mapping that location.
Here, the process is much more delicate.
Because another vCPU in the same VM may be concurrently accessing the same physical locations
mapped by the shared Stage~2 table,
pKVM cannot assume its own internal locks are sufficient to prevent race conditions.
Therefore, this is one case where pKVM is \emph{required} to use the break-before-make sequence as described by
\S\ref{subsubsec:questions:tlb:bbm}.

In the following code (Figures~\ref{fig:hosthandletrapstagetwoidmaplthreealreadyexists_code} and \ref{fig:hosthandletrapstagetwoidmaplthreealreadyexists_exec}), the initial state starts out with \cc{x} mapped at level 3, but without read permissions.
The first load will fault and pKVM will naturally map a writeable mapping on-demand,
following the break-before-make sequence invalidating the entry by writing zero (event \herd{b})
before writing a new valid descriptor (event \herd{m}).
After returning to the guest, we ask whether loads of \cc{x} are allowed to fault at Stage~2.

\begin{figure}
    \noindent\hspace*{-19mm}
    \resizebox*{!}{0.95\textheight}{\newlength{\hosthandletrapstagetwoidmaplthreealreadyexistsH}
\setlength{\hosthandletrapstagetwoidmaplthreealreadyexistsH}{0cm}
\newsavebox{\hosthandletrapstagetwoidmaplthreealreadyexistsI}
\begin{lrbox}{\hosthandletrapstagetwoidmaplthreealreadyexistsI}
\begin{lstlisting}[language=AArch64,showlines=true]
STR X0,[X1]
LDR X2,[X3]
\end{lstlisting}
\end{lrbox}
\newlength{\hosthandletrapstagetwoidmaplthreealreadyexistsJ}
\settowidth{\hosthandletrapstagetwoidmaplthreealreadyexistsJ}{\usebox{\hosthandletrapstagetwoidmaplthreealreadyexistsI}}
\addtolength{\hosthandletrapstagetwoidmaplthreealreadyexistsH}{\hosthandletrapstagetwoidmaplthreealreadyexistsJ}
\newsavebox{\hosthandletrapstagetwoidmaplthreealreadyexistsK}
\begin{lrbox}{\hosthandletrapstagetwoidmaplthreealreadyexistsK}
\begin{lstlisting}[language=AArch64,showlines=true]
0x1400:
// count number of exceptions
add x10,x10,#1
// remember which IPA failed
mrs x9,hpfar_el2
// really pKVM would do a read of the pagetable to tell if it needs to update
// we elide that code from the test and just skip on the second fault
cmp x10,#2
b.eq 1f
// pkvm code
// pgtable.c:573
STR X12,[X13]
// tlb.c:63
DSB ISH
// tlb.c:66
MSR VTTBR_EL2,X14
// tlb.c:66
ISB
// tlb.c:74
TLBI IPAS2E1IS,X15
// tlb.c:82
DSB ISH
// tlb.c:83
TLBI VMALLE1IS
// tlb.c:84
DSB ISH
// tlb.c:85
ISB
// tlb.c:109
MSR VTTBR_EL2,X16
// tlb.c:109
ISB
STLR X17,[x18]
DSB ISHST
// return to next instruction
// pKVM doesn't really do this, it just tries the same instr again
// but without this the test can loop forever ...
1:
mrs x20,elr_el2
add x20,x20,#4
msr elr_el2,x20
// return from handle_trap
eret
\end{lstlisting}
\end{lrbox}
\newlength{\hosthandletrapstagetwoidmaplthreealreadyexistsL}
\settowidth{\hosthandletrapstagetwoidmaplthreealreadyexistsL}{\usebox{\hosthandletrapstagetwoidmaplthreealreadyexistsK}}
\addtolength{\hosthandletrapstagetwoidmaplthreealreadyexistsH}{\hosthandletrapstagetwoidmaplthreealreadyexistsL}
\newsavebox{\hosthandletrapstagetwoidmaplthreealreadyexistsM}
\begin{lrbox}{\hosthandletrapstagetwoidmaplthreealreadyexistsM}
\begin{lstlisting}[language=IslaPageTableSetup,showlines=true]
option default_tables = false;
physical pa1 pa2;
intermediate ipa1;

s2table vm_stage2 0x260000  {
    ipa1 |-> pa1 with [AP=0b00] and default;
    ipa1 ?-> invalid;
    ipa1 ?-> pa2;

    s1table host_s1 0x2C0000 {
        x |-> ipa1;
    }
}

s1table hyp_map 0x200000  {
    x |-> invalid;
    s2table vm_stage2;
    identity 0x1000 with code;
}

*pa1 = 1;
*pa2 = 2;
\end{lstlisting}
\end{lrbox}
\newlength{\hosthandletrapstagetwoidmaplthreealreadyexistsN}
\settowidth{\hosthandletrapstagetwoidmaplthreealreadyexistsN}{\usebox{\hosthandletrapstagetwoidmaplthreealreadyexistsM}}
\newsavebox{\hosthandletrapstagetwoidmaplthreealreadyexistsO}
\begin{lrbox}{\hosthandletrapstagetwoidmaplthreealreadyexistsO}
\begin{minipage}{\hosthandletrapstagetwoidmaplthreealreadyexistsH}
\vphantom{$\vcenter{\hbox{\rule{0pt}{1.8em}}}$}Initial state:\\
\lstinline[language=IslaLitmusExp]|R1=x|\\
\lstinline[language=IslaLitmusExp]|R16=ttbr(vmid=0x001,base=vm_stage2)|\\
\lstinline[language=IslaLitmusExp]|R17=mkdesc3(oa=pa2)|\\
\lstinline[language=IslaLitmusExp]|R0=0b0|\\
\lstinline[language=IslaLitmusExp]|PSTATE.EL=0b01|\\
\lstinline[language=IslaLitmusExp]|R18=pte3(ipa1,vm_stage2)|\\
\lstinline[language=IslaLitmusExp]|VTTBR_EL2=ttbr(vmid=0x001,base=vm_stage2)|\\
\lstinline[language=IslaLitmusExp]|R10=0b0|\\
\lstinline[language=IslaLitmusExp]|R15=page(ipa1)|\\
\lstinline[language=IslaLitmusExp]|R13=pte3(ipa1,vm_stage2)|\\
\lstinline[language=IslaLitmusExp]|R14=ttbr(vmid=0x001,base=0b0)|\\
\lstinline[language=IslaLitmusExp]|R12=0b0|\\
\lstinline[language=IslaLitmusExp]|R3=x|\\
\lstinline[language=IslaLitmusExp]|TTBR0_EL1=ttbr(asid=0x000,base=host_s1)|\\
\lstinline[language=IslaLitmusExp]|VBAR_EL2=0x1000|\\
\lstinline[language=IslaLitmusExp]|TTBR0_EL2=ttbr(asid=0x000,base=hyp_map)|\\
\end{minipage}
\end{lrbox}
\begin{tabular}{|l|l|}
  \multicolumn{2}{l}{\textbf{AArch64} \lstinline[language=IslaLitmusName]|pKVM.host_handle_trap.stage2_idmap.l3.already_exists|}\\
  \hline
  \multirow{6}{*}{\begin{minipage}{\hosthandletrapstagetwoidmaplthreealreadyexistsN}\vphantom{$\vcenter{\hbox{\rule{0pt}{1.8em}}}$}Page table setup:\\\usebox{\hosthandletrapstagetwoidmaplthreealreadyexistsM}\end{minipage}} & \cellcolor{IslaInitialState}{\usebox{\hosthandletrapstagetwoidmaplthreealreadyexistsO}}\\
  \cline{2-2}
  & Thread 0\\
  \cline{2-2}
  & \usebox{\hosthandletrapstagetwoidmaplthreealreadyexistsI}\\
  \cline{2-2}
  & thread0 el2 handler\\
  \cline{2-2}
  & \usebox{\hosthandletrapstagetwoidmaplthreealreadyexistsK}\\
  \cline{2-2}
  & Final state: \lstinline[language=IslaLitmusExp]|0:R10=2 | 0:R2=1|\\
  \hline
\end{tabular}
}\
    \caption{\PKVMTEST{pKVM.host\_handle\_trap.stage2\_idmap.l3.already\_exists}: code listing}\label{fig:hosthandletrapstagetwoidmaplthreealreadyexists_code}
  \end{figure}

  \begin{figure}
    \centering
    \includegraphics[trim=10mm 0 0 0, clip,width=140mm,height=120mm]{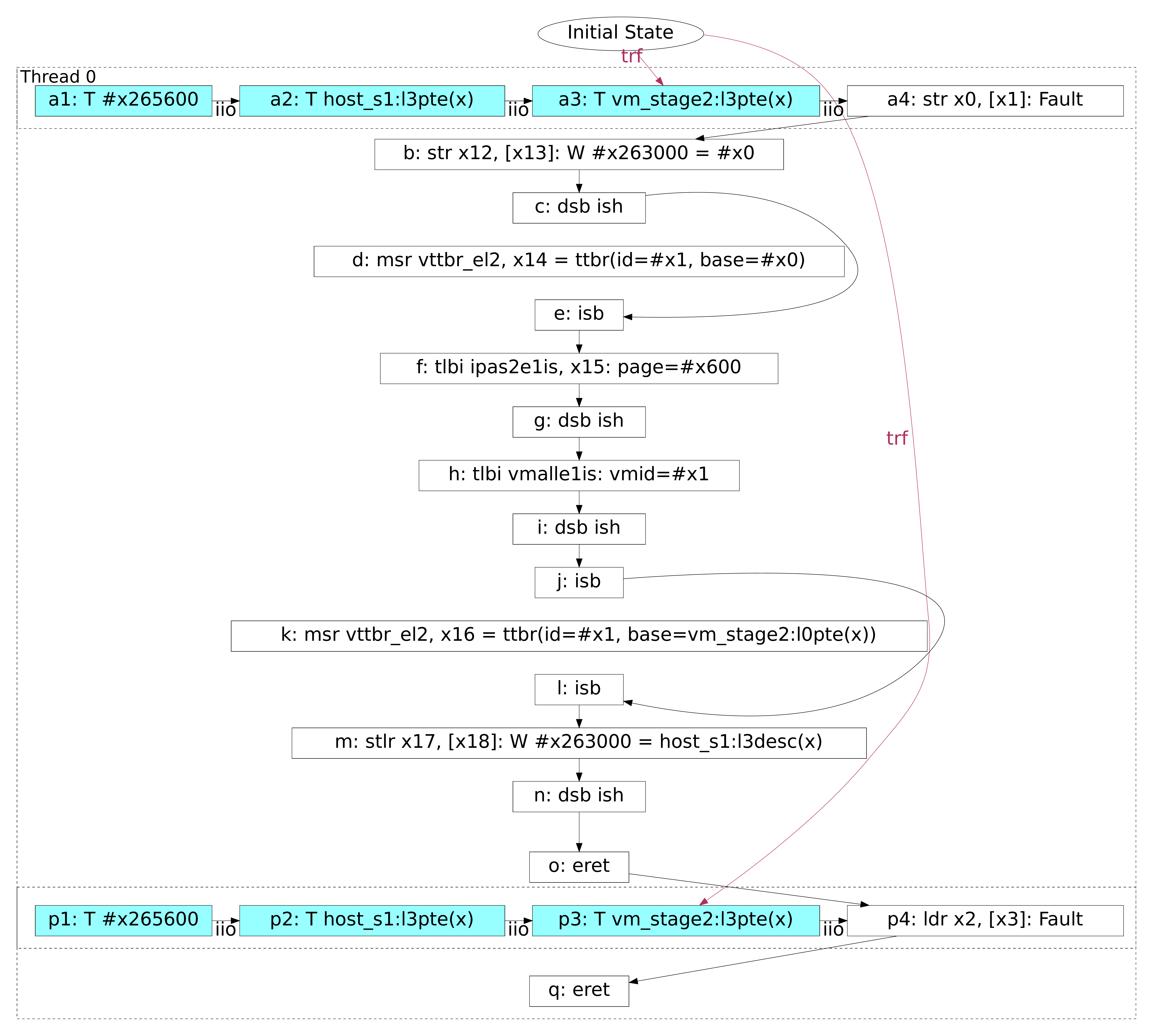}
    \caption{\PKVMTEST{pKVM.host\_handle\_trap.stage2\_idmap.l3.already\_exists}: forbidden candidate execution}\label{fig:hosthandletrapstagetwoidmaplthreealreadyexists_exec}
  \end{figure}

\clearpage

\testpara{pKVM.host\_handle\_trap.stage2\_idmap.change\_block\_size}
The complexity of the previous scenario is compounded by the fact that pKVM may wish to map a larger region
(higher in the table)
than is currently mapped, and, without \asm{FEAT\_BBM}, this adds extra break-before-make requirements.

In general, pKVM will itself perform a translation table walk.
On the way down, it will look for the entry to be replaced,
invalidate it, and perform TLB maintenance,
ensuring that all entries from old leaf entries below it are cleaned away,
but additionally that any old stale Stage~1 translations are invalidated,
before it replaces the entry.

Figures~\ref{fig:hosthandletrapstagetwoidmapchangeblocksize_code} and \ref{fig:hosthandletrapstagetwoidmapchangeblocksize_exec} contain the code listing and the interesting candidate execution diagram
(as generated by Isla).
We consider a scenario where there are two locations,
\cc{x} and \cc{y}
which are in the same 2~MiB region of memory
but are mapped by different Level~3 entries.
If the Level~3 table contains all valid entries except for
one entry for \cc{x},
then on updating the entry for \cc{x} to be valid,
if the host kernel maps the whole 2M region,
then pKVM will invalidate the 2M entry before
writing a new 2M block entry.

This in effect `promotes' the set of 4K mappings into a single 2M mapping.
pKVM then frees the child table to be re-allocated later.
pKVM may want to \emph{remove} a mapping for an IPA.
This is very similar to the previously described break-before-make scenario,
but without the final make.
pKVM just has to ensure that the old IPA mapping is invalidated,
and the necessary TLB maintenance is performed.
We do not include this test at present.

\begin{figure}
  \noindent\hspace*{-19mm}
  \resizebox*{!}{0.95\textheight}{\newlength{\hosthandletrapstagetwoidmapchangeblocksizeH}
\setlength{\hosthandletrapstagetwoidmapchangeblocksizeH}{0cm}
\newsavebox{\hosthandletrapstagetwoidmapchangeblocksizeI}
\begin{lrbox}{\hosthandletrapstagetwoidmapchangeblocksizeI}
\begin{lstlisting}[language=AArch64,showlines=true]
// in guest
MOV X0,#0
LDR X0,[X1]
MOV X2,#0
LDR X2,[X3]
\end{lstlisting}
\end{lrbox}
\newlength{\hosthandletrapstagetwoidmapchangeblocksizeJ}
\settowidth{\hosthandletrapstagetwoidmapchangeblocksizeJ}{\usebox{\hosthandletrapstagetwoidmapchangeblocksizeI}}
\addtolength{\hosthandletrapstagetwoidmapchangeblocksizeH}{\hosthandletrapstagetwoidmapchangeblocksizeJ}
\newsavebox{\hosthandletrapstagetwoidmapchangeblocksizeK}
\begin{lrbox}{\hosthandletrapstagetwoidmapchangeblocksizeK}
\begin{lstlisting}[language=AArch64,showlines=true]
0x1400:
// count number of exceptions
add x10,x10,#1
// remember which IPA failed
mrs x9,hpfar_el2
// on second fault just exit test
cmp x10,#2
b.eq 1f
// then pkvm code
0:
// pgtable.c:181
STR X12,[X13]
// tlb.c:116
DSB ISH
// tlb.c:119
MSR VTTBR_EL2,X14
// tlb.c:119
ISB
// tlb.c:121
TLBI vmalls12e1is
// tlb.c:122
DSB ISH
// tlb.c:123
ISB
// tlb.c:125
MSR VTTBR_EL2,X16
// tlb.c:125
ISB
STLR X17,[x18]
DSB ISHST
1:
// return to next instruction
// pKVM doesn't really do this, it just tries the same instr again
// but without this the test can loop forever ...
mrs x20,elr_el2
add x20,x20,#4
msr elr_el2,x20
// return from handle_trap
eret
\end{lstlisting}
\end{lrbox}
\newlength{\hosthandletrapstagetwoidmapchangeblocksizeL}
\settowidth{\hosthandletrapstagetwoidmapchangeblocksizeL}{\usebox{\hosthandletrapstagetwoidmapchangeblocksizeK}}
\addtolength{\hosthandletrapstagetwoidmapchangeblocksizeH}{\hosthandletrapstagetwoidmapchangeblocksizeL}
\newsavebox{\hosthandletrapstagetwoidmapchangeblocksizeM}
\begin{lrbox}{\hosthandletrapstagetwoidmapchangeblocksizeM}
\begin{lstlisting}[language=IslaPageTableSetup,showlines=true]
option default_tables = false;
physical pa1 pa2;
intermediate ipa1 ipa2;

s2table vm_stage2 0x260000  {
    ipa1 |-> invalid at level 3;
    ipa2 |-> pa2 at level 3;
    ipa1 ?-> pa1 at level 3;

    ipa1 ?-> invalid at level 2;
    ipa2 ?-> invalid at level 2;
    ipa2 ?-> pa2 at level 2;

    s1table host_s1 0x2C0000 {
        x |-> ipa1;
        y |-> ipa2;
    }
};

s1table hyp_map 0x200000  {
    x |-> invalid;
    s2table vm_stage2;
    identity 0x1000 with code;
}

*pa1 = 1;
*pa2 = 2;
\end{lstlisting}
\end{lrbox}
\newlength{\hosthandletrapstagetwoidmapchangeblocksizeN}
\settowidth{\hosthandletrapstagetwoidmapchangeblocksizeN}{\usebox{\hosthandletrapstagetwoidmapchangeblocksizeM}}
\newsavebox{\hosthandletrapstagetwoidmapchangeblocksizeO}
\begin{lrbox}{\hosthandletrapstagetwoidmapchangeblocksizeO}
\begin{minipage}{\hosthandletrapstagetwoidmapchangeblocksizeH}
\vphantom{$\vcenter{\hbox{\rule{0pt}{1.8em}}}$}Initial state:\\
\lstinline[language=IslaLitmusExp]|R15=ipa1|\\
\lstinline[language=IslaLitmusExp]|R13=pte2(ipa1,vm_stage2)|\\
\lstinline[language=IslaLitmusExp]|R17=mkdesc2(oa=pa2)|\\
\lstinline[language=IslaLitmusExp]|VTTBR_EL2=ttbr(vmid=0x0001,base=vm_stage2)|\\
\lstinline[language=IslaLitmusExp]|R12=0b0|\\
\lstinline[language=IslaLitmusExp]|R16=ttbr(vmid=0x001,base=vm_stage2)|\\
\lstinline[language=IslaLitmusExp]|R3=y|\\
\lstinline[language=IslaLitmusExp]|TTBR0_EL2=ttbr(base=hyp_map,asid=0x0000)|\\
\lstinline[language=IslaLitmusExp]|TTBR0_EL1=ttbr(base=host_s1,asid=0x0000)|\\
\lstinline[language=IslaLitmusExp]|R18=pte2(x,vm_stage2)|\\
\lstinline[language=IslaLitmusExp]|R14=ttbr(base=0b0,vmid=0x001)|\\
\lstinline[language=IslaLitmusExp]|R1=x|\\
\lstinline[language=IslaLitmusExp]|R10=0b0|\\
\lstinline[language=IslaLitmusExp]|VBAR_EL2=0x1000|\\
\lstinline[language=IslaLitmusExp]|PSTATE.EL=0b01|\\
\end{minipage}
\end{lrbox}
\begin{tabular}{|l|l|}
  \multicolumn{2}{l}{\textbf{AArch64} \lstinline[language=IslaLitmusName]|pKVM.host_handle_trap.stage2_idmap.change_block_size|}\\
  \hline
  \multirow{6}{*}{\begin{minipage}{\hosthandletrapstagetwoidmapchangeblocksizeN}\vphantom{$\vcenter{\hbox{\rule{0pt}{1.8em}}}$}Page table setup:\\\usebox{\hosthandletrapstagetwoidmapchangeblocksizeM}\end{minipage}} & \cellcolor{IslaInitialState}{\usebox{\hosthandletrapstagetwoidmapchangeblocksizeO}}\\
  \cline{2-2}
  & Thread 0\\
  \cline{2-2}
  & \usebox{\hosthandletrapstagetwoidmapchangeblocksizeI}\\
  \cline{2-2}
  & thread0 el2 handler\\
  \cline{2-2}
  & \usebox{\hosthandletrapstagetwoidmapchangeblocksizeK}\\
  \cline{2-2}
  & Final state: \lstinline[language=IslaLitmusExp]|0:R10=2 | 0:R2=1|\\
  \hline
\end{tabular}
}\
  \caption{\PKVMTEST{pKVM.host\_handle\_trap.stage2\_idmap.change\_block\_size}: code listing}\label{fig:hosthandletrapstagetwoidmapchangeblocksize_code}
\end{figure}

\begin{figure}
  \centering
  \includegraphics[trim=10mm 0 0 0, clip,width=140mm,height=120mm]{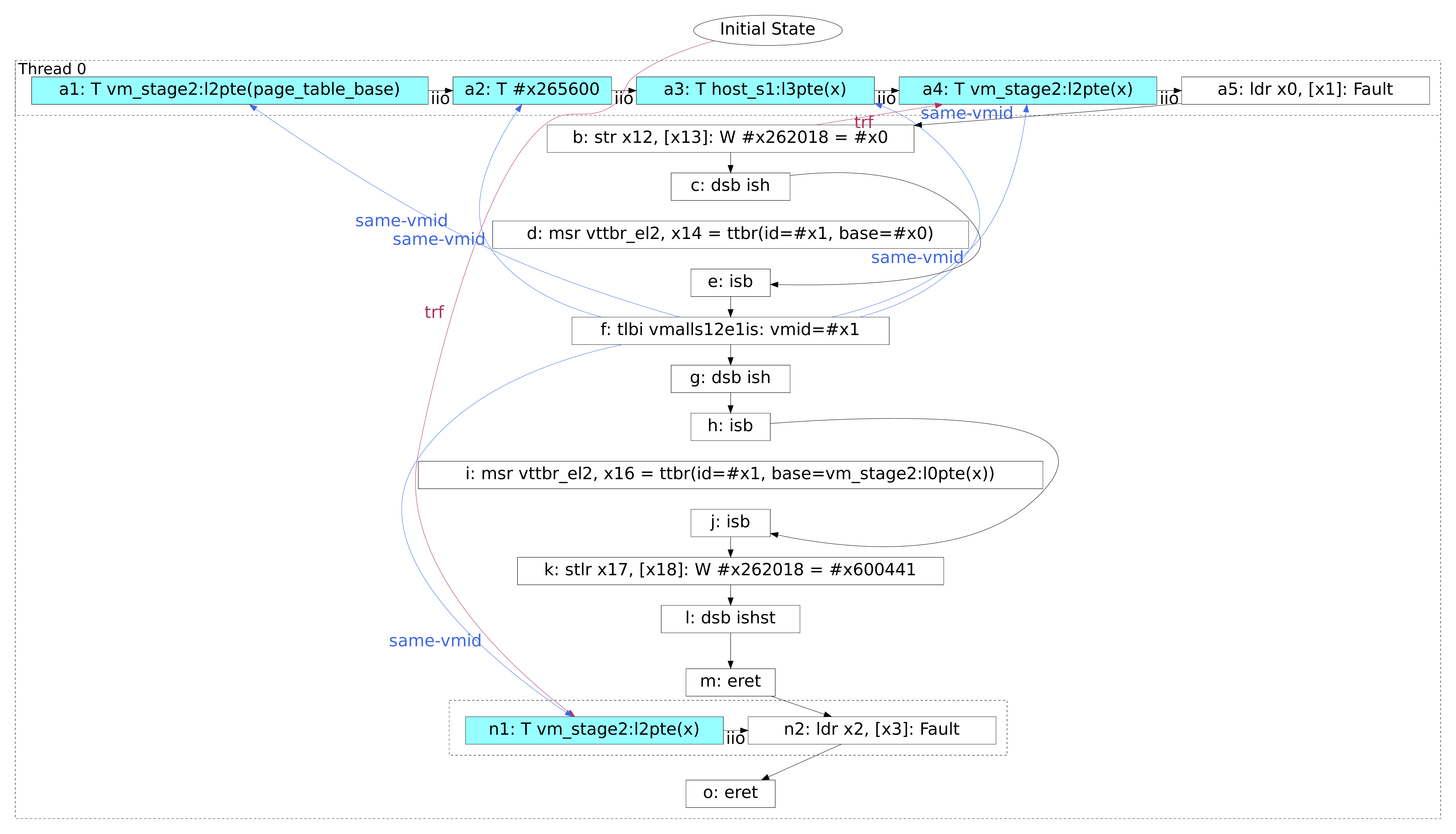}
  \caption{\PKVMTEST{pKVM.host\_handle\_trap.stage2\_idmap.change\_block\_size}: forbidden candidate execution}\label{fig:hosthandletrapstagetwoidmapchangeblocksize_exec}
\end{figure}

\clearpage

\subsection{Initialisation}

During execution,
and especially at initial start-up of pKVM,
it will be required to create its own Stage~1 translation tables.

Currently, this only happens on a single core,
so there are no interesting concurrent cases.

In the future, however, it is expected that pKVM may need to
dynamically map some parts of memory, as the hypervisor gains richer features.
For example, it may need to map some shared page between guests, or between pKVM and a guest, in which to pass messages.

\clearpage

\testpara{pKVM.switch\_to\_new\_table}

When pKVM first starts, it is using translation tables set up by Linux,
so one of the first things pKVM does is to create its own tables and switch to them.

Usually, switching from using one translation table to another happens at a higher exception level,
and then the new one is not used until the return to the lower exception level.
The current case
 is more complicated, as pKVM has to change its own translation tables
\emph{while it is executing}.
If pKVM only had to change the \asm{TTBR}, then this would not be a problem,
but there are many system registers involved in configuring the translation tables
(the \asm{TCR}, \asm{MAIR}, and so on), and these registers cannot all be updated `atomically'.

To maintain the atomicity, pKVM switches the \asm{TTBR}
to a page of memory where the code it is executing is identity-mapped,
then it disables translation (disabling the MMU),
before updating all of the required system registers (including the new \asm{TTBR}),
before re-enabling the MMU.

Figure~\ref{fig:switchtonewtable_code} contains the code for a litmus test
that tries to capture the core of this process.
This \PKVMTEST{pKVM.switch\_to\_new\_table} test is an EL2 Stage~1 test
with two tables, \cc{hyp\_pgtable} and \cc{new\_hyp\_pgtable}, in memory.
The code disables the MMU (by writing to the appropriate field of the SCTLR),
invalidates all of the old cached TLB entries for EL2,
updates all the system registers (only the \asm{TTBR} is included here)
and then re-enables the MMU.
We check that the final load in the hypervisor after the switch has happened
reads using the new state.

This test has been cut down for brevity to remove the writes of other system registers,
which would be present in a full execution of pKVM;
ideally, we would have a test that included those too,
and a final state which ensured that the load used a translation using
all of the new register values.

Note that the model presented in Section~\ref{sec:models}
does not currently contain axioms for when the MMU is disabled,
but the semantics seem clear, %
and we do not see any impediment in extending the model to handle it fully.

\begin{figure}[h]
  \centering
  \scalebox{0.8}{\newlength{\switchtonewtableH}
\setlength{\switchtonewtableH}{0cm}
\newsavebox{\switchtonewtableI}
\begin{lrbox}{\switchtonewtableI}
\begin{lstlisting}[language=AArch64,showlines=true]
// hyp-init.S:247
mrs x2, sctlr_el2
bic x3, x2, #1
msr sctlr_el2, x3
isb
tlbi alle2
msr ttbr0_el2, x4
msr sctlr_el2, x2
isb
ldr x5,[x6]
\end{lstlisting}
\end{lrbox}
\newlength{\switchtonewtableJ}
\settowidth{\switchtonewtableJ}{\usebox{\switchtonewtableI}}
\addtolength{\switchtonewtableH}{\switchtonewtableJ}
\newsavebox{\switchtonewtableK}
\begin{lrbox}{\switchtonewtableK}
\begin{lstlisting}[language=AArch64,showlines=true]
0x2200:
mov x5, #0
\end{lstlisting}
\end{lrbox}
\newlength{\switchtonewtableL}
\settowidth{\switchtonewtableL}{\usebox{\switchtonewtableK}}
\addtolength{\switchtonewtableH}{\switchtonewtableL}
\newsavebox{\switchtonewtableM}
\begin{lrbox}{\switchtonewtableM}
\begin{lstlisting}[language=IslaPageTableSetup,showlines=true]
option default_tables = false;
physical pa1;

s1table hyp_pgtable 0x200000  {
    x |-> invalid at level 3;
    x ?-> pa1;
    identity 0x2000 with code;
}

s1table new_hyp_pgtable 0x240000  {
    x |-> pa1 at level 3;
    x ?-> invalid;
    identity 0x2000 with code;
}

*pa1 = 1;
\end{lstlisting}
\end{lrbox}
\newlength{\switchtonewtableN}
\settowidth{\switchtonewtableN}{\usebox{\switchtonewtableM}}
\newsavebox{\switchtonewtableO}
\begin{lrbox}{\switchtonewtableO}
\begin{minipage}{1.4\switchtonewtableH}
\vphantom{$\vcenter{\hbox{\rule{0pt}{1.8em}}}$}Initial state:\\
\lstinline[language=IslaLitmusExp]|R6=x|\\
\lstinline[language=IslaLitmusExp]|R4=ttbr(asid=0x0000,base=new_hyp_pgtable)|\\
\lstinline[language=IslaLitmusExp]|VBAR_EL2=0x2000|\\
\lstinline[language=IslaLitmusExp]|PSTATE.SP=0b1|\\
\lstinline[language=IslaLitmusExp]|PSTATE.EL=0b10|\\
\lstinline[language=IslaLitmusExp]|TTBR0_EL2=ttbr(base=hyp_pgtable,asid=0x0000)|\\
\end{minipage}
\end{lrbox}
\begin{tabular}{|l|l|}
  \multicolumn{2}{l}{\textbf{AArch64} \lstinline[language=IslaLitmusName]|pKVM.switch_to_new_table|}\\
  \hline
  \multirow{6}{*}{\begin{minipage}{\switchtonewtableN}\vphantom{$\vcenter{\hbox{\rule{0pt}{1.8em}}}$}Page table setup:\\\usebox{\switchtonewtableM}\end{minipage}} & \cellcolor{IslaInitialState}{\usebox{\switchtonewtableO}}\\
  \cline{2-2}
  & Thread 0\\
  \cline{2-2}
  & \usebox{\switchtonewtableI}\\
  \cline{2-2}
  & thread0 el2 handler\\
  \cline{2-2}
  & \usebox{\switchtonewtableK}\\
  \cline{2-2}
  & Final state: \lstinline[language=IslaLitmusExp]|0:R2=0|\\
  \hline
\end{tabular}
}
\caption{pKVM.switch\_to\_new\_table}\label{fig:switchtonewtable_code}
\end{figure}

\clearpage

\testpara{pKVM.create\_hyp\_mappings.inv.l2}

Constructing new translation tables is done incrementally,
starting from a single zero'd page of memory as the root table,
and then performing a manual translation table walk on insertion to locate an entry to insert into.

In the case where there is no Level~3 table to install into,
pKVM first creates a Level~2 table and installs that,
before writing the new valid Level~3 entry.
To avoid break-before-make violations here,
pKVM always ensures the table is zeroed before
inserting it into the table.

The \PKVMTEST{pKVM.create\_hyp\_mappings.inv.l2} test,
in Figures~\ref{fig:createhypmappingsinvltwo_code} and \ref{fig:createhypmappingsinvltwo_exec},
gives the case where pKVM is trying to create a new mapping for itself
for a 4K page block.
This block mapping \emph{must} be installed as a Level~3 entry,
as installing it any higher would end up mapping more than 4K of memory.
Initially, \herd{x} is translated using an invalid Level~2 entry in \cc{hyp\_pgtable}.
The table \cc{hyp\_pgtable\_new} contains the new Level~3 table which starts life zeroed (all invalid);
The test then sets the Level~2 entry in \cc{hyp\_pgtable} to point to the new table,
and then updates the Level~3 (leaf) entry with a valid descriptor.

Interestingly, we note that pKVM sets these entries with a store-release;
although there seems to be no relaxed-virtual-memory reason why.
Given that pKVM is well locked, it is not clear why making these writes store-releases helps.

\begin{figure}[h]
  \scalebox{0.8}{\newlength{\createhypmappingsinvltwoH}
\setlength{\createhypmappingsinvltwoH}{0cm}
\newsavebox{\createhypmappingsinvltwoI}
\begin{lrbox}{\createhypmappingsinvltwoI}
\begin{lstlisting}[language=AArch64,showlines=true]
STLR X0,[X1]
STLR X2,[X3]
DSB SY
ISB
L0:
LDR X4,[X5]
\end{lstlisting}
\end{lrbox}
\newlength{\createhypmappingsinvltwoJ}
\settowidth{\createhypmappingsinvltwoJ}{\usebox{\createhypmappingsinvltwoI}}
\addtolength{\createhypmappingsinvltwoH}{\createhypmappingsinvltwoJ}
\newsavebox{\createhypmappingsinvltwoK}
\begin{lrbox}{\createhypmappingsinvltwoK}
\begin{lstlisting}[language=AArch64,showlines=true]
0x1200:
mov x2, #0
mrs x20,ELR_EL2
add x20,x20,#4
msr ELR_EL2,x20
eret
\end{lstlisting}
\end{lrbox}
\newlength{\createhypmappingsinvltwoL}
\settowidth{\createhypmappingsinvltwoL}{\usebox{\createhypmappingsinvltwoK}}
\addtolength{\createhypmappingsinvltwoH}{\createhypmappingsinvltwoL}
\newsavebox{\createhypmappingsinvltwoM}
\begin{lrbox}{\createhypmappingsinvltwoM}
\begin{lstlisting}[language=IslaPageTableSetup,showlines=true]
option default_tables = false;
physical pa1;

s1table hyp_pgtable_new 0x280000 {
    x |-> invalid at level 3;
    x ?-> pa1 at level 3;
}

s1table hyp_pgtable 0x200000  {
    x |-> invalid at level 2;
    x ?-> table(0x283000) at level 2;
    identity 0x1000 with code;
    s1table hyp_pgtable_new;
}

*pa1 = 1;
\end{lstlisting}
\end{lrbox}
\newlength{\createhypmappingsinvltwoN}
\settowidth{\createhypmappingsinvltwoN}{\usebox{\createhypmappingsinvltwoM}}
\newsavebox{\createhypmappingsinvltwoO}
\begin{lrbox}{\createhypmappingsinvltwoO}
\begin{minipage}{1.4\createhypmappingsinvltwoH}
\vphantom{$\vcenter{\hbox{\rule{0pt}{1.8em}}}$}Initial state:\\
\lstinline[language=IslaLitmusExp]|R0=mkdesc2(table=0x283000)|\\
\lstinline[language=IslaLitmusExp]|R2=mkdesc3(oa=pa1)|\\
\lstinline[language=IslaLitmusExp]|PSTATE.SP=0b1|\\
\lstinline[language=IslaLitmusExp]|R1=pte2(x,hyp_pgtable)|\\
\lstinline[language=IslaLitmusExp]|R5=x|\\
\lstinline[language=IslaLitmusExp]|PSTATE.EL=0b10|\\
\lstinline[language=IslaLitmusExp]|TTBR0_EL2=ttbr(asid=0x0000,base=hyp_pgtable)|\\
\lstinline[language=IslaLitmusExp]|VBAR_EL2=0x1000|\\
\lstinline[language=IslaLitmusExp]|R3=bvor(0x283000,offset(level=3,va=x))|\\
\end{minipage}
\end{lrbox}
\begin{tabular}{|l|l|}
  \multicolumn{2}{l}{\textbf{AArch64} \lstinline[language=IslaLitmusName]|pKVM.create_hyp_mappings.inv.l2|}\\
  \hline
  \multirow{6}{*}{\begin{minipage}{\createhypmappingsinvltwoN}\vphantom{$\vcenter{\hbox{\rule{0pt}{1.8em}}}$}Page table setup:\\\usebox{\createhypmappingsinvltwoM}\end{minipage}} & \cellcolor{IslaInitialState}{\usebox{\createhypmappingsinvltwoO}}\\
  \cline{2-2}
  & Thread 0\\
  \cline{2-2}
  & \usebox{\createhypmappingsinvltwoI}\\
  \cline{2-2}
  & thread0 el2 handler\\
  \cline{2-2}
  & \usebox{\createhypmappingsinvltwoK}\\
  \cline{2-2}
  & Final state: \lstinline[language=IslaLitmusExp]|0:R2=0|\\
  \hline
\end{tabular}
}%
\caption{\PKVMTEST{pKVM.create\_hyp\_mappings.inv.l2}: Code listing }\label{fig:createhypmappingsinvltwo_code}
\end{figure}

\begin{figure}[h]
  \includegraphics[trim=10mm 0 0 0, clip,width=120mm]{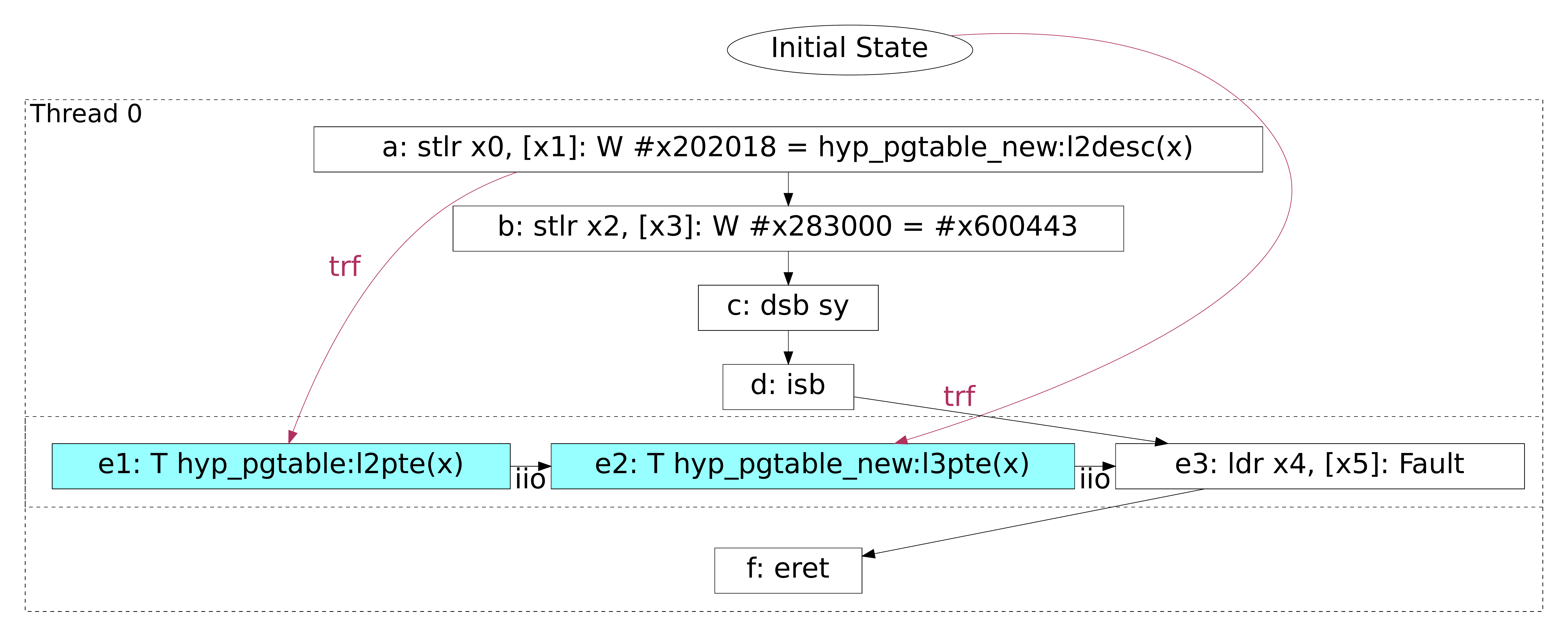}
\caption{\PKVMTEST{pKVM.create\_hyp\_mappings.inv.l2}: (forbidden) candidate execution}\label{fig:createhypmappingsinvltwo_exec}
\end{figure}

\clearpage

\testpara{pKVM.create\_hyp\_mappings.inv.l3}

pKVM can set entries at any level of the table,
assuming they are initially invalid.
The following test is a variation on the previous,
 where the Level~3 table is already created, but contains an invalid entry at Level~3.

As before, we will map a 4K region, and so it must go at Level~3.

\begin{figure}[h]
  \centering
  \scalebox{0.8}{\newlength{\createhypmappingsinvlthreeH}
\setlength{\createhypmappingsinvlthreeH}{0cm}
\newsavebox{\createhypmappingsinvlthreeI}
\begin{lrbox}{\createhypmappingsinvlthreeI}
\begin{lstlisting}[language=AArch64,showlines=true]
STLR X0,[X1]
DSB SY
ISB
L0:
LDR X2,[X3]
\end{lstlisting}
\end{lrbox}
\newlength{\createhypmappingsinvlthreeJ}
\settowidth{\createhypmappingsinvlthreeJ}{\usebox{\createhypmappingsinvlthreeI}}
\addtolength{\createhypmappingsinvlthreeH}{\createhypmappingsinvlthreeJ}
\newsavebox{\createhypmappingsinvlthreeK}
\begin{lrbox}{\createhypmappingsinvlthreeK}
\begin{lstlisting}[language=AArch64,showlines=true]
0x1000:
mov x2, #0
\end{lstlisting}
\end{lrbox}
\newlength{\createhypmappingsinvlthreeL}
\settowidth{\createhypmappingsinvlthreeL}{\usebox{\createhypmappingsinvlthreeK}}
\addtolength{\createhypmappingsinvlthreeH}{\createhypmappingsinvlthreeL}
\newsavebox{\createhypmappingsinvlthreeM}
\begin{lrbox}{\createhypmappingsinvlthreeM}
\begin{lstlisting}[language=IslaPageTableSetup,showlines=true]
option default_tables = false;

physical pa1;

s1table hyp_pgtable 0x200000  {
    x |-> invalid at level 3;
    x ?-> pa1;
    identity 0x1000 with code;
}

*pa1 = 1;
\end{lstlisting}
\end{lrbox}
\newlength{\createhypmappingsinvlthreeN}
\settowidth{\createhypmappingsinvlthreeN}{\usebox{\createhypmappingsinvlthreeM}}
\newsavebox{\createhypmappingsinvlthreeO}
\begin{lrbox}{\createhypmappingsinvlthreeO}
\begin{minipage}{1.8\createhypmappingsinvlthreeH}
\vphantom{$\vcenter{\hbox{\rule{0pt}{1.8em}}}$}Initial state:\\
\lstinline[language=IslaLitmusExp]|R0=mkdesc3(oa=pa1)|\\
\lstinline[language=IslaLitmusExp]|TTBR0_EL2=ttbr(base=hyp_pgtable,asid=0x0000)|\\
\lstinline[language=IslaLitmusExp]|R3=x|\\
\lstinline[language=IslaLitmusExp]|VBAR_EL2=0x1000|\\
\lstinline[language=IslaLitmusExp]|PSTATE.EL=0b10|\\
\lstinline[language=IslaLitmusExp]|R1=pte3(x,hyp_pgtable)|\\
\end{minipage}
\end{lrbox}
\begin{tabular}{|l|l|}
  \multicolumn{2}{l}{\textbf{AArch64} \lstinline[language=IslaLitmusName]|pKVM.create_hyp_mappings.inv.l3|}\\
  \hline
  \multirow{6}{*}{\begin{minipage}{\createhypmappingsinvlthreeN}\vphantom{$\vcenter{\hbox{\rule{0pt}{1.8em}}}$}Page table setup:\\\usebox{\createhypmappingsinvlthreeM}\end{minipage}} & \cellcolor{IslaInitialState}{\usebox{\createhypmappingsinvlthreeO}}\\
  \cline{2-2}
  & Thread 0\\
  \cline{2-2}
  & \usebox{\createhypmappingsinvlthreeI}\\
  \cline{2-2}
  & thread0 el2 handler\\
  \cline{2-2}
  & \usebox{\createhypmappingsinvlthreeK}\\
  \cline{2-2}
  & Final state: \lstinline[language=IslaLitmusExp]|0:R2=0|\\
  \hline
\end{tabular}
}%
\caption{\PKVMTEST{pKVM.create\_hyp\_mappings.inv.l3}: Code listing }\label{fig:createhypmappingsinvlthree_code}
\end{figure}
\begin{figure}[h]
  \centering
  \includegraphics[trim=10mm 0 0 0, clip,width=50mm,height=40mm]{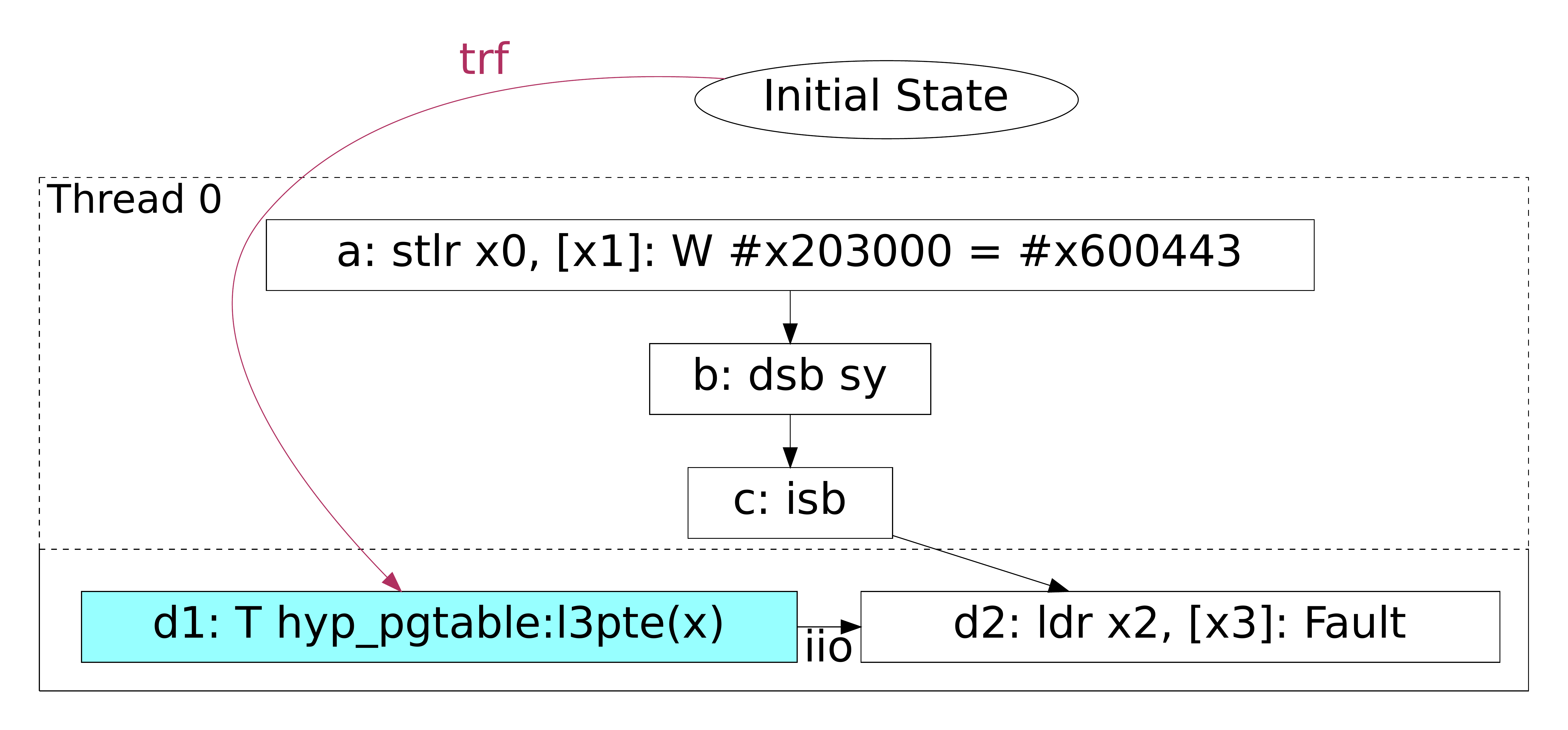}
\caption{\PKVMTEST{pKVM.create\_hyp\_mappings.inv.l3}:  (forbidden) candidate execution}\label{fig:createhypmappingsinvlthree_exec}
\end{figure}

\clearpage

\section{Model}\label{sec:models}

We now define a semantic model for Armv8-A relaxed virtual memory
that, to the best of our knowledge, captures the Arm architectural intent
for the scope laid out in \S\ref{intro} and discussed in \S\ref{questions},
including Stage 1 and Stage 2 translation-table walks and the required TLB maintenance. 
For some important questions, most notably for multi-copy atomicity, the Arm intent is currently tentative, so it is not possible to be more definitive.
To capture just the synchronization required for ``simple'' software such as pKVM to work correctly
we also give a \emph{weaker} model (in App.~\ref{app:models}):
instead of trying to exactly capture the architecture or the behaviour of hardware, 
it has individual axioms for each behaviour that such software needs to rely on.
This gives an over-approximation to the architecture,
which we prove sound with respect
to the model given in this section. %
The two models together delimit the design space
as we understand it. 

In \S\ref{questions} and \S\ref{sec:pkvm} we described the design issues in microarchitectural terms,
discussing the behaviour of TLB caching and translation-walk non-TLB reads,
along with the needs of system software. %
We now abstract from microarchitecture:
instead of explicitly modelling TLBs
we include a translation-read event for each read performed by the architected translation-table walker,
and make those reads read-from writes in the execution
(so there are no special `pagetable write' events).
We give the model in an axiomatic Herd-like~\cite{alglave:herd} style,
as an extension to
the base Armv8-A semantics~\cite{deacon-cat,armv8-mca,G.a}.  
In principle it would be desirable to also have equivalent
abstract-microarchitectural operational models,
as for base Armv8-A~\cite{armv8-mca,pultethesis} but 
with explicit TLBs for each thread
and events for reading from and into the TLB.
However,
address translation introduces many more events to litmus-test executions,
which would make them harder to explore exhaustively,
and a proof of equivalence would be a major undertaking,
so we leave this to future work. 

The base Armv8-A axiomatic model is defined as a predicate over \emph{candidate executions},
each of which is a graph with various events (reads, writes, barriers) and relations over them,
notably the per-thread program order \herd{po},
the per-location coherence order \herd{co},
the reads-from relation \herd{rf} from writes to reads,
the \herd{addr}, \herd{data}, and \herd{ctrl}-dependency subsets of \herd{po}, and others.
These candidates may be arbitrarily inconsistent graphs,
possibly containing executions that can never happen.

The model is then a per-candidate consistency check consisting of two parts:
that the graph corresponds to \emph{some} execution consistent with the underlying ISA,
but with arbitrary memory reads and writes;
and a global consistency check over those reads and writes which enforces memory consistency properties
such as \emph{coherence}.

The base memory consistency model is essentially the conjunction of two acyclicity checks with an emptiness check for atomics:
an \herd{external} (inter-thread) acyclicity property,
effectively stating that the execution must respect some total order of events hitting the shared memory, constrained by the derived ordered-before (\herd{ob}) relation; and an \herd{internal} acyclicity property, enforcing per-location coherence; and an \herd{atomic} axiom for atomic and exclusive operations.
 
As usual in Herd-style models, relations are suffixed \herd{e} or \herd{i} to restrict to their inter-thread or intra-thread parts.
The Herd concrete syntax for relational algebra uses \herd{[X]} for the
identity on a set \herd{X}, \herd{;} for composition, \herd{\char`~} for complement, \herd{|} and \herd{\&} for
union and intersection, and \herd{*} for product. 
To extend this base memory consistency model to the world with translations and TLBs we
add translation data to events, including virtual, intermediate physical, and physical addresses (as determined by the translation regime).

We add the following events and relations:
\begin{itemize}
  \item \herdevent{T} for reads originating from architected translation-table walks. \\
        These roughly correspond to the actual satisfaction from memory which with TLBs may happen very early.
  \item \herdevent{TLBI} events for each \asm{TLBI} instruction,
        with a single such event per \asm{TLBI} instruction,
        corresponding to the \asm{TLBI} being completed on all relevant cores.
  \item \herdevent{TE} and \herdevent{ERET} events for taking and returning from an exception 
        (these might not correspond to changes in exception level).
  \item \herdevent{MSR} events for writes to relevant system registers, such as the \asm{TTBR}.
  \item \herdevent{DSB} events for \asm{DSB} instructions.
  \item \herd{trf}, \herd{tfr} relations as analogues to \herd{rf} and \herd{fr} but for translation-read events (\herdevent{T}s).
  \item \herd{iio} relation (``intra-instruction order'') which relates events of the same instruction in the order they occur during execution
        of that instruction's intra-instruction semantics as defined by the Arm ASL.
  \item \herd{same-va}, \herd{same-ipa}, \herd{same-pa} relations which relate events whose virtual, intermediate physical or physical address
        of the associated explicit memory access are the same.
  \item \herd{same-va-page}, \herd{same-ipa-page}, \herd{same-pa-page} which relate events whose associated explicit memory events are in the same
        page (e.g. 4KiB chunk) of the virtual, intermediate physical or physical address space.
  \item \herd{same-asid}, \herd{same-vmid} relates events for which translations for the associated memory event are using the same ASID or VMID.
\end{itemize}

In addition we modify some existing events and relations:
\begin{itemize}
  \item \herdevent{R}, \herdevent{W} events are now to a physical location.
  \item \herd{loc} and \herd{co} both relate events which are to the same physical address.
  \item \herd{addr} which is derived from a new \herd{tdata} relation, which relates the event which provide the input address for a translation.
  \item We re-arrange the barrier events into a hierarchy which includes \asm{DSB}s, see Figure~\ref{fig:barrierscat}.
\end{itemize}

For convenience we define new event sets:
\herdevent{C} for all cache-maintenance operations (\asm{DC}, \asm{IC}, and \asm{TLBI} instructions);
\herdevent{T\_f} for all translation-read events which read a descriptor which causes a translation fault;
\herdevent{W\_inv} for all the write events which write an invalid descriptor;
\herdevent{Stage1} and \herdevent{Stage2} for the \herdevent{T} events
which originate from the respective stage of translation;
\herdevent{ContextChange} for all context-changing events (such as writes to translation-controlling system registers); and 
\herdevent{CSE} for all context-synchronizing events (taking and returning from exceptions and \asm{ISB}).

\subsection{Strong model}\label{sec:models:subsec:strong}

\begin{figure}[t!]
\scalebox{0.8}{
\begin{minipage}{1.2\textwidth}
\begin{multicols}{2}
\begin{lstlisting}[language=cat]
let tlb-affects =
  see Figure 16

let TLB_barrier =
  ([TLBI] ; tlb-affects ; [T] ; tfr ; [W])^-1
  & wco

let maybe_TLB_cached =
  ([T] ; trf^-1 ; wco ; [TLBI-S1]) & tlb-affects^-1

let tcache1 = [T & Stage1] ; tfr ; TLB_barrier
let tcache2 = [T & Stage2] ; tfr ; TLB_barrier

let speculative =
    ctrl
  | addr; po
  | [T] ; instruction-order
(* translation-ordered-before *)
let tob =
    [T_f] ; tfre
  | ([T_f] ; tfri)
    & (po ; [DSB.SY] ; instruction-order)^-1
  | [T] ; iio ; [R|W] ; po ; [W]
  | speculative ; trfi
(* observed by *)
let obs = [*rfe | fr | wco*]
  | trfe
(* ordered-before TLBI and translate *)
let obtlbi_translate =
    tcache1
  | tcache2
    & (iio^-1 ; [T & Stage1] ; trf^-1 ; wco^-1)
  | (tcache2 ; wco? ; [TLBI-S1])
    & (iio^-1 ; [T & Stage1] ; maybe_TLB_cached)

(* ordered-before TLBI *)
let obtlbi =
    obtlbi_translate
  | [R|W|Fault] ; iio^-1 ; (obtlbi_translate & ext) ; [TLBI]

(* context-change ordered-before *)
let ctxob =
    speculative ; [MSR]
  | [CSE] ; instruction-order
  | [ContextChange] ; po ; [CSE]
  | speculative ; [CSE]
  | po ; [ERET] ; instruction-order ; [T]

(* ordered-before a translation fault *)
let obfault =
    data ; [Fault & IsFromW]
  | speculative ; [Fault & IsFromW]
  | [dmbst] ; po ; [Fault & IsFromW]
  | [dmbld] ; po ; [Fault & (IsFromW|IsFromR)]
  | [A|Q] ; po ; [Fault & (IsFromW | IsFromR)]
  | [R|W] ; po ; [Fault & IsFromW & IsReleaseW]

(* ETS-ordered-before *)
let obETS =
    (obfault ; [Fault]) ; iio^-1 ; [T_f]
  | ([TLBI] ; po ; [dsb] ; instruction-order ; [T]) & tlb-affects

(* dependency-ordered-before *)
let dob =
  [*  addr | data*]
  [*| speculative ; [W]*]
  [*| addr; po; [W]*]
  [*| (addr | data); rfi*]
  | (addr | data); trfi

(* atomic-ordered-before *)
[*let aob = rmw*]
[*  | [range(rmw)]; rfi; [A | Q]*]

(* barrier-ordered-before *)
let bob = [*[R] ; po ; [dmbld]*]
  | [*[W] ; po ; [dmbst]*]
  | [*[dmbst]; po; [W]*]
  | [*[dmbld]; po; [R|W]*]
  | [*[L]; po; [A]*]
  | [*[A | Q]; po; [R | W]*]
  | [*[R | W]; po; [L]*]
  | [F | C]; po; [dsbsy]
  | [dsb] ; po

(* Ordered-before *)
let ob = [*(obs | dob | aob | bob*]
  | iio | tob | obtlbi | ctxob | obfault | obETS)^+

(* Internal visibility requirement *)
[*acyclic po-loc | fr | co | rf as internal*]
(* External visibility requirement *)
[*irreflexive ob as external*]
(* Atomic requirement *)
[*empty rmw & (fre; coe) as atomic*]
(* Writes cannot forward to po-future translates *)
acyclic (po-pa | trfi) as translation-internal
\end{lstlisting}
\end{multicols}
\end{minipage}}
\caption{Strong Model (with baseline Armv8-A model parts in {\color{gray}gray})}\label{fig:strongmodel}
\end{figure}

\begin{figure}[t!]
\scalebox{0.8}{
\begin{minipage}{1.2\textwidth}
\begin{lstlisting}[language=cat]
let tlb_might_affect =
    [ TLBI-S1 & ~TLBI-S2 &  TLBI-VA  &  TLBI-ASID & TLBI-VMID] ; (same-va-page & same-asid & same-vmid) ; [T & Stage1]
  | [ TLBI-S1 & ~TLBI-S2 & ~TLBI-VA  &  TLBI-ASID & TLBI-VMID] ; (same-asid & same-vmid) ; [T & Stage1]
  | [ TLBI-S1 & ~TLBI-S2 & ~TLBI-VA  & ~TLBI-ASID & TLBI-VMID] ; same-vmid ; [T & Stage1]
  | [~TLBI-S1 &  TLBI-S2 &  TLBI-IPA & ~TLBI-ASID & TLBI-VMID] ; (same-ipa-page & same-vmid) ; [T & Stage2]
  | [~TLBI-S1 &  TLBI-S2 & ~TLBI-IPA & ~TLBI-ASID & TLBI-VMID] ; same-vmid ; [T & Stage2]
  | [ TLBI-S1 &  TLBI-S2 & ~TLBI-IPA & ~TLBI-ASID & TLBI-VMID] ; same-vmid ; [T]
  | ( TLBI-S1 &            ~TLBI-IPA & ~TLBI-ASID & ~TLBI-VMID) * (T & Stage1)
  | (            TLBI-S2 & ~TLBI-IPA & ~TLBI-ASID & ~TLBI-VMID) * (T & Stage2)

let tlb-affects =
    [TLBI-IS] ; tlb_might_affect
  | ([~TLBI-IS] ; tlb_might_affect) & int
\end{lstlisting}
\end{minipage}}
\caption{The \herd{tlb-affects} relation.}
\label{fig:tlbaffectscat}
\end{figure}

\begin{figure}[t!]
  \scalebox{0.8}{
  \begin{minipage}{1.2\textwidth}
\begin{lstlisting}[language=cat]
let dsbsy = DSB.ISH | DSB.SY | DSB.NSH
let dsbst = dsbsy | DSB.ST | DSB.ISHST | DSB.NSHST
let dsbld = dsbsy | DSB.LD | DSB.ISHLD | DSB.NSHLD
let dsbnsh = DSB.NSH
let dmbsy = dsbsy | DMB.SY
let dmbst = dmbsy | dsbst | DMB.ST | DSB.ST | DSB.ISHST | DSB.NSHST
let dmbld = dmbsy | dsbld | DMB.LD | DSB.ISHLD | DSB.NSHLD
let dmb = dmbsy | dmbst | dmbld
let dsb = dsbsy | dsbst | dsbld
\end{lstlisting}
  \end{minipage}}
  \captionsetup{justification=centering}
  \caption{Barrier definitions. \\
  Note we do not distinguish between Inner-Shareable and Full-System barriers}
  \label{fig:barrierscat}
  \end{figure}

The model is given in full in Figure~\ref{fig:strongmodel} with auxiliary definitions
of the \herd{tlb-affects} and barrier hierarchy given in Figures~\{\ref{fig:tlbaffectscat},\ref{fig:barrierscat}\}.

Its basic form is very similar to previous multicopy-atomic Armv8-A models.
It still has \herd{external}, \herd{internal}, and \herd{atomic} axioms,
to which we add  a \herd{translation-internal} axiom for ensuring translations do not read from po-later writes.

Most of the changes to the model are in the \herd{external} axiom,
where we add several relations to  ordered-before (\herd{ob}):
\herd{iio} orders the intra-instruction events as ordered by the ASL;
\herd{tob} (``translation ordered-before'') ensures the order arising from the act of translation itself is respected;
\herd{obtlbi} orders translates and their explicit memory events with \asm{TLBI}s
which affect these translations;
and
\herd{ctxob} (``context ordered-before'') orders events which must come before some context-changing operation
or after some context-synchronizing
operation.
We also add a generalised coherence-order relation, \herd{wco},
an existentially quantified total order expressing when TLBIs complete w.r.t.~writes.

\mysubsubsection{Coherence}
By making \herd{loc}
(and therefore
\herd{rf} and \herd{co})
relate events with the same physical addresses, we get
coherence over physical addresses rather than virtual,
and all the previously allowed shapes are also allowed
when there is aliasing with different virtual addresses.
Coherence of writes to translation tables is expressed in two places:
including \herd{trfe} in \herd{obs} captures the fact that translation-table reads from memory %
microarchitecturally come from the `flat' coherent storage subsystem, 
and so the writes that
they read from must have been propagated before the translation happened;
and the \herd{translation-internal} axiom forbids forwarding against program-order.
Note that including only \herd{trfe} allows forwarding locally (a \herd{trfi} edge),
and including \herd{(addr|data);trfi} in \herd{dob} ensures those forwarded writes cannot form bad self-satisfying %
cycles.

\mysubsubsection{TLB maintenance and break-before-make}
The \herd{obtlbi} relation ensures that instructions whose translations
read from writes which
are ``hidden'' by some \asm{TLBI} instruction
are ordered before the completion of that \asm{TLBI}.
This is achieved by the two clauses of \herd{obtlbi}:
the first clause ensures the translation-before-TLBI ordering is preserved,
and
the second clause orders
the explicit memory access of any such instruction with the same \herd{TLBI}
as the first clause.
To do this, the model computes the set of writes which are in effect ``barriered''
by a given \asm{TLBI} instruction,
by looking at all translations in the execution,
and if any translation reads-from a write which is before a \asm{TLBI},
we then get \herd{TLB\_barrier} between them.
The \herd{tcache1} and \herd{tcache2} relations then simply relate
translations which read from coherence-predecessors of any of those writes
with their respective barriering \asm{TLBI}.

To accurately match up each of the various \asm{TLBI} instructions
with the translations they may affect, we define a \herd{tlb-affects} 
relation which relates \herd{TLBI} events with the \herd{T} events they are relevant to.
Its definition uses sets \herd{TLBI-VA}, \herd{TLBI-ASID}, \herd{TLBI-IPA}, \herd{TLBI-VMID}, and \herd{TLBI-ALL} for each of the categories of \asm{TLBI} instruction.
Note that some instructions can fall into multiple categories,
such as \asm{TLBI VAE1} which is in \herd{TLBI-VA} for the specified virtual address,  \herd{TLBI-ASID},
as the register input contains an ASID to perform the invalidation for, and also \herd{TLBI-VMID} as the invalidations only affect translations in the same VM.\@

We add \herd{obtlbi\_translate} to relate those translations to \asm{TLBI}s which
invalidate the writes they read from.
For Stage~1 translations we can simply order any Stage~1 translation
before any \asm{TLBI} which would \asm{tlb\_affect}
this translation where the translation
reads from a write which is ordered-before than the \asm{TLBI}.
However, for Stage~2 translations this is not sufficient.
Recall that microarchitecturally the TLB could store whole virtual-to-physical mappings,
and so a Stage~2 translation-read is only ordered after the \asm{TLBI}s which
remove not only any Stage~2 mappings but also those that would remove the combined Stage~1 and Stage~2 mappings.
For a Stage~2 translation whose previous Stage~1 walk only read from writes newer
than the \asm{TLBI} then the Stage~2 invalidation is sufficient.
But where any of the reads read-from a write older than the \asm{TLBI},
a cached virtual-to-physical mapping could exist and Stage~1 invalidation
is required, hence the Stage~2 translation is ordered after the \emph{Stage~1} invalidation.

\mysubsubsection{Translation-table-walk reading from memory}

As noted in \S\ref{subsec:questions:non_tlb}, a translation which results in a translation fault
must read from memory or be forwarded from program-order earlier instructions,
and those memory reads behave \emph{multi-copy atomically}.
In general the only time the model can guarantee that such a memory read happens is
when the read results in a translation fault,
since entries that result in a translation fault cannot be stored in the TLB (\S\ref{subsubsec:questions:tlb:what_cached}).
The model captures this succinctly by including \herd{[T\_f];tfr} in \herd{ob}.

In general, a translation-read is ordered after the write which it reads from,
as captured by the inclusion of the \herd{trfe} edge in \herd{ob};
this is strong enough to ensure that TLB fills and faulting memory walks pull values out of the memory system in a coherent way,
but still weak enough to allow \emph{other}-multi-copy-atomic behaviour such as forwarding.

As discussed in \S\ref{nonspec-same},
a \asm{DSB} ensures that writes are propagated out to memory.
For translations this amounts to ensuring that a faulting translation
cannot read-from something older than a po-previous \asm{DSB}-barriered write,
as captured by the last edge in \herd{tob} which says that a \herd{tfri} edge
from such a faulting translation must not have an interposing \asm{DSB}.

Note that the absence of the full \herd{tfr} relation in \herd{ob} for non-faulting translations
intentionally allows some incoherence,
in essence allowing a translation-read to ``ignore'' a newer write.

\mysubsubsection{Context-changing operations}
In general, the sequential semantics takes care of the context,
such as current base register and system register state,
for us.
The %
\herd{ctxob} relation simply ensures that
such context-changing operations cannot be taken speculatively,
and that context-synchronization ensures that all po-%
previous context-changing operations
are ordered-before po-later translations.

\mysubsubsection{Detecting BBM Violations}
As discussed in \S\ref{subsubsec:questions:tlb:bbm}, we do not model in detail
the bounded-catch-fire semantics that currently architecturally results from 
a missing break-before-make sequence, as that would make it hard to enumerate possible litmus-test executions.
Instead, because what one normally wants to know for litmus tests is that a test does not exhibit a BBM failure,
we conservatively detect the existence of such violations and flag them for the user.
This is achieved through a per-candidate-execute predicate, written in SMT,
which looks for a situation which \emph{could} be a break-before-make violation.
It does this by asserting that there does not exist a pair of writes which conflict
such that there is no interposing break-and-\asm{TLBI} sequence.
This approach is slightly over-approximate, as it might look for two writes that technically conflict
even if they (for other reasons) are not used at the same time.
This means that while we support programs that switch from one page table to another,
we do not support programs that garbage collect page-table memory and then repurpose it.

\mysubsubsection{ETS}
We discussed the Armv8-A optional ETS feature, providing additional ordering strength for translations, in 
\S\ref{nonspec-same},\ref{twoets}. 
The intuition is that the model would have
\emph{ghost} events in the event an instruction faults, 
to represent the explicit read or write which would have happened
had the instruction not faulted. %
The model would then have to compute a special variant of \herd{ob}
including such dependencies, 
but without the physical-address-dependent relations such as \herd{loc}, \herd{rf} and \herd{co}.
Then any edge in the version of \herd{ob} with the ghost events would become
an edge in the real \herd{ob} but attached to the faulting translation.
To capture this, our model produces fault events which have the correct dependencies
(and fault information) and the model orders the fault event
with respect to program-order previous events which would have ordered and place those
into \herd{ob}.
To achieve this,
we manually insert all edges from the syntactic subsets (those edges which do not rely on \herd{loc}) from \herd{bob} and \herd{dob}
into a \herd{obfault} relation.
We use this to build an \herd{obETS} relation which then orders translations that result in a translation fault after anything the fault is ordered-after.

An additional complexity here is for thread-local behaviours of \asm{TLBI} instructions.
With ETS one does not require context synchronization to see the effect of a \asm{TLBI} thread-locally.
Our \herd{obETS} covers this with its second clause which orders translations from instructions \herd{po}-after
a subsequent \asm{DSB} as happening after any \asm{TLBI} which affected that translation.

\mysubsubsection{Reclaimation of pagetable memory}

There may be cases where the memory being used to store a translation table may become unreachable by any TTBR
and properly cleaned from the TLBs.
In practice this means the memory can now be reclaimed and re-purposed.

Allowing this is work-in-progress but the model as presented here 
does not support it and our break-before-make-violation detection predicate
will assume that this is a break-before-make violation.

\subsection{Weak Model}
\label{sec:weakestmodel}

Relaxed memory model design for hardware architectures has to resolve a three-way tension between providing enough strength for software (forbidding enough behaviours so that code works as desired without needing excessive synchronisation), weakness for hardware (rendering desirable microarchitectural optimisations sound), and simplicity. 
For ``user'' concurrency, one has to accommodate the broad space of concurrent code in the wild, which is hard to map, but systems concurrency is managed by much smaller bodies of code, in more specific ways.
This makes it interesting to explore models which are as weak as possible, subject to the constraints from system software such as pKVM. 

We define such a model by capturing just the requirements 
we identified from pKVM usage (\S\ref{sec:tests}),
expressing them as additional axioms over the Armv8-A base model:
coherence over physical memory;
no self-satisfying translations
or translations using speculative writes;
`breaking' a translation with an invalidation and a
broadcast \asm{TLBI} ensuring that all cores have finished using
that translation before the \asm{TLBI} returns;
writing a new entry to a broken page without \asm{TLB} maintenance;
and
changing translation tables and context
without TLB maintenance.
This  weak model uses the same candidates and auxiliary definitions as the strong, but
instead of including extra edges in \herd{ob} we impose new axioms for each of those behaviours.
For example, for the `break' part of Stage~1 break-before-make we have:
{
\begin{lstlisting}[language=cat]
  empty ([W] ; co ; [W_invalid] ; ob ; [dsb.sy] ; po
    ; ([TLBI-S1] ; po ; [dsb.sy] ; ob ; [CSE] ; instruction-order ; [T]) & tlb_affects
    ) & trf
    as brk2
\end{lstlisting}
}
This forbids the case where a write is read-from by a translation-table-walk when there is an interposing break and Stage~1 TLB-invalidation sequence,
which would `hide' that write from future translations.
Note the \herd{ob} edges allow the sequence to be split over multiple threads
in the context-switching scheduler case described in \S\ref{forbidden-past}.

Note that the pKVM developers believe that pKVM does not rely on the ETS feature,
and so the weak model does not include ETS.

\section{Metatheory: relationships between models}\label{proof}

The virtual memory mechanisms are complex in both sequential and concurrent ways, as we have seen, but they are intended to let system software provide a relatively simple abstraction to higher-level code.
As first steps towards establishing this, and as sanity checks of our models, we prove three theorems: 
that for static injectively-mapped address spaces, any execution which is consistent in the model with translation, erasing translation events gives an execution that is consistent in the original Armv8-A model without translation;
that for any consistent execution in the original Armv8-A model, there is a corresponding consistent execution in our extended model with translations;
and that our weak model is a sound over-approximation of our full translation model, i.e.,
that for any consistent execution in our full translation model, that same execution is consistent in the weak translation model.
Details are in App.~\ref{app:modelsrelation}.

\section{Isla-based model evaluation}
\label{isla}

Making relaxed-memory semantics exhaustively executable is essential
for exploring their behaviour on examples%
~\cite{%
DBLP:conf/ipps/YangGLS04,%
x86popl,%
pldi105,%
cpp-popl,%
alglave:herd,%
micro2015,%
DBLP:conf/popl/WickersonBSC17,%
DBLP:conf/pldi/BornholtT17,%
DBLP:conf/asplos/TrippelMLPM17,%
armv8-mca,%
iflat-esop2020-extended}.
Handling relaxed virtual memory brings several new challenges.
First, even just the sequential definition of Armv8-A address translation, with the page-table walk and its options, is remarkably intricate, defined in thousands of lines of Arm's ASL instruction description language.  Manually reimplementing a simplified version would be error-prone and incomplete, so 
we instead build on our Isla tool~\cite{isla-cav}, which integrates the full 123,000 line Armv8-A ISA semantics (as defined by Arm in ASL and automatically translated into Sail~\cite{sail-popl2019}), with SMT-based tooling to evaluate tests w.r.t.~axiomatic concurrency models.
Previously Isla supported only ``user'' models, expressed in a language based on relational-algebra similar to the Cat language of Herd~\cite{alglave:herd}.
The integration with a full ISA semantics led us to raise 
several of the questions of \S\ref{questions}, e.g.~relating to system registers and mixed-size effects, which would not arise in a more idealised setting.
The second main challenge is the combinatorics.

Previous litmus tests typically involved
only a few abstract memory locations and events, but even %
simple virtual memory tests
require 30kB of page
tables,  each ``user'' memory access might have 24 or more page-table accesses,
and %
each 64-bit descriptor may be represented by
a symbolic value representing all possible states that descriptor can
be in. To avoid overwhelming the SMT solver during symbolic execution,
the formula representing each symbolic descriptor is created
dynamically when read. When encoding the final SMT problem that
decides whether a candidate execution is allowed,  %
we ensure that only the parts of the page
tables actually used by that candidate execution are
included.
We also implemented a model-specific optimization that removes irrelevant translation events which cannot affect the result
of the test, improving performance by a factor of 13 on average, and up to 90 times for some tests.
Third, we had to provide a convenient way to express the
page table configuration for each test, with the declarative 
 language of which we saw a small part on the
left-hand side of the \S\ref{sec:pkvm} test. %

  \vspace*{-4mm}
\begin{center}
  \includegraphics[trim=0 170mm 0 0, clip,width=0.9\textwidth]{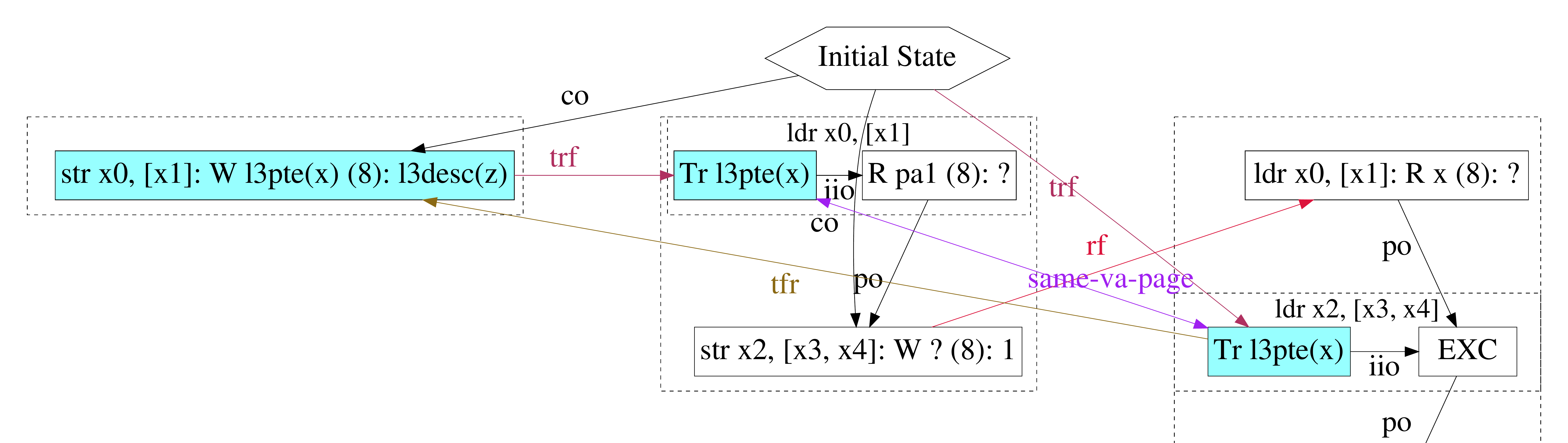}
\end{center}
  \vspace*{-4mm}
A good user interface is essential. Above, we show an Isla-generated execution for a WRC test like that of \S\ref{wrctestsec}, showing how uninteresting translation events can be suppressed in the output to avoid overwhelming noise.

The main result is that, in the strong model, all \numberOfIslaLitmusTests{} litmus tests and
\numberOfpKVMTests{} pKVM tests are allowed or forbidden as intended, based on our discussion with Arm of their
architectural intent, except two pKVM tests which time out.
Additionally, we
tested that the weak model never forbids any test allowed by the
strong model.
The tool performance is eminently usable in practice: most tests take around 1 minute, and the full set of litmus tests can be run in less than 2 hours CPU time, 
on a 36-core Intel Xeon Gold 6240.
\TODO{BS: time to run + timeouts}
Details are in App.~\ref{app:results}.

A further key property is that for ordinary relaxed-memory litmus
tests which do not involve virtual memory,  our model
should give the same results as the published
\mbox{Armv8-A}~\cite{deacon-cat,armv8-mca,G.a} axiomatic memory model. To
validate this (and our tools) we test our strong model on an existing
library of tests, comparing to reference results from Herd and
RMEM~\cite{FGP16}. We ran an additional 1927 such litmus tests, which
all returned the expected results.  

\section{Experimental testing of hardware}
\label{harness}

Experimental investigation of hardware implementation behaviour,
and experimental validation of models with respect to that,
is one important input to the development of practically relevant relaxed memory semantics~\cite{DBLP:books/daglib/0073498,x86popl,Adir:2003,Alglave:2011:LRT:1987389.1987395}.
However, almost all that work has focussed on ``user'' concurrency, with litmus tests that could be run as user processes under a normal OS, and that could easily iterate tests over arrays.
Experimental testing of virtual memory behaviour is considerably more challenging, as one needs to run code at higher privilege levels, including exception handlers, and manipulate the page tables that a normal OS and/or hypervisor would be depending on.
When we started this work, that was not supported by litmus, so we have developed a litmus-like test harness for running virtual-memory tests bare-metal or in KVM. 
Currently it runs Stage~1 tests only; for Stage~2 tests some adaption to run code at EL2 is still needed.
The harness can be found at \url{https://github.com/rems-project/system-litmus-harness}.
At present this and Isla use different test formats, so we have some tests manually written in both.

We ran tests on three devices:
a Raspberry Pi 3 (Arm A53),
a Raspberry Pi 4 (Arm A72),
and an AWS \texttt{m6g.metal} (AWS Graviton2, claiming to be an A72).
Our experimental data suggests that all are multi-copy atomic with respect to translation-table-walks,
respect coherence over physical locations,
correctly perform TLB maintenance,
and do not disagree with the tests presented here except for one behaviour:
we sometimes observe anomalous results with respect to writes not being made globally visible to translation-table-walks
(a \asm{DSB} not sufficing); this is currently under discussion with Arm.
Full results are in App.~\ref{app:results}.

{Further testing on other platforms would be desirable, but our emphasis in this work is principally on exploring the design space and capturing the architectural intent, and the main validation is from discussion with the Arm Chief Architect, who  ultimately is responsible for determining what the architecture is. 
In this context, experimental data serves mainly to provide
reassurance that some envisaged architecture strength is not
invalidated by extant hardware implementations.
}

\section{Related work}\label{related}%
There is extensive previous work on ``user'' relaxed-memory semantics of modern architectures, but very little extending this to cover systems aspects such as virtual memory. 
We build on the approaches established in ``user'' models for x86, IBM Power, Arm, and RISC-V, combining executable-as-test-oracle models, discussion with architects, and experimental testing~\cite{x86popl,damp09,cav2010,tphols09,cacm,pldi105,popl2012,pldi2012,tut,alglave:herd,micro2015,FGP16,%
mixed17,armv8-mca,riscv-20181221-Public-Review-draft}.

{Arm publish a machine-readable version of their Armv8-A relaxed memory model}~\cite{arm-memory-model-tool}, {in the Cat language of the Herd7 tool}~\cite{herd7}, {but 
that model does not currently cover the relaxed virtual-memory semantics.
Independent work in progress %
by Alglave et al.{}
is similarly aiming to characterise this, and to update Arm's published model in due course, but with complementary scope to the current paper:
including hardware updates of access and dirty bits, but
without integration with the full ASL/Sail instruction semantics and its multiple levels and stages of translation.
Both have been informed by discussion with senior Arm staff, and one would hope to synthesise the understanding in future.}
Hossain et al.~\cite{DBLP:conf/isca/HossainTM20} develop an ``estimated'' model for virtual memory in x86 (which has a much less relaxed base semantics) in a broadly similar axiomatic style.
Tao et al.~\cite{Gu2021} axiomatise six conditions for \emph{weak data-race-freedom} that should be satisfied by Armv8-A kernel code that uses virtual memory in simple ways,
and an extension of Promising-Arm~\cite{DBLP:conf/pldi/PultePKLH19} that effectively builds in these conditions%
; they extend the sequential verification of the SeKVM hypervisor by Li et al.~\cite{DBLP:conf/uss/LiLGNH21} to show it satisfies these conditions.
The paper does not
attempt to characterise the exact guarantees provided by the Armv8-A architecture, or discuss the issues of our \S\ref{questions}.  A foundational model such as our \S\ref{sec:models} would let one ground such results on the actual architecture. 
Simner et al.~\cite{iflat-esop2020-extended} study
relaxed instruction-fetch semantics.
Several works give non-relaxed-memory semantics for Arm or x86 %
address translation, more or less simplified and with or without TLBs:
Bauereiss~\cite{sail-popl2019},
Goel et al.~\cite{GoelPhD,DBLP:books/sp/17/GoelHK17},
Syeda and Klein \cite{DBLP:conf/lpar/SyedaK17,DBLP:conf/itp/SyedaK18,DBLP:phd/basesearch/Syeda19,DBLP:journals/jar/SyedaK20},
Degenbaev~\cite{DBLP:conf/birthday/DegenbaevPS09}
(used for verification of a hypervisor shadow pagetable implementation~%
\cite{DBLP:phd/dnb/Kovalev13,DBLP:phd/dnb/Degenbaev12,DBLP:conf/vstte/AlkassarCKP12,DBLP:conf/fmcad/AlkassarCHKP10}),
Barthe et al.%
~\cite{DBLP:conf/aplas/BartheKS08,DBLP:conf/fm/BartheBCL11,DBLP:conf/csfw/BartheBCL12,DBLP:conf/types/BartheBCCL13},
Tews et al.%
~\cite{DBLP:journals/jar/TewsVW09},
Kolanski%
~\cite{DBLP:phd/basesearch/Kolanski11}, and
Guanciale et al.~\cite{DBLP:journals/jcs/GuancialeNDB16}.

\section{Acknowledgments}
{We thank Arm Ltd.~for its support of Simner's PhD and the wider project of which this is part. }
We thank the Google pKVM development team, especially Will Deacon, Quentin Perret, Andrew Scull, Andrew Walbran, and Serban Constantinescu, for discussions on pKVM, and the Google Project Oak team, Ben Laurie, Hong-Seok Kim, and Sarah de Haas, for their support. We thank Luc Maranget for comments on a draft. 

This work was partially funded by an Arm/EPSRC iCASE PhD studentship (Simner), Arm Limited, Google, ERC Advanced Grant (AdG) 789108 ELVER, and the 
UK Government Industrial Strategy Challenge Fund (ISCF) under the Digital Security by Design (DSbD) Programme, to deliver a DSbDtech enabled digital platform (grant 105694).

\appendix

\newcommand{\mysetheader}[1]{
}
\makeatother

\newpage

\newcommand{\myappendix}[2]{%
\section{#1}\label{#2}%
\mysetheader{Appendix~\ref{#2}: #1}
}

\appendix
% auto-generated by cutlines; do not edit

\newcommand{\myenlargethispage}[1]{}

\ifvmsaTestIndexStandalone
\newcommand{\thisdoc}{document}
\else
\newcommand{\thisdoc}{appendix}
\fi

\newcommand{\xxtodiscuss}[1]{}

\newcommand{\mynewpage}{\newpage}

\myappendix{VMSA litmus tests}{app:vmsa}

This \thisdoc{} gives the main Armv8-A virtual-memory-systems-architecture litmus tests that we have developed, systematically exploring the design space.

It is structured into subsections,
with each building upon the tests of the previous and
expanding the architectural scope of the tests.
Each subsection is divided into subsubsections for each \emph{shape}.
Each shape may have many variations, e.g.~with different choices of dependencies or barriers or cache maintenance instructions. 
\S\ref{sec:testformat} explains the test format,
then subsequent sections describe test shapes in detail:

\begin{description}
    \item[\S\ref{sec:aliasing}] explores coherence over physical and virtual addresses.
    \item[\S\ref{sec:map}] gives tests which create new simple mappings for previously unused pages.
    \item[\S\ref{sec:unmap}] considers unmapping in-use pages and the TLB invalidation requirements.
    \item[\S\ref{sec:tlbi}] considers extended questions about the operation of the \asm{TLBI} instruction.
    \item[\S\ref{sec:s1bbm}] gives tests which swap one translation for another, and the required break-before-make sequence.
    \item[\S\ref{sec:ttwordering}] considers ordering within a single translation-table walk. 
    \item[\S\ref{sec:mca}] considers for multi-copy atomicity.
    \item[\S\ref{sec:multiprocess}] gives  address-space tests.
    \ifwiptests\item[\S\ref{sec:s2bbm}] considers the Stage~2 requirements for re-mapping entries.\fi
\end{description}

Throughout, unless otherwise stated, the tests apply to both Stage~1 and Stage~2 translations,
and for all exception levels, 
and memory is by default \emph{normal} and \emph{cacheable}.

\ifvmsaTestIndexStandalone
\Mysubsection{Test Results}
The Isla-generated results can be found alongside each test,
but we also include a table here with all results.

\begin{longtable}{p{.6\textwidth}@{}l@{}lr@{}l}
&  \multicolumn{2}{c}{Strong model } & \multicolumn{2}{c}{Ets model } \\
Test Name & allow? & \hfill time \hfill & allow? & \hfill time \\
\TEST{BBM+dsb-tlbiis-dsb}\dotfill  & \hfill $\checkmark$ & \hfill 12702ms \hfill & $\checkmark$ & \hfill 19328ms \\
\TEST{BBM.Tf+dsb-tlbiis-dsb}\dotfill  & \hfill $\checkmark$ & \hfill 9196ms \hfill & $\checkmark$ & \hfill 13659ms \\
\TEST{Break2.news1}\dotfill  & \hfill $\times$ & \hfill 43306ms \hfill & $\times$ & \hfill 22677ms \\
\TEST{CoRR0.alias+po}\dotfill  & \hfill $\times$ & \hfill 2512ms \hfill & $\times$ & \hfill 1420ms \\
\TEST{CoRR2.alias+po}\dotfill  & \hfill $\times$ & \hfill 1934ms \hfill & $\times$ & \hfill 2260ms \\
\TEST{CoRT.inv+addr-trfi}\dotfill  & \hfill \textbf{\textcolor{red}{!}} & \hfill 1098ms \hfill & \textbf{\textcolor{red}{!}} & \hfill 4778ms \\
\end{longtable}

\else
The Isla-generated results can be found alongside each test,
but we also include a table in App.~\ref{app:results} with all results.

\fi

\mynewpage
\Mysubsection{Test Format}\label{sec:testformat}

In this document the tests are given in a consistent format.
Each test is given in three parts:
\begin{itemize}
    \item The test listing.
    \item Execution witness diagram.
    \item Isla output.
\end{itemize}

\subsubsection{Naming Convention}

Throughout this document we will use a standard convention for names.

Each test name is of the following format:
\begin{lstlisting}
    TestName ::= ExtendedShape ("+" ThreadEdges)+
    ThreadEdges ::= EDGE | ThreadEdges "-" ThreadEdges
    ExtendedShape ::= SHAPE ("." ["Tf"|"T"|"R"|"Rpte"])* (".EL1")? (".inv")?
\end{lstlisting}

For tests that are completely new shapes, those shapes have their own names:
\begin{itemize}
    \item ROT (``Re-ordered translations'', \S\ref{test:ROT.inv+dsb})
    \item RBS (``Read broken secret'', \S\ref{test:RBS+dsb-tlbiis-dsb})
    \item BBM (``Break-before-make'', \S\ref{test:BBM+dsb-tlbiis-dsb})
    \item etc
\end{itemize}

Otherwise the underlying shape is just one of the original `data memory' shapes:
\begin{itemize}
    \item MP
    \item SB
    \item LB
    \item etc
\end{itemize}

To produce the full name we take the shape and expand it out to include an extra \cc{R} for each read,
for example \cc{MP} becomes \cc{MP.RR} as there are two reads.
The \cc{R}s represent the loads that happen with each thread appearing as a block
with the \cc{R}s within the block following program order.
Then each \cc{R} can be replaced with either, \cc{T} (a successful translation), \cc{T\_f} (a translation which results in a fault),
a \cc{Rpte} (a load of the pagetable entry itself) or remain a \cc{R} (a data memory load)

For example, \cc{MP.RpteT.inv+dsb-isb} represents an \cc{MP}-shaped test
where the first load on the receiver thread is replaced with a load of the pagetable
and the second read is a translation (which succeeds) reading from the (initially invalid) initial state,
where there is a \asm{DSB SY ; ISB} between the two instructions on the receiving thread.
See \TEST{MP.RpteT.inv+dsb-isb} for the full test and diagram

\subsubsection{Test Listing}

The test listing is comprised of 4 main sections:
\begin{itemize}
    \item memory and page-table initialisation code (on the left-hand-side).
    \item per-thread initial state (in the ``Initial state'' section on the right-hand-side).
    \item thread sections (labelled `Thread~0',`Thread~1',`Thread~0 EL1 handler', etc).
    \item final state condition.
\end{itemize}

\paragraph{Pagetable setup}
The core of the pagetable setup is a small DSL,
whose syntax is given by the grammar in Fig.~\ref{fig:pgtable_dsl}.

The setup is a sequence of \emph{constraints}.
Initially the table is unconstrained,
except for some initial mappings (each code section is identity mapped executable, etc).

Constraints can create new physical, intermediate-physical  or virtual addresses to be used in the program,
and set initial (and other possible) states of their mappings.
The pagetable setup code must describe not only the initial state of all translation tables,
but also any intermediate or final state that the translation tables could be in during execution of the test.

This constraint language comes with some built-in functions:
\begin{itemize}
    \item[-] \cc{raw(N)} is a raw 64-bit number,  useful as right-hand-side of \cc{|->} relations.
    \item[-] \cc{table(addr)} is a mapping for a whole table, useful as the right-hand-side of a \cc{|->} relation. 
    \item[-] \cc{va_to_pa} casts a virtual address to a physical one. \\
                See also \cc{pa_to_va}, \cc{ipa_to_va}, \cc{ipa_to_pa}, etc.  
\end{itemize}

\begin{figure}
    \ifvmsaTestIndexStandalone
\begin{lstlisting}

PageTableSetup ::= Constraint (";" Constraint)*

Constraint ::=
    "option" name "=" ("true" | "false") # enable/disable option
  | Alignment? "virtual" name+           # va/ipa/pa with optional alignment
  | Alignment? "intermediate" name+
  | Alignment? "physical" name+
  | "identity" Expr WithAttrs? Level?    # identity mapping
  | Expr "|->" Expr WithAttrs? Level?    # mapsto
  | Expr "?->" Expr WithAttrs? Level?    # maybe mapsto
  | "*" Expr "=" Expr                    # deref equality
  | "assert" Expr
  | Stage name Expr?
  | Stage name Expr? "{" PageTableSetup "}"

Expr ::=
    Expr "&&" Expr    # boolean AND
  | Expr "||" Expr    # boolean OR
  | Expr "&" Expr     # bitwise AND
  | Expr "|" Expr     # bitwise OR
  | Expr "^" Expr     # bitwise XOR
  | Expr "==" Expr    # equality
  | Expr "!=" Expr    # inequality
  | "~" Expr          # bitwise negation
  | name "(" (Arg ("," Arg)*)? ")"  # call
  | hex | nat | bin
  | "(" Expr ")"

Arg ::=
    name "=" Expr       # keyword argument
  | Expr

Stage ::= "s1table" | "s2table"
Level ::= "at" "level" nat

Alignment ::= "aligned" u64

AttrField ::= name "=" (hex | bin)
Attrs ::= 
    "default"
  | "code"
  | "[" AttrField ("," AttrField)* "]"

WithAttrs ::= 
    "with" Attrs              # with S1
  | "with" Attrs "and" Attrs  # with S1 and S2
\end{lstlisting}
    \else
\begin{lstlisting}

PageTableSetup ::= Constraint (";" Constraint)*

Constraint ::=
    "option" name "=" ("true" | "false") # enable/disable option
  | Alignment? "virtual" name+           # va/ipa/pa with optional alignment
  | Alignment? "intermediate" name+
  | Alignment? "physical" name+
  | "identity" Expr WithAttrs? Level?    # identity mapping
  | Expr "|->" Expr WithAttrs? Level?    # mapsto
  | Expr "?->" Expr WithAttrs? Level?    # maybe mapsto
  | "*" Expr "=" Expr                    # deref equality
  | "assert" Expr
  | Stage name Expr?
  | Stage name Expr? "{" PageTableSetup "}"

Expr ::=
    Expr "&&" Expr    # boolean AND
  | Expr "||" Expr    # boolean OR
  | Expr "&" Expr     # bitwise AND
  | Expr "|" Expr     # bitwise OR
  | Expr "^" Expr     # bitwise XOR
  | Expr "==" Expr    # equality
  | Expr "!=" Expr    # inequality
  | "~" Expr          # bitwise negation
  | name "(" (Arg ("," Arg)*)? ")"  # call
  | hex | nat | bin
  | "(" Expr ")"

Arg ::=
    name "=" Expr       # keyword argument
  | Expr

Stage ::= "s1table" | "s2table"
Level ::= "at" "level" nat

Alignment ::= "aligned" u64

AttrField ::= name "=" (hex | bin)
Attrs ::= 
    "default"
  | "code"
  | "[" AttrField ("," AttrField)* "]"

WithAttrs ::= 
    "with" Attrs              # with S1
  | "with" Attrs "and" Attrs  # with S1 and S2
\end{lstlisting}
    \fi
    \caption{Pagetable Setup DSL --- Simplified Grammar}
    \label{fig:pgtable_dsl}
\end{figure}

\paragraph{Initial and final state}

The initial state box is a key-value store,
mapping the per-thread registers to initial values.
These values are set just after machine reset.

The final state box contains a single expression which asserts the expected values of registers in the relaxed outcome.
If the final state is allowed, then the test will have exhibited relaxed behaviours.

Both the final and initial state boxes use a simplified expression language which is common to both,
and its full simplified syntax is given by the grammar in Fig.~\ref{fig:toml_expr_syntax}.

This expression language comes with some built-in functions, which make it easier to write the tests:
\begin{itemize}
    \item[-] \cc{extz(v,bits)} zero-extends \cc{v} to be \cc{bits} wide.
    \item[-] \cc{ttbr(id=id,base=base)} produces a correctly-packed 64-bit number suitable as a \asm{TTBRx\_ELy} value, with base and the given asid/vmid.
    \item[-] \cc{pte3(IA,PTE BASE)} returns the address of the level~3 descriptor used to translate the virtual or intermediate-physical \cc{IA} starting from a table rooted at \cc{PTE BASE}. \\
            There are also \cc{pte2}, \cc{pte1} and \cc{pte0} variants.
    \item[-] \cc{desc3(IA, PTE BASE)} which is roughly \cc{*pte3(IA, PTE BASE)}, that is, the actual 64-bit descriptor found at the address given by the \cc{pteN(...)}.
    \item[-] \cc{raw(N)} is a raw 64-bit number,  useful as right-hand-side of \cc{|->} relations.
    \item[-] \cc{mkdesc3(oa=OA)} constructs a fresh level3 \emph{block} descriptor with default permissions and output address \cc{OA}. \\
            (See also \cc{mkdesc2}, \cc{mkdesc1}).
    \item[-] \cc{mkdesc2(table=ADDR)} constructs a fresh level2 \emph{table} descriptor with default permissions and table address \cc{ADDR}. \\
            (See also \cc{mkdesc1}, \cc{mkdesc0}).
    \item[-] \cc{page(addr)} is the page the address is found in, defined as \cc{addr} right-shifted 12 bits.
    \item[-] \cc{asid(id)} is a 64-bit value suitable for use in \asm{TLBI}-by-ASID instructions, defined to be \cc{id} left-shifted 48 bits.
\end{itemize}

\begin{figure}
    \ifvmsaTestIndexStandalone
\begin{lstlisting}
Expr ::=
    Loc "=" Expr
  | label ":"       # code label
  | bin | hex | nat
  | name "(" (Arg ("," Arg)*)? ")"
  | "(" Expr ")"
  | "~" Expr
  | "true" | "false"
  | Expr ("&" Expr)+
  | Expr ("|" Expr)+
  | Expr "->" Expr

Loc ::=
    nat ":" regname  # register
  | "*" name         # deref

Arg ::=
    name "=" Expr
  | Expr
\end{lstlisting}
    \else
\begin{lstlisting}
Expr ::=
    Loc "=" Expr
  | label ":"       # code label
  | bin | hex | nat
  | name "(" (Arg ("," Arg)*)? ")"
  | "(" Expr ")"
  | "~" Expr
  | "true" | "false"
  | Expr ("&" Expr)+
  | Expr ("|" Expr)+
  | Expr "->" Expr

Loc ::=
    nat ":" regname  # register
  | "*" name         # deref

Arg ::=
    name "=" Expr
  | Expr
\end{lstlisting}
    \fi
    \caption{Litmus Test --- Simplified Expression Grammar}
    \label{fig:toml_expr_syntax}
\end{figure}

\subsubsection{Execution witness}

The test is run by Isla
with a model with no axioms to allow all behaviours.
Executions that satisfy the final state then have graphs produced
(and if multiple, the `interesting' execution is hand-picked for display).

The diagrams contain an `initial state' node which represents all initial writes in the system
(which may or may not be writes of zero).
Threads are then laid out in a row,
with instructions within each Thread box placed in a single column
with \herd{po} (`program-order') going top to bottom.
Multiple events within the same instruction are then aligned horizontally within the same row,
where possible.

Translates are highlighted (in blue) and interesting relations
(\herd{iio}, \herd{po}, \herd{co}, \herd{rf}, \herd{trf}, \herd{fr}, \herd{tfr}, \herd{same-va-page}, and \herd{same-ipa-page})
are shown (and where the relation is transitively closed, we display the transitive reduction of that relation to reduce clutter).
Labels for \herd{po} are elided to reduce clutter.

\subsubsection{Isla output}

The test is run in Isla,
using the strong model
\ifvmsaTestIndexStandalone
(see ESOP'2022 paper draft).
\else
(see App.~\ref{app:models}).
\fi

The generated Isla output that is produced is cut down to just the final line of output,
which contains five key pieces of information:
\begin{itemize}
    \item The test name
    \item Model outcome (allowed or forbidden)
    \item The total number of executions, and how many were allowed.
    \item The total Isla execution time for the test.
\end{itemize}

\subsubsection{Example}

Consider \TEST{CoTW1.inv}. 
It has one thread (Thread~0) and one other code section (Thread~0's EL1 exception vector).
The initial state says that the thread starts from EL0 (from the \asm{PSTATE\.EL} register),
with \asm{R1} (aka \asm{X1}) containing the 64-bit virtual address named \cc{x},
\asm{R2} containing the 64-bit level~3 descriptor which the initial pagetable setup uses to translate \asm{y},
\asm{R3} containing the 64-bit virtual address of the location that contains the level~3 descriptor used in translating \asm{x},
and finally that the thread's EL1 vector base address (\asm{VBAR}) is at \cc{0x1000}.

The pagetable setup has two virtual addresses (\asm{x} and \asm{y}),
with one physical address (\asm{pa1}).
Initially \asm{x} is unmapped,
and \asm{y} maps to \asm{pa1}
where \asm{pa1} is initially 1.
The page containing the vector table (starting at \asm{0x1000}) is identity mapped as executable.

During execution of the test it is expected that at some point \asm{x} may map to \asm{pa1},
and so there is a \asm{x ?-> pa1} constraint.
Without this constraint Isla will not generate any executions that involve translating \asm{x}
resulting in a translation to \asm{pa1}.
(In fact,  in this instance Isla will fail on symbolic evaluation of \asm{STR X2,[X3]} as the write is unsatisfiable.)

The exception handler of interest is located at \asm{0x1400},
this is at \asm{VBAR+0x400}.
An offset of \asm{0x400} represents a synchronous exception from a lower exception level.
The handler overwrites \asm{X0} with 0, to mark that an exception has occured,
and then does an exception-return to the next-instruction-address (i.e. \asm{ELR+4}).

The final state asserts that \asm{0:R0=1},
that is that seeing \asm{X0} being 1 would imply a relaxed execution of the above test.
See the \TEST{CoTW1.inv} section for an explanation of which relaxed behaviour(s) this outcome corresponds to.

The diagram then shows such an execution.

Finally we see that the test is forbidden by our strong model,
that Isla generates 2 candidate executions for this test (and neither are allowed),
and that it took 38589 miliseconds for Isla to run the test (just under 40 seconds).

\newsavebox{\CoTWoneinvbox}
\savebox{\CoTWoneinvbox}{\newlength{\CoTWonedotinvH}
\setlength{\CoTWonedotinvH}{0cm}
\newsavebox{\CoTWonedotinvI}
\begin{lrbox}{\CoTWonedotinvI}
\begin{lstlisting}[language=AArch64,showlines=true]
LDR X0,[X1]
STR X2,[X3]
\end{lstlisting}
\end{lrbox}
\newlength{\CoTWonedotinvJ}
\settowidth{\CoTWonedotinvJ}{\usebox{\CoTWonedotinvI}}
\addtolength{\CoTWonedotinvH}{\CoTWonedotinvJ}
\newsavebox{\CoTWonedotinvK}
\begin{lrbox}{\CoTWonedotinvK}
\begin{lstlisting}[language=AArch64,showlines=true]
0x1400:
MOV X0,#0
MRS X13,ELR_EL1
ADD X13,X13,#4
MSR ELR_EL1,X13
ERET
\end{lstlisting}
\end{lrbox}
\newlength{\CoTWonedotinvL}
\settowidth{\CoTWonedotinvL}{\usebox{\CoTWonedotinvK}}
\addtolength{\CoTWonedotinvH}{\CoTWonedotinvL}
\newsavebox{\CoTWonedotinvM}
\begin{lrbox}{\CoTWonedotinvM}
\begin{lstlisting}[language=IslaPageTableSetup,showlines=true]
    physical pa1;
    x |-> invalid;
    x ?-> pa1;
    y |-> pa1;
    *pa1 = 1;
    identity 0x1000 with code;
\end{lstlisting}
\end{lrbox}
\newlength{\CoTWonedotinvN}
\settowidth{\CoTWonedotinvN}{\usebox{\CoTWonedotinvM}}
\newsavebox{\CoTWonedotinvO}
\begin{lrbox}{\CoTWonedotinvO}
\begin{minipage}{\CoTWonedotinvH}
\vphantom{$\vcenter{\hbox{\rule{0pt}{1.8em}}}$}Initial state:\\
\lstinline[language=IslaLitmusExp]|PSTATE.EL=0b00|\\
\lstinline[language=IslaLitmusExp]|PSTATE.SP=0b0|\\
\lstinline[language=IslaLitmusExp]|R1=x|\\
\lstinline[language=IslaLitmusExp]|R2=desc3(y,page_table_base)|\\
\lstinline[language=IslaLitmusExp]|R3=pte3(x,page_table_base)|\\
\lstinline[language=IslaLitmusExp]|VBAR_EL1=0x1000|\\
\end{minipage}
\end{lrbox}
\begin{tabular}{|l|l|}
  \multicolumn{2}{l}{\textbf{AArch64} \lstinline[language=IslaLitmusName]|CoTW1.inv|}\\
  \hline
  \multirow{6}{*}{\begin{minipage}{\CoTWonedotinvN}\vphantom{$\vcenter{\hbox{\rule{0pt}{1.8em}}}$}Page table setup:\\\usebox{\CoTWonedotinvM}\end{minipage}} & \cellcolor{IslaInitialState}{\usebox{\CoTWonedotinvO}}\\
  \cline{2-2}
  & Thread 0\\
  \cline{2-2}
  & \usebox{\CoTWonedotinvI}\\
  \cline{2-2}
  & thread0 el1 handler\\
  \cline{2-2}
  & \usebox{\CoTWonedotinvK}\\
  \cline{2-2}
  & Final state: \lstinline[language=IslaLitmusExp]|0:R0=1|\\
  \hline
\end{tabular}
}
\newcommand{\CoTWoneinv}{
    \vspace*{2mm}
    \hspace*{-20mm}\adjustbox{max width=\textwidth}{\usebox{\CoTWoneinvbox}}
    {\begin{center}\testdiagram{CoTW1.inv}{}\end{center}}

    \begin{tabular}{ll}
        Model & Result \\
        \hline \\
        Base & \testresult{Strong}{CoTW1.inv} \\
        ETS & \testresult{ETS}{CoTW1.inv} \\
    \end{tabular}
}
\CoTWoneinv{}

\mynewpage
\Mysubsection{Aliasing}\label{sec:aliasing}

\Mysubsubsection{Coherence}

Arm's notion of \emph{coherence} gives a fixed total order per location of all writes to that location.
With virtual memory, that becomes a total order per \emph{physical address} location.

\begin{TESTGROUP}{CoRR}
\TESTPARAGRAPH{CoRR0.alias+po}
forbid

This is the classic coherence shape.
Here, we ask whether two reads with different VAs but which map to the same PA
are allowed to re-order with respect to each other if they read from different writes.  For Arm, they are not, as coherence is with respect to physical addresses.

\testfig{}

\mynewpage
\TESTPARAGRAPH{CoRR2.alias+po}
forbid

This is another standard variant of CoRR, adapted to physical memory. 

\testfig{}

\end{TESTGROUP}

\mynewpage
\TESTPARAGRAPH{CoWR.alias}
forbid

If one writes to one VA, and reads with another that is mapped to the same PA,
must the read read-from the program-order preceding write of the same PA,
or something newer,
regardless of the second VA? For Armv8-A, yes.

\testfig{}

\mynewpage
\Mysubsubsection{Write-Forwarding}

\TESTPARAGRAPH{PPOCA.alias}
allow

\xxtodiscuss{}

Can a load from one virtual address have its value forwarded from a store to distinct VA that is mapped to the same PA, on a speculative branch?

Our model says yes.

\testfig{}

\mynewpage
\Mysubsubsection{Out-of-order reads}

\myenlargethispage{\baselineskip}
\TESTPARAGRAPH{RSW.alias}
allow

If two reads from different VAs which translate to the same PA
read from the same write, %
they can be re-ordered with respect to program-order.

\testfig{}

\mynewpage
\TESTPARAGRAPH{RDW.alias} forbid

\xxtodiscuss{}

If two loads of different VAs which translate to the same PA
read from different writes, then can they be re-ordered?

Our model says no.

\testfig{}

\mynewpage
\TESTPARAGRAPH{CoWW.alias} forbid

\xxtodiscuss{}

Should the coherence-order of writes respect program-order in the same thread
even if they are to different VAs?

Our model says yes, and forbids the following CoWW.alias test.

\testfig{}

\mynewpage
\TESTPARAGRAPH{MP.alias3+rfi-data+dmb}
allow

This shows thread-local forwarding of a write to a read with distinct VA but the same PA, in a potentially non-speculative path.

\testfig{}

\mynewpage
\Mysubsection{Writing new entries}\label{sec:map}

\Mysubsubsection{Translation tables as data memory}

Writes to the translation tables are treated as completely normal
writes to memory as far as normal reads are concerned, like any other location:
they can be re-ordered, cached,  and take part in coherence
as far as their memory attributes allow. 
We assume here that all reads and writes are to `normal' cacheable memory.

\TESTPARAGRAPH{CoWR.inv}
forbid

Writing a new entry to the page-table then
loading the location again performs a normal data memory read.

We do not adapt all of the standard ``user'' data memory tests here
with translation tables as memory locations.
Instead, we just give one representative co-shaped example.

\testfig{}

\mynewpage
\Mysubsubsection{Making a new entry}

If a VA is currently unmapped
(and that has been fully synchronized with sufficient TLBI and barrier instructions), 
then, to produce a new virtual-to-physical mapping,
all that is needed is to simply write to the physical location that
contains the invalid entry for that VA.

To ensure that the new entry is seen by the same processor,
the pipeline must be flushed with an \asm{ISB} or other
\emph{context-synchronizing} event.
Without this, the processor can re-order (or perhaps even speculatively perform)
the translation.

\mynewpage
\myenlargethispage{\baselineskip}
\begin{TESTGROUP}{CoWTf}
\TESTPARAGRAPH{CoWTf.inv+po}
allow

If a thread writes to a page table entry initially containing an invalid descriptor,
and the translation of the address of the next instruction uses the page table entry,
then the translate is allowed to see the old, invalid descriptor.

To detect this, we install a handler for synchronous aborts which
writes \asm{0} to \asm{X2}
before incrementing the \asm{ELR} to the next instruction address and performing
an exception-return.

If the final state sees \asm{X2=1}, then we know the load read-from the new physical location,
but if it saw \asm{X2=0}, then it must have been caused by a translation-fault.

The role of \asm{y} in this litmus test is to make it possible to succinctly describe the new descriptor for \asm{x}.

\TODOLATER{BS: explain process a bit better}

\testfig{}

\mynewpage
\TESTPARAGRAPH{CoWTf.inv+dsb-isb}
forbid

If there is a \asm{DSB; ISB} interposed in between the overwriting of the invalid descriptor with the valid descriptor
and the translation, then the translation is required to see a write no older than that of the valid descriptor,
as the \asm{DSB; ISB} causes a pipeline flush.

\testfig{}

\end{TESTGROUP}

\mynewpage
\Mysubsubsection{Creating a new entry for another core}
If two CPUs are using the same (or overlapping) translation tables,
then, necessarily, writes to the translation table by one CPU can be visible to the other.

\mynewpage
\TESTPARAGRAPH{S.T+dmb+po}
forbid

In this \cc{S}-shaped test,
Thread~0 writes some data and then
gives Thread~1 a new mapping.

If Thread~1 sees the mapping,
then the program-order-later store must wait for the translation to finish before propagating to memory.

\testfig{}

\mynewpage
\begin{TESTGROUP}{MP.RTf.inv}
\TESTPARAGRAPH{MP.RTf.inv+dmb+dsb-isb}
forbid

Note that Thread 0 only has a \asm{DMB SY};
in fact, any ordered-before relation here would suffice.
Thread 1 requires the pipeline flush, as described above. %

\testfig{}

\mynewpage
\TESTPARAGRAPH{MP.RTf.inv+dmbs}
allow
(forbid with ETS)

The \asm{DSB; ISB} is required for the base architecture,
as illustrated by the this test.

However,
if the implementation has the ETS optional feature (``Enhanced Translation Synchronization''),
then this test is forbidden.
This is because ETS ensures that a translation-table-walk which results in a translation-fault
(that is, one that reads an invalid entry)
is ordered-after any memory event which would be ordered-before the read/write of any load/store (as appropriate)
in the place of the instruction which generated the translation-fault.
\TODOLATER{BS: talk about how it's a syntactic subset of ordered-before (?)}

\testfig{}

\mynewpage
\TESTPARAGRAPH{MP.RTf.inv+dmb+ctrl-isb}
forbid?

\xxtodiscuss{}

\testfig{}

\mynewpage
\TESTPARAGRAPH{MP.RTf.inv+dmb+addr}
forbid?

\xxtodiscuss{}

\testfig{}

\mynewpage
\TESTPARAGRAPH{MP.RTf.inv+dmb+po} allow

Even with ETS, program-order alone is not enough to ensure that Thread 1 sees the write of the valid descriptor.

\testfig{}

\mynewpage
\TESTPARAGRAPH{MP.RTf.inv.EL1+dsb-tlbiis-dsb+po} allow

This is a variant of message passing where
\asm{x} is initially not mapped,
Thread 0 maps \asm{x}, performs a \asm{DSB;TLBI;DSB}, and writes to the flag,
and Thread 1 reads the flag and reads \asm{x}.
Because there is only program-order in Thread 1, the reads can be completely
reordered, and thus the second read can happen entirely before the TLBI.

\testfig{}

\mynewpage
\TESTPARAGRAPH{MP.RTf.inv.EL1+dsb-tlbiis-dsb+dmb} forbid

A fault inherits the order that the corresponding memory access would have had if it had not faulted.
(With ETS, its translates also inherit the order, making this test forbidden more directly.)
Moreover, the pipeline effect of the broadcast TLBI enforces that a memory access and its translate-reads are `atomically' ordered with respect to the TLBI:
they are either both ordered-before it, or both ordered-after it.
Therefore, because the counterfactual load of \asm{x} in Thread 1 is ordered after the load of the flag \asm{y} by the \asm{DMB SY},
the translate is also ordered after it.
Therefore, if Thread 1 sees that the flag \asm{y} is set to 1,
then the translate is guaranteed to translate-read something at least as new as the new, valid mapping that Thread 0 wrote.

\testfig{}

\mynewpage
\TESTPARAGRAPH{MP.RTf.inv.EL1+dsb-tlbiis-dsb+addr} forbid

\testfig{}

\mynewpage
\TESTPARAGRAPH{MP.RTf.inv.EL1+dsb-tlbiis-dsb+data} forbid

\testfig{}

\mynewpage
\TESTPARAGRAPH{MP.RTf.inv+dmb+data} allow?

The version with just a \asm{DMB} is not enough.

\testfig{}

\mynewpage
\TESTPARAGRAPH{MP.RTf.inv.EL1+dsb-tlbiis-dsb+ctrl} forbid

\testfig{}

\mynewpage
\TESTPARAGRAPH{MP.RTf.inv.EL1+dsb-tlbiis-dsb+dsb-isb} forbid

\testfig{}

\mynewpage
\TESTPARAGRAPH{MP.RTf.inv.EL1+dsb-tlbiis-dsb+ctrl-isb} forbid

\testfig{}

\mynewpage
\TESTPARAGRAPH{MP.RTf.inv.EL1+dsb-tlbiis-dsb+poap} forbid

\testfig{}

\end{TESTGROUP}

\mynewpage
\TESTPARAGRAPH{LB.TT.inv+pos} forbid

This is a variant of load buffering
where the first thread's store writes the descriptor that the
translate for the second thread's load translate-reads from,
and symmetrically.
This kind of self-satisfying cycle would be very problematic,
and this test is forbidden.

\testfig{}

\begin{TESTGROUP}{S.RTf.inv.EL}

\mynewpage
\TESTPARAGRAPH{S.RTf.inv.EL1+dsb-tlbiis-dsb+data} forbid

\testfig{}

\mynewpage
\TESTPARAGRAPH{S.RTf.inv.EL1+dsb-tlbiis-dsb+ctrl} forbid

\testfig{}

\mynewpage
\TESTPARAGRAPH{S.RTf.inv.EL1+dsb-tlbiis-dsb+dmb} forbid

\testfig{}

\mynewpage
\TESTPARAGRAPH{S.RTf.inv.EL1+dsb-tlbiis-dsb+popl} forbid

\testfig{}

\mynewpage
\TESTPARAGRAPH{S.RTf.inv.EL1+dsb-tlbiis-dsb+poap} forbid

\testfig{}

\end{TESTGROUP}

\mynewpage
\Mysubsubsection{Coherence}
Similarly to our previous questions about Instruction$\leftrightarrow$Data coherence,
we can ask questions about Translation$\leftrightarrow$Data coherence:

\begin{enumerate}
    \item If a translation-table-walk reads-from a write, must a later translation-table-walk that reads the same location read-from the same write or something coherence-newer?
        (Translation$\rightarrow$Translation Coherence).
    \item If a translation-table-walk reads-from a write, must a later load/store that reads/writes that location read-from something at least that new, or write something coherence-after it?
        (Data$\rightarrow$Translation Coherence).
    \item If a load reads-from a write, must a later translation-table-walk which reads that location read-from that write or something newer?
        (Translation$\rightarrow$Data Coherence).
\end{enumerate}

\mynewpage
\TESTPARAGRAPH{CoTW1.inv}
forbid

Translations cannot read-from writes which appear program-order after the instruction that does the translation.

\CoTWoneinv{}

\ifwiptests
\mynewpage
\TESTPARAGRAPH{CoTfWinv1+si}
forbid

A special case of the above is if the translation is from the same instruction:
if a write writes to the translation-table entry which translates itself,
that translation cannot read-from itself.

In this test,
there cannot be a translation-fault from the store to the translation table.
\fi

\mynewpage
\myenlargethispage{10\baselineskip}
\begin{TESTGROUP}{CoTTf}
\TESTPARAGRAPH{CoTTf.inv+dsb-isb}
forbid

Here, Thread~0 makes a new mapping, and
Thread~1 observes that new mapping
by performing a translation using it,
and then later tries to load that same location.
If the first read is translated using the new entry,
then the second one is not allowed to fault.

Note this test's handler writes to \asm{X2},
so the test saves it into \asm{X0} after the first load.

This suggests a kind of translation$\rightarrow$translation coherence.
In general, you do observe such coherence when TLB-misses (and therefore walks in memory) occur.
However, the \TEST{CoTfT+dsb-isb} test (later in this document) shows that this is not guaranteed for all translations.

\testfig{}

\mynewpage
\TESTPARAGRAPH{CoTTf.inv+po}
allow

Same as above,
but with no explicit order between the two loads.

This is allowed.

\testfig{}

\end{TESTGROUP}

\mynewpage
\TESTPARAGRAPH{CoTfT+dsb-isb}
allow

This test is in contrast to the previous test,
where the VA used by the loads in Thread~1 are instead valid from the start,
and Thread 0 attempts to `break' the entry by writing an invalid descriptor (e.g. \asm{0}) to the entry.
In contrast to the previous test, this one does not obey the translation$\leftrightarrow$translation coherence principle.

\testfig{}

\mynewpage
\begin{TESTGROUP}{CoRpteTf.inv}
\TESTPARAGRAPH{CoRpteTf.inv+dsb-isb}
forbid

In this test, Thread 0 writes a new mapping,
which Thread 1 reads-from with a load
before trying to access the location mapped by that entry.

So long as the later translation is ordered after the read
(with a context-synchronizing event, or if with ETS then any ordered-before),
then the translation must see the new mapping too.

This implies a kind of Translation$\rightarrow$Data coherence.
Note that this only applies going Invalid$\rightarrow$Valid;
removing a mapping does not guarantee this coherence and TLB maintenance is required (c.f. \TEST{CoRT+dsb-isb})

\testfig{}

\mynewpage
\TESTPARAGRAPH{CoRpteTf.inv+dsb}
allow (unless ETS)

Same as previous, but without the \asm{ISB}.
Allowed (unless ETS, then forbidden).

\testfig{}
\end{TESTGROUP}

\mynewpage
\TESTPARAGRAPH{CoRpteT+dsb-isb}
allow

Here, Thread 0 \emph{unmaps} a location,
and Thread 1 loads the translation table entry containing the invalid descriptor.
Given Thread 1 read the invalid descriptor,
is a `later' translation of the possibly-unmapped location
required to fault?

On Arm, no.
The TLB can cache the old mapping and then the later translation can read that cached value.

\testfig{}

\mynewpage
\TESTPARAGRAPH{CoRpteT.EL1+dsb-tlbi-dsb-isb}
forbid

In order to forbid the case in the previous test,
an extra \asm{TLBI} (with correct synchronization) must be inserted to remove those cached entries
before trying to load the possibly-unmapped location.

\testfig{}

\mynewpage
\TESTPARAGRAPH{CoRpteT.EL1+dsb-tlbi-dsb}
allow

Same as previous test, but without the \asm{ISB}.

\testfig{}

\mynewpage
\TESTPARAGRAPH{CoTRpte.inv+dsb-isb}
forbid

Here, Thread 0 makes a new mapping,
and Thread 1 reads-from that entry during a translation-table-walk
before trying to load the entry itself.

If Thread 1 sees the new mapping, then `later' loads of the translation table
must see that entry or something newer.
This is Data$\rightarrow$Translation coherence.

\testfig{}

\mynewpage
\begin{TESTGROUP}{CoTfRpte}
\TESTPARAGRAPH{CoTfRpte+dsb-isb}
forbid

Here, Thread~0 unmaps a location
and Thread~1 translates that location getting a translation fault.
Is a later load of the translation-table entry in Thread~1 required to see the
write of the invalid descriptor or something newer?
On Armv8-A, yes, 
since the fault is from a TLB-miss it reads from memory
and therefore respects the total coherence-order for that entry.

\testfig{}

\mynewpage
\TESTPARAGRAPH{CoTfRpte+po}

This test cannot exist, because a fault causes an exception,
which the test has to take,
and from which the test has to return,
both of which cause synchronisation.

\mynewpage
\TESTPARAGRAPH{CoTfRpte+eret}
forbid

\testfig{}

\end{TESTGROUP}

\mynewpage
\begin{TESTGROUP}{CoTfW.inv}
\TESTPARAGRAPH{CoTfW.inv+dsb-isb}
forbid

Here, Thread~0 writes a new mapping, and
Thread~1 then sees this mapping with a translation-table walk,
before overwriting it with another entry.
Can this later write be coherence-before the original write?

Our model forbids this,
requiring that writes pass the point of coherence before being visible to translations.

\testfig{}

\mynewpage
\TESTPARAGRAPH{CoTfW.inv+po}
forbid

This is like the previous test,
except that there is no barrier between the instruction which faults
and the program-order-later write.

However,
this write is still ordered after the translation-table-walk.
This is due to the write being unable to propagate until
after the translation-table-walk has finished and the fault is known,
as stores cannot propagate while speculative.

\testfig{}

\end{TESTGROUP}

\mynewpage
\TESTPARAGRAPH{PPODA.RT.inv} allow?

\xxtodiscuss{}

Can writes be forwarded to translation-table-walks in general?

\testfig{}

\mynewpage
\Mysubsubsection{Write-forwarding}

\TESTPARAGRAPH{MP.RT.inv+dmb+ctrl-trfi}
forbid

\xxtodiscuss{}
It seems that write forwarding should not allowed down speculative paths.
Allowing it could be problematic,
as forwarding in a non-taken path from a write to the translate for a read
would make it possible for the read to read from device memory,
which the device could observe.

\testfig{}

\mynewpage
\TESTPARAGRAPH{MP.RT.inv+dmb+addr-trfi}
forbid

\testfig{}

\mynewpage
\Mysubsubsection{Address dependencies}

\paragraph{MP.RTf.inv+dmb+addr}
forbid

(See \TEST{MP.RTf.inv+dmb+addr} from earlier)

\mynewpage
\Mysubsubsection{Data dependencies}

\paragraph{MP.RTf.inv+dmb+data}
allow

The data dependency (unlike an address dependency)
is to the memory access itself, not its translates,
so a translate can happen much earlier.

(See \TEST{MP.RTf.inv+dmb+data for diagrams/results})

\mynewpage
\Mysubsection{Unmapping memory and TLB invalidation}\label{sec:unmap}

In the previous section, we explored the sequences required to successfully create,
or \emph{map} a given virtual address in the current address space.

Removing a mapping (or \emph{unmapping}) is more subtle:
simply overwriting the entry back to zero is not enough,
due to caching of those old entries in the thread's TLB.\@

\newpage
\Mysubsubsection{Same-thread unmap}

\begin{TESTGROUP}{CoWinvT}
\TESTPARAGRAPH{CoWinvT+dsb-isb}
allow

\testfig{}

\mynewpage
\TESTPARAGRAPH{CoWinvT.EL1+dsb-tlbi-dsb}
allow

This is the `break' side of break-before-make,
but without an \asm{ISB} at the end on the same thread,
so it is not guaranteed that the po-later translations for this core are restarted.

\testfig{}

\mynewpage
\TESTPARAGRAPH{CoWinvT.EL1+dsb-tlbiis-dsb}
allow

This is the same as the previous test, but with a broadcast TLBI.

\testfig{}

\mynewpage
\TESTPARAGRAPH{CoWinvT.EL1+dsb-tlbiis-dsb-isb}
forbid

This (similarly to the previous test) is the `break' side of break-before-make,
but now including the \asm{ISB} at the end on the same thread,
so this time it is guaranteed that the po-later translations for this core are restarted.

\testfig{}

\mynewpage
\TESTPARAGRAPH{MP.RT.EL1+dsb-tlbiis-dsb+dsb-isb}
forbid

This is the `break' side of break-before-make,
now with a message passing to another core.

\testfig{}

\end{TESTGROUP}

\mynewpage
\TESTPARAGRAPH{RBS+dsb-tlbiis-dsb}
forbid

This `read-broken-secret' (RBS) test is a fundamental test for the security guarantees `break' gives you.
Thread~0 unmaps a VA
before writing to the original PA (for example, here, through an alias).
Thread~1 attempts to read that VA.

It is allowed for Thread~1 to see the translation-fault,
or to translate the VA using the old mapping and see the old write,
but it is forbidden to translate the VA to the PA and see the new write.

This ensures that, once a mapping to a location is `broken',
later writes to that location are `secret'
for any cores that were using that VA.

\testfig{}

\mynewpage
\Mysubsection{More TLB invalidation}\label{sec:tlbi}

\Mysubsubsection{TLBI-pipeline interactions}

Broadcast TLBI variants,
also called TLB \emph{shootdowns},
not only clean the cached entries in the TLB,
but also go into the pipeline of other cores to invalidate any unfinished instructions
that had already started using any of the old translations.

\mynewpage
\begin{TESTGROUP}{MP.RT.EL1}
\TESTPARAGRAPH{MP.RT.EL1+dsb-tlbiis-dsb+dmb} forbid

The message-passing ensures that the read of \asm{x} is ordered after the broadcast \asm{TLBI},
and the translate of \asm{x} has to be on the same side as the access it translates for,
so it has to be after the \asm{TLBI} too,
and therefore has to fault.

\testfig{}

\end{TESTGROUP}

\mynewpage
\myenlargethispage{10\baselineskip}
\Mysubsubsection{Thread-local TLBIs}
In the previous section, we described broadcast TLB maintenance (aka `TLB shootdowns').
But Arm also have thread-local TLB-maintenance instructions.
These thread-local TLBIs can be used to clear thread-local context information
(such as local TLB-cached entries),
and can be combined together to emulate a broadcast TLB-maintenance
by interrupting other cores and performing a thread-local TLBI.

\TESTPARAGRAPH{CoWinvT.EL1+dsb-tlbi-dsb-isb}
forbid

The effect of a thread-local TLBI is enough for the thread executing it.

\testfig{}

\mynewpage
\TESTPARAGRAPH{MP.RT.EL1+dsb-tlbi-dsb+dsb-isb}
allow

The effect of a thread-local TLBI is indeed thread-local,
and does not get carried over to another thread by message-passing.

\testfig{}

\mynewpage
\myenlargethispage{10\baselineskip}
\TESTPARAGRAPH{MP.RT.EL1+dsb-shootdown-dsb+dsb-isb}
forbid

\xxtodiscuss{}

A broadcast TLBI can be emulated by performing
a thread-local TLBI on each core, with sufficient synchronization between them.

In the following test,
we take \TEST{MP.RT.EL1+dsb-tlbiis-dsb+dsb-isb} and split the broadcast TLBI over
many threads.
Thread 0 `breaks' the location and then sends messages to each core
requesting it perform the TLB maintenance locally.

Note that to correctly emulate the behaviour of the broadcast TLBI,
each core must perform an \asm{ISB} (or other context-synchronizing event)
to get the pipeline effects of the TLBI.
This requirement is slightly stronger than the TLBI semantics,
also flushing unrelated accesses.

\testfig{}

\mynewpage
\myenlargethispage{20\baselineskip}
\Mysubsubsection{Multiple locations}

\TESTPARAGRAPH{MP.RTT.EL1+dsb-tlbiis-tlbiis-dsb+dsb-isb}
forbid

Here, we invalidate \emph{two} different VAs, and perform their TLBIs
together in the same thread,
effectively allowing concurrent execution of the two TLBIs on the same core.

\testfig{}

\mynewpage
\myenlargethispage{10\baselineskip}
\Mysubsection{Stage~1 Re-mapping and break-before-make}\label{sec:s1bbm}

\Mysubsubsection{Break-before-make}

\begin{TESTGROUP}{BBM}
\TESTPARAGRAPH{BBM+dsb-tlbiis-dsb}
allow

This is the smallest example of a safe change of output address mapping,
using the `break-before-make' (BBM) pattern.

Thread~1 loads a fixed VA, and
Thread~0 tries to re-map that VA from the initial PA to a new one.

To do this safely,
Arm prescribe a ``break-before-make'' sequence
to ensure that the other threads will not ever see both the new and old mappings at the same time.
Instead, the mapping must be `broken' (unmapped) and cleaned between the two states.

This test is the minimum required to correctly change OA (``output address'') for a given mapping.

\testfig{}

\mynewpage
\TESTPARAGRAPH{BBM.Tf+dsb-tlbiis-dsb}
allow

This illustrates the other allowed outcome of the previous test:
the correct use of break-before-make ensures that the change of output address is safe,
but does not guarantee that the new page table entry is seen.
Thread~1 sees a translation-fault from the transient invalid entry during the break-before-make sequence.

\testfig{}

\end{TESTGROUP}

\mynewpage
\myenlargethispage{10\baselineskip}
\begin{TESTGROUP}{MP.BBM}
\TESTPARAGRAPH{MP.BBM1+dsb-tlbiis-dsb-dsb+dsb-isb}
forbid

In this test,
Thread~0 break-before-makes a new mapping,
and then synchronises with Thread~1 with a message pass.

While this is a slightly unusual setup,
as one would mostly expect break-before-make to happen concurrently with the other thread, rather than be synchronised-before it,
this is still an interesting test to explore the architecture.

Arm forbid both the translation-fault and the translation with the old entry.

\testfignumbered{2}

\mynewpage
\TESTPARAGRAPH{MP.BBM1+dsb-tlbiis-dsb-dsb+ctrl-isb} forbid

\testfig{}

\end{TESTGROUP}

\mynewpage
\Mysubsection{Translation-table-walk ordering}\label{sec:ttwordering}

\Mysubsubsection{Inter-instruction ordering}

Typically, translations of separate instructions are not ordered with respect to each other:
simply having program-order between them introduces no strength.
Earlier, we saw that even same-VA did not give strength (\TEST{CoTTf.inv+po})

\begin{TESTGROUP}{MP.TT.inv}
\TESTPARAGRAPH{MP.TTf.inv+dsb+po}
allow

The translates on the reader side are not ordered.

\testfig{}

\mynewpage
\TESTPARAGRAPH{MP.TTf.inv+dsbs}
allow (unless ETS, then forbid)

The translates on the reader side are not ordered,
as the \asm{DSB SY} does not order them by itself
(it needs an \asm{ISB}),
except with ETS, which would provide order with the faults.

\testfig{}

\mynewpage
\TESTPARAGRAPH{MP.TTf.inv+dsb+dsb-isb}
forbid

The translates on the reader side are ordered by the \asm{DSB SY;ISB} combination.

\testfig{}

\mynewpage
\TESTPARAGRAPH{MP.TTf.inv+dsb+ctrl-isb}
forbid

\testfig{}

\ifwiptests
\mynewpage
\TESTPARAGRAPH{MP.TTf.inv+dsb+addr}
forbid

\testfig{}
\fi

\mynewpage
\TESTPARAGRAPH{MP.TTf.inv+dmb+dsb-isb}
forbid

\xxtodiscuss{How is that compatible with the statement that ``\asm{DSB}s make writes visible to translation table walks''?}

The \asm{DMB SY} on the writer thread is enough
(a \asm{DSB SY} is not needed),
as it is only here to order the writes as normal writes.

\testfig{}

\ifwiptests
\mynewpage
\TESTPARAGRAPH{MP.TTf.inv+dmb+ctrl-isb}
\todojp{???}

\testfig{}
\fi

\mynewpage
\TESTPARAGRAPH{MP.TTf.inv+dmb+po}
allow

Note that although the final translation-fault in the second thread must be ordered-before
the write of the valid entry in Thread 1
(as the fault must come from a non-TLB read of memory),
the translation-fault is not necessarily ordered-after the initial translate
or even the first load's read.

\testfig{}

\mynewpage
\TESTPARAGRAPH{MP.TTf.inv.EL1+dsb-tlbiis-dsb+po}
allow

Program-order between the loads does not induce order between the translates of the loads.

\testfig{}

\mynewpage
\TESTPARAGRAPH{MP.TTf.inv.EL1+dsb-tlbiis-dsb+dsb-isb}
forbid

A \asm{DSB; ISB} between the loads does not induce order between the translates of the loads.

\testfig{}

\end{TESTGROUP}

\mynewpage
\Mysubsubsection{Multi-level translations}

\begin{TESTGROUP}{ROT.inv}
\TESTPARAGRAPH{ROT.inv+dsb}
forbid

In this \TEST{ROT} test (``reorder translation''),
Thread~0 writes to the leaf entry of a fresh (unused) translation-table,
and then replaces an (initially invalid) leaf higher in the table
with a new table entry which points to the freshly created table.

Thread~1 then tries to load an address that would use this new freshly made entry.
If the individual accesses during the translation-table-walk are allowed to re-order,
then it would be possible for Thread~1 to see the updated table
but still see the old leaf entry.

This must be forbidden, %
requiring that the translation-table-walk happens `in-order'
and the ordering on Thread~0 ensures the two writes are visible to the walker in that order.

The exception handler code records which translation level caused the exception.

\testfig{}

\mynewpage
\TESTPARAGRAPH{ROT.inv+dmbst}
forbid

This is like the previous test,
but with a much weaker barrier between the two writes.

This is also forbidden:
any respected ordering between the writes would suffice.

\testfig{}

\end{TESTGROUP}

\mynewpage
\TESTPARAGRAPH{LB+data-trfis}
forbid

This is a variant of LB+datas+WW.

\testfig{}

\mynewpage
\TESTPARAGRAPH{LB+addr-trfis}
forbid

This is a variant of LB+datas+WW.

\testfig{}

\ifwiptests
\mynewpage
\TESTPARAGRAPH{RWC.RTfR.inv+addr+dmb}
forbid

\xxtodiscuss{???}

\testfig{}

\mynewpage
\TESTPARAGRAPH{RWC.RTR.EL1+dsb-isb+dsb-tlbi-dsb}
forbid

\testfig{}

\mynewpage
\TESTPARAGRAPH{RWC.RTR.EL1+ctrl-isb+dsb-tlbi-dsb}
\todojp{???}

\testfig{}

\mynewpage
\TESTPARAGRAPH{SB.TfTf.inv+dsb-isbs}
forbid

\testfig{}

\mynewpage
\TESTPARAGRAPH{SB.TfTf.inv+dmb-ctrl-isbs}
\todojp{???}

\testfig{}

\mynewpage
\TESTPARAGRAPH{SB.TfTf.inv+rfi-ctrl-isbs}
\todojp{???}

\testfig{}

\mynewpage
\TESTPARAGRAPH{IRIW.TTf.TTf.inv+addrs}
\todojp{allow? + also the name probably needs fixing}

\todojp{MCA? not so much...}

\testfig{}

\mynewpage
\TESTPARAGRAPH{ISA2.RRTf.inv+dsb+addr+addr}
forbid

\todojp{I don't think this test is that interesting,
because it does not add to an MP shape}

\testfig{}

\mynewpage
\TESTPARAGRAPH{ISA2.RRTf.inv+dsb+data+addr}
\todojp{should be the same as dsb+addr+addr???}

\testfig{}

\mynewpage
\TESTPARAGRAPH{CoWT.inv+po-ctrl-isb+po}
\todojp{unclear}

\testfig{}

\mynewpage
\TESTPARAGRAPH{MP.RT.inv+trfi-data+addr}
allow

\testfig{}

\mynewpage
\TESTPARAGRAPH{CoWTf.inv+rfi-addr}
\todojp{???}

\testfig{}
\fi

\mynewpage
\TESTPARAGRAPH{WRC.TfRT+po+dsb-isb}
allow

\testfig{}

\mynewpage
\TESTPARAGRAPH{WRC.TfRT+dsb-tlbiis-dsb+dsb-isb}
allow

\testfig{}

\mynewpage
\Mysubsection{Multi-copy atomicity}\label{sec:mca}

\xxtodiscuss{This whole section about MCA}

\Mysubsubsection{MCA translation-table-walk}
A fundamental guarantee given by Armv8 over data memory is that of \emph{multi-copy atomicity},
that is,
once a write is seen by one other core, then all cores must see it or something newer if they try read that location.

Translation-table-walks are a kind of read, and we can ask whether those reads come with the same guarantee.
There are multiple ways in which multi-copy atomicity \emph{could} be violated here:
\begin{enumerate}
    \item If a translation-table-walk reads from a write, must another core's translation-table-walk that reads the same entry read-from that same write or something newer?
    \item If a load reads a translation table entry directly, must another core's translation-table-walk that reads that location read-from that write or something newer?
    \item If a translation-table-walk reads from a write, must another core which loads that entry read-from that write or something newer?
\end{enumerate}

We tackle each in turn.

Note that questions about multi-copy atomicity are only interesting under assumptions about coherence,
and due to the lack of data$\rightarrow$translation coherence in the non-TLB-miss case (c.f. \TEST{CoRT+dsb-isb}),
all the tests below start from an invalid state to avoid those uninteresting cases.

\mynewpage
\TESTPARAGRAPH{CoWTf.inv+po-ctrl-isb+po}
forbid?

Can another core see a write propagate to its translation-table-walk
`before' the writer thread's own translation-table-walker does?

In this test Thread~0 writes a new valid descriptor which Thread~1 uses in
a translation-table-walk before sending a message back to Thread~0;
if Thread~0 sees that message can a later translation-table-walk still see the old invalid entry?

\testfig{}

\mynewpage
\myenlargethispage{10\baselineskip}
\begin{TESTGROUP}{WRC.TRTf.inv}
\TESTPARAGRAPH{WRC.TRTf.inv+dsb+dsb-isb}
forbid

In this WRC-shaped test,
Thread~0 writes a new mapping
before Thread~1 translates using that entry.
Thread~1 then messages Thread~2, which then tries to translate the same location that Thread~1 did.
If Thread~2 were allowed to see a translation fault, then this would be a kind of \emph{non-}multi-copy atomic behaviour.

Multi-copy atomicity would forbid this,
requiring that translation-table walks are multi-copy atomic reads of flat memory.

\testfig{}

\mynewpage
\TESTPARAGRAPH{WRC.TRTf.inv+addrs}
forbid

The address-dependency into the instruction which yields a translation fault
ensures that the translation-table-walk happens after the address is determined,
and so the fault is ordered-after the read which its address depends on.

\testfig{}

\mynewpage
\TESTPARAGRAPH{WRC.TRTf.inv+dsbs}
allow (unless ETS, then forbid)

\testfig{}

\mynewpage
\TESTPARAGRAPH{WRC.TRTf.inv+dmbs}
allow (unless ETS, then forbid)

\testfig{}

\mynewpage
\TESTPARAGRAPH{WRC.TRTf.inv+pos}
allow

\testfig{}
\end{TESTGROUP}

\mynewpage

\TESTPARAGRAPH{WRC.TTTf.inv+addrs}
forbid

\testfig{}

\TESTPARAGRAPH{WRC.TTTf.inv+data+addr}
forbid

\testfig{}

\mynewpage
\begin{TESTGROUP}{WRC.RRTf.inv}
\TESTPARAGRAPH{WRC.RRTf.inv+dsb+dsb-isb}
forbid

This test is like the previous one,
except, instead of loading the unmapped location in Thread~1
(therefore reading from the entry during the translation-table walk),
it loads the entry itself directly.

Arm also forbid this test,
as the load in Thread~1 will ensure that the write is visible to the translation-table-walk
that would be performed if Thread~2 had a translation-fault.

\testfig{}

\mynewpage
\TESTPARAGRAPH{WRC.RRTf.inv+dsb+ctrl-isb}
forbid

\testfig{}

\mynewpage
\TESTPARAGRAPH{WRC.RRTf.inv+dsbs}
allow (unless ETS, then forbid)

\testfig{}

\mynewpage
\TESTPARAGRAPH{WRC.RRTf.inv+dmbs}
allow (unless ETS, then forbid)

\testfig{}

\mynewpage
\TESTPARAGRAPH{WRC.RRTf.inv+pos}
allow

\testfig{}

\mynewpage
\TESTPARAGRAPH{WRC.RRTf.inv+addrs}
forbid

\testfig{}

\end{TESTGROUP}

\mynewpage
\myenlargethispage{\baselineskip}
\begin{TESTGROUP}{WRC.TfRR}
\TESTPARAGRAPH{WRC.TfRR+dsb-isb+dsb}
forbid

This is like the previous tests,
except that, here, Thread~1 loads the unmapped address and suffers a translation-fault.
Can Thread~2 load the entry and read-from a write before the break?

As before, Arm forbid this,
enforcing a kind of multi-copy atomicity for translation-table-walks.

\testfig{}

\TESTPARAGRAPH{WRC.TfRR+ctrl-isb+dsb}
forbid

\testfig{}

\mynewpage
\TESTPARAGRAPH{WRC.TfRR+dsbs}
forbid

\testfig{}

\mynewpage
\TESTPARAGRAPH{WRC.TfRR+po+dsb}
forbid

Note that the translation-fault to store ordering is preserved
(See \TEST{CoTfW.inv+po}).

\testfig{}

\mynewpage
\TESTPARAGRAPH{WRC.TfRR+pos}
allow

\testfig{}
\end{TESTGROUP}

\mynewpage
\Mysubsection{Multi-address-space support with ASIDs}\label{sec:multiprocess}

To support systems software with multiple address spaces,
such as operating systems with many concurrently executing processes,
Arm provide two features that allow the hardware and software to manage the
translation tables for these processes more effectively:
\begin{description}
    \item[\S\ref{sec:multiprocess:subsec:ttbrs}] A \asm{TTBR} (``Translation Table Base Register''), which can be changed to point to a new translation table.
    \item[\S\ref{sec:multiprocess:subsec:asids}] ASIDs (``Address Space Identifiers''), which are used to tag TLB entries with their process/address space, to reduce TLB maintenance burden.
\end{description}

\Mysubsubsection{TTBRs}\label{sec:multiprocess:subsec:ttbrs}

The translation tables are stored in normal memory
in a hierarchical tree structure.
In order for the processor to know where the root of this tree is,
it reads a register called the \asm{TTBR} (or ``Translation-Table Base Register'').
Each \emph{translation regime} has its own base register:
\begin{itemize}
    \item[-] \asm{TTBR0\_EL1}:  for Stage~1 translations in the `low' (positive) portion of the address map, from EL1\&0.
    \item[-] \asm{TTBR1\_EL1}:  for Stage~1 translations in the `high' (negative) portion of the address map, from EL1\&0.
    \item[-] \asm{TTBR0\_EL2}:  for Stage~1 translations from EL2.
    \item[-] \asm{VTTBR\_EL2}:  for Stage~2 translations from accesses from EL1\&0.
\end{itemize}

\mynewpage
\myenlargethispage{10\baselineskip}
\Mysubsubsection{ASIDs}\label{sec:multiprocess:subsec:asids}

\TESTPARAGRAPH{CoWinvTa1.1+dsb-tlbiasidis-dsb-eret}
forbid

In this test, a virtual address is unmapped, and only ASID~\#1 is cleaned;
since the thread uses that ASID, the TLB invalidation affects all translations in that thread,
and so the final outcome is forbidden.

\testfig{}

\mynewpage
\myenlargethispage{10\baselineskip}
\TESTPARAGRAPH{CoWinvTa2.1+dsb-tlbiasidis-dsb-eret}
allow

Same as previous,
but invalidating the `wrong' ASID.

\testfig{}

\mynewpage

\Mysubsection{Additional tests, as-yet unsorted}

\TESTPARAGRAPH{MP.RT.inv+dmb+addr-po-msr-isb}
forbid

This exercises the ctxob edges via \\
{\tt speculative ; [MSR]} + \\
 {\tt [ContextChange] ; po ; [CSE]} + \\
 {\tt [CSE] ; instruction-order}

\testfig{}

\mynewpage
\TESTPARAGRAPH{MP.RT.inv+dmb+addr-po-isb}
allow

\testfig{}

\mynewpage
\TESTPARAGRAPH{MP.TR.inv+dmb+msr}
allow

\testfig{}

\mynewpage
\TESTPARAGRAPH{MP.TR.inv+dmb+isb}
allow

\testfig{}

\mynewpage
\TESTPARAGRAPH{MP.TR.inv+dmb+msr-isb}
forbid

\testfig{}

\mynewpage
\TESTPARAGRAPH{SwitchTable.different-asid+eret}
forbid

If the page the page table root is changed, together with an unused ASID,
then a new translation has to read-from a page table entry from the new page table.

\testfig{}

\mynewpage
\TESTPARAGRAPH{SwitchTable.same-asid+eret}
allow

If the page the page table root is changed, together with an already-used ASID,
then a new translation can read-from a page table entry from the old page table.

\testfig{}

\mynewpage
\TESTPARAGRAPH{WDS+po-dsb-tlbiipa-dsb-tlbiis-dsb-eret}

forbid

Write to Different Stages.

If two stages of translation are updated,
then both stages need to be invalidated, in the right order,
to be guaranteed to see the new mapping.

\testfig{}

\mynewpage
\TESTPARAGRAPH{WDS+po-dsb-tlbiipa-dsb-eret}

allow

\testfig{}

\mynewpage
\TESTPARAGRAPH{WDS+dsb-tlbiipa-dsb-eret-po}

\ \todojp{???}

\testfig{}

\mynewpage
\TESTPARAGRAPH{WDS+dsb-tlbiipa-dsb-po-eret}

\ \todojp{???}

\testfig{}

\mynewpage
\TESTPARAGRAPH{WBM+dsb-tlbiis-dsb}

forbid

\testfig{}

\mynewpage
\TESTPARAGRAPH{WBM+dsb-tlbiis-dsb-[dmb]-dmb}

forbid

\testfig{}

Write to Broken Mapping.

This is a variant of RBS, but with a write to the invalidated location, instead of the read from it.

\mynewpage
\TESTPARAGRAPH{CoWTf.inv.EL2+dsb-tlbiipa-dsb-tlbiis-dsb-eret}
forbid

\testfig{}

\ifwiptests
\mynewpage
\Mysubsection{Stage~2}\label{sec:s2bbm}

In general,
Stage~2 translations share all the same relaxed behaviours as the Stage~1 translations do,
as, microarchitecturally, the same translation-table-walk and TLB caching is permitted for Stage~2 translations
as was permitted for Stage~1 translations.

Instead of repeating all the previous tests but at EL2 with Stage~2 translations,
we focus only on the differences here: break-before-make (\S\ref{sec:stage2:subsec:bbm})
and the set of pKVM-derived tests (\S\ref{sec:pkvm},\ref{app:pkvm:tests}).

\newpage
\myenlargethispage{10\baselineskip}
\subsection{Break-before-make}\label{sec:stage2:subsec:bbm}

\begin{TESTGROUP}{MP.RT.EL2}

\TESTPARAGRAPH{MP.RT.EL2+dsb-tlbiipais-dsb+dsb-isb}
allow

Invalidating \emph{just} the intermediate-physical mappings is not enough,
as an old direct virtual-to-physical mapping could have been cached
and used by the translation, even though the underlying intermediate-physical-to-physical mappings has been invalided.

\testfignumbered{3}

\mynewpage
\myenlargethispage{10\baselineskip}
\TESTPARAGRAPH{MP.RT.EL2+dsb-tlbiipais-dsb-tlbiis-dsb+dsb-isb}
forbid

On the other hand, invalidating \emph{both} the intermediate-physical mapping
and the virtual-intermediate physical mapping,
in that order,
is enough.
Because the intermediate mapping is invalidated first,
a spontaneous translation after the invalidate cannot repopulate the TLB with a virtual-to-physical mapping.

\testfig{}

\mynewpage
\myenlargethispage{10\baselineskip}
\TESTPARAGRAPH{MP.RT.EL2+dsb-tlbiis-dsb-tlbiipais-dsb+dsb-isb}
allow

Here, the two TLB invalidations happen,
but in the incorrect order.
Since the intermediate mapping may still exist after the Stage~1 invalidations,
a spontaneous translation at that point could repopulate the TLB with a virtual-to-physical mapping.

\testfig{}

\end{TESTGROUP}

\fi

\ifwiptests
\mynewpage
\myenlargethispage{10\baselineskip}
\Mysubsection{Access Permissions}\label{sec:ap}

\todojp{This shape has a break-before-make violation, because the contents of the old OA and of the new OA differ. Not clear how to fix.}

\begin{TESTGROUP}{CoTT.ro}
\TESTPARAGRAPH{CoTT.ro+po} allow

In this test, Thread~0 overwrites one read-only mapping
with another to a different physical location.
If Thread~1 translates using that new mapping,
it does not mean that later translations cannot see the old mapping anymore.

\testfig{}

\mynewpage
\TESTPARAGRAPH{CoTT.ro+dmb} allow

\testfig{}

\mynewpage
\TESTPARAGRAPH{CoTT.ro+dsb-isb} allow

\testfig{}
\end{TESTGROUP}
\fi

\newpage

\newpage
\myappendix{Full models}{app:models}%

Here we include the entire strong and weak model
(Note that the main relations are the same as those in \S\ref{sec:models}
but may be presented differently).

\subsection{Common}

The models both include a common core,
which defines the shared set of derived relations and events.

\subsubsection{Barriers}

First we define a hierarchy of barriers,
so that an ordering \cc{[e1] ; dmb ; [e2]}
implies \cc{[e1] ; dsb ; [e2]}.

{
  \noindent{}
\begin{lstlisting}[language=cat]
(* define a hierarchy of barriers *)
(*  e.g. if  [e1] ; dmbst ; [e2]  is forbidden
 *  then 
 *           [e1] ; dmbsy ; [e2] 
 *       and [e1] ; dsbsy ; [e2]
 *  are forbidden too
 *)
(* we do not model NSH so pretend it's SY *)
let dsbsy = DSB.ISH | DSB.SY | DSB.NSH
let dsbst = dsbsy | DSB.ST | DSB.ISHST | DSB.NSHST
let dsbld = dsbsy | DSB.LD | DSB.ISHLD | DSB.NSHLD
let dsbnsh = DSB.NSH
let dmbsy = dsbsy | DMB.SY
let dmbst = dmbsy | dsbst | DMB.ST | DSB.ST | DSB.ISHST | DSB.NSHST
let dmbld = dmbsy | dsbld | DMB.LD | DSB.ISHLD | DSB.NSHLD
let dmb = dmbsy | dmbst | dmbld
let dsb = dsbsy | dsbst | dsbld
\end{lstlisting}
}

\subsubsection{Common Core}

Here we define all the relations common to both models,
as they are given to isla-axiomatic:

{
  \noindent{}
\begin{lstlisting}[language=cat]
include "barriers.cat"

(* For each instruction, for each read performed by the translation
   table walk ASL code, we generate one translate-read (T) event. If
   the translation finds an invalid entry, the translate-read event
   will additionally belong to T_f. *)
set T
set T_f

(* T events which are part of a Stage 1 or 2 walk *)
set Stage1
set Stage2

set read_VMID
set read_ASID
relation same-translation

(* A write of an invalid descriptor (an even value) is in W_invalid *)
set W_invalid

(* A write of a valid descriptor (an odd value) is in W_valid *)
set W_valid

(* initial writes *)
set is_IW

(* trf is the translate analogue of rf, such that writes are
   trf-related to translates that read them. trf1 is trf
   restricted to stage 1 reads, and trf2 for stage 2 reads *)
relation trf
relation trf1
relation trf2
relation wco

relation iio
relation instruction-order

(* e1 speculative e2
 * iff e2 was conditionally executed based on the value of e1
 *)
let speculative =
    ctrl
  | addr; po
  | [T] ; instruction-order

(* po-pa relates all events from instructions in instruction-order to the same PA *)
let po-pa = instruction-order & loc

let trfi = trf & int
let trfe = trf \ trfi

(* likewise, tfr is the translate analogue of fr *)
(* we use overlap-loc not loc here to handle the case where
 * multiple translations get merged into a single event *)
relation overlap-loc
let tfr1 = (((trf1^-1); co) \ id) & overlap-loc
let tfr2 = (((trf2^-1); co) \ id) & overlap-loc
let tfr = tfr1 | tfr2
let tfri = tfr & int
let tfre = tfr \ tfri

(* translate and TLBI events with VAs within the same 4K region
 * are related by same-va-page, similarly for IPAs and same-ipa-page *)
relation same-va-page
relation same-ipa-page
relation same-asid-internal
relation same-vmid-internal
relation tlbi-to-asid-read
relation tlbi-to-vmid-read

(* for convenience, derive some handy embeddings of program-order into
 * some subset of the events *)
(*
 * THESE ARE GENERATED BY ISLA
let instruction-order = iio^-1 ; fpo ; iio
let po = [M|F|C] ; instruction-order ; [M|F|C]
let tpo = [T] ; instruction-order ; [T]
*)

(* addr is now derived from the data dependency into the translate-reads
 * if a translation exists *)
(*
 * THIS IS GENERATED BY ISLA
let _addr =
    tdata ; [M]
  | tdata ; iio+ ; [R|W]
  | tdata ; [T_f]
*)

(* CSEs are context-synchronization-events
 * that is an ISB, and taking/returning from an exception *)
(*
let CSE = ISB | TE | ERET
*)

(* Context changing operations
 * are those that write to system registers
 *)
let ContextChange = MSR | TE | ERET

(* fault events come from reads or writes *)
let Fault = TE  (* TakeException, this is overly general *)
let IsTranslationFault = Fault
let IsPermissionFault = Fault
set IsFromR  (* events originating from an Arm LDR instruction *)
set IsFromW  (* events originating from an Arm STR instruction *)
set IsFromReleaseW  (* events originating from an Arm STLR instruction *)

(* Currently we only use same-vmid/same-asid between TLBIs and
 * translates, so this relation defines them in a minimal way.
 *)
let same-vmid = tlbi-to-vmid-read; [read_VMID]; same-translation
let same-asid = tlbi-to-asid-read; [read_ASID]; same-translation

(* A TLBI barriers some writes, making them unobservable to "future" reads from a translation table walk.
 *
 * tseq1 relates writes with TLBIs that ensure their visibility
 * e.g.  `a: Wpte(x) ; b: Wpte(x) ; c: Wpte(x) ; d: TLBI x`
 *  then `c ; tseq1 ; d`
 *  as a, b are no longer visible to translation table walks
 *)
let tlb_might_affect =
    [ TLBI-S1 & ~TLBI-S2 &  TLBI-VA  &  TLBI-ASID & TLBI-VMID] ; (same-va-page & same-asid & same-vmid) ; [T & Stage1]
  | [ TLBI-S1 & ~TLBI-S2 & ~TLBI-VA  &  TLBI-ASID & TLBI-VMID] ; (same-asid & same-vmid) ; [T & Stage1]
  | [ TLBI-S1 & ~TLBI-S2 & ~TLBI-VA  & ~TLBI-ASID & TLBI-VMID] ; same-vmid ; [T & Stage1]
  | [~TLBI-S1 &  TLBI-S2 &  TLBI-IPA & ~TLBI-ASID & TLBI-VMID] ; (same-ipa-page & same-vmid) ; [T & Stage2]
  | [~TLBI-S1 &  TLBI-S2 & ~TLBI-IPA & ~TLBI-ASID & TLBI-VMID] ; same-vmid ; [T & Stage2]
  | [ TLBI-S1 &  TLBI-S2 & ~TLBI-IPA & ~TLBI-ASID & TLBI-VMID] ; same-vmid ; [T]
  | ( TLBI-S1 &            ~TLBI-IPA & ~TLBI-ASID & ~TLBI-VMID) * (T & Stage1)
  | (            TLBI-S2 & ~TLBI-IPA & ~TLBI-ASID & ~TLBI-VMID) * (T & Stage2)
(*  | (TLBI-ALL * T) *)

let tlb-affects =
    [TLBI-IS] ; tlb_might_affect
  | ([~TLBI-IS] ; tlb_might_affect) & int

(* [T] -> [TLBI] where the T reads-from a write before the TLBI and the TLBI is to the same addr 
 * this doesn't mean the T happened before the TLBI, but it does mean there could have been a cached version
 * which the TLBI threw away
 *)
let maybe_TLB_cached =
  ([T] ; trf^-1 ; wco ; [TLBI-S1]) & tlb-affects^-1

(* translation-ordered-before *)
let tob =
  (* a faulting translation must read from flat memory or newer *)
    [T_f] ; tfre
  (* cannot forward past a DSB *)
  | ([T_f] ; tfri ; [W]) & (po ; [dsbst] ; instruction-order)^-1
  (* no forwarding from speculative writes *)
  | speculative ; trfi

let tlb_barriered =
  ([T] ; tfr ; wco ; [TLBI]) & tlb-affects^-1

let obtlbi_translate =
  (* A S1 translation must read from TLB/memory before the TLBI which
   * invalidates that entry happens *)
  [T & Stage1] ; tlb_barriered ; [TLBI-S1]
  (* if the S2 translation is ordered before some S2 write
   * then the S1 translation has to be ordered before the subsequent
   * S1 invalidate which would force the S2 write to be visible
   *
   * this applies to S2 translations during a S1 walk as well
   * here the Stage2 translation is only complete once the TLBI VA which
   * invalidates previous translation-table-walks have been complete *)
  (* if the S1 translation is from after the TLBI VA
   * then the S2 translation is only ordered after the TLBI IPA
   *)
  | ([T & Stage2] ; tlb_barriered ; [TLBI-S2])
     & (same-translation ; [T & Stage1] ; trf^-1 ; wco^-1)
  (* if the S1 translation is from before the TLBI VA,
   * then the S2 translation is ordered after the TLBI VA
   *)
  | (([T & Stage2] ; tlb_barriered ; [TLBI-S2]) ; wco? ; [TLBI-S1])
    & (same-translation ; [T & Stage1] ; maybe_TLB_cached)

(* ordered-before-TLBI *)
let obtlbi =
    obtlbi_translate
  (*
   * a TLBI ensures all instructions that use the old translation
   * and their respective memory events
   * are ordered before the TLBI.
   *)
  | [R|W|Fault] ; iio^-1 ; (obtlbi_translate & ext) ; [TLBI]

(* context-change ordered-before *)
(* note that this is under-approximate and future work is needed
 * on exceptions and context-changing operations in general *)
let ctxob =
 (* no speculating past context-changing operations *)
    speculative ; [MSR]
 (* context-synchronization orders everything po-after with the synchronization point *)
  | [CSE] ; instruction-order
 (* context-synchronization acts as a barrier for context-changing operations *)
  | [ContextChange] ; po ; [CSE]
 (* context-synchronization-events cannot happen speculatively *)
  | speculative ; [CSE]

(* ordered-before a translation fault *)
let obfault =
    data ; [Fault & IsFromW]
  | speculative ; [Fault & IsFromW]
  | [dmbst] ; po ; [Fault & IsFromW]
  | [dmbld] ; po ; [Fault & (IsFromW | IsFromR)]
  | [A|Q] ; po ; [Fault & (IsFromW | IsFromR)]
  | [R|W] ; po ; [Fault & IsFromW & IsFromReleaseW]

(* ETS-ordered-before *)
(* if FEAT_ETS then if E1 is ordered-before some Fault
 * then E1 is ordered-before the translation-table-walk read which generated that fault
 * (but not *every* read from the walk, only the one that directly led to the translation fault)
 *
 * Additionally, if ETS then TLBIs are guaranteed completed after DSBs
 * hence po-later translations must be ordered after the TLBI (D5.10.2)
 *)
let obETS =
    (obfault ; [Fault]) ; iio^-1 ; [T_f]
  | ([TLBI] ; po ; [dsb] ; instruction-order ; [T]) & tlb-affects

include "shows.cat"
\end{lstlisting}
}

\subsection{Strong Model}
\begin{figure}[H]
  \scalebox{0.8}{
  \begin{minipage}{1.2\textwidth}
\begin{lstlisting}[language=cat]
"VMSA strong"
include "cos.cat"
include "barriers.cat"
include "aarch64_mmu_common.cat"

(* observed by *)
let obs = rfe | fr | wco
  (* observing a write through a fetch or translate
   * means the write is now visible to the rest of the system
   * aka {instruction, translation}->data coherence *)
  | trfe

(* dependency-ordered-before *)
let dob =
    addr | data
  | speculative ; [W]
  | addr; po; [W]
  | (addr | data); rfi
  | (addr | data); trfi

(* atomic-ordered-before *)
let aob = rmw
  | [range(rmw)]; rfi; [A | Q]

(* barrier-ordered-before *)
let bob = [R] ; po ; [dmbld]
  | [W] ; po ; [dmbst]
  | [dmbst]; po; [W]
  | [dmbld]; po; [R|W]
  | [L]; po; [A]
  | [A | Q]; po; [R | W]
  | [R | W]; po; [L]
  | [F | C]; po; [dsbsy]
  | [dsb] ; po

(* Ordered-before *)
let _ob = obs | dob | aob | bob | iio | tob | obtlbi | ctxob | obfault
let ob = _ob^+

(* Internal visibility requirement *)
acyclic po-loc | fr | co | rf as internal

(* External visibility requirement *)
irreflexive ob as external

(* Atomic: Basic LDXR/STXR constraint to forbid intervening writes. *)
empty rmw & (fre; coe) as atomic

(* Writes cannot forward to po-future translations *)
acyclic (po-pa | trfi) as translation-internal

(* No translations interposing well-bracketed take/return exceptions *)
(* empty take-to-return & ((ob & int) ; [T] ; (ob & higher-EL)) *)
\end{lstlisting}
  \end{minipage}}
  \caption{Strong Model}
\end{figure}

\subsubsection{Translation Faults}

To correctly implement ETS and TLBI-completion ordering for translation-faults
we produce \herd{fault} events which exist \herd{iio}-after the \herd{T} event
which causes them.

To get the correct ETS ordering,
we add \herd{FromR} and \herd{FromW} sets for faults that originate from load or store instructions.

Then we duplicate edges from \herd{ob} which end in a \herd{[R]} or \herd{[W]}
to also end in \herd{[Fault \& FromR]} or \herd{[Fault \& FromW]},
and those get included in the \herd{obfault} relation which is included in \herd{ob}.
To model ETS, we can then simply add \herd{[R|W] ; obfault ; [fault] ; iio\^{}-1 ; [T\_f]} to \herd{ob}.

See the relevant part of the Arm ARM (D.5.10.2 --- \textbf{Ordering and completion of TLB maintenance instructions})
\begin{lstlisting}[breaklines=true,language=none]
  A TLB maintenance operation without the nXS qualifier generated by a TLB maintenance instruction is
  finished for a PE when:
  - All memory accesses generated by that PE using in-scope old translation information are complete.
  - All memory accesses RWx generated by that PE are complete.

  RWx is the set of all memory accesses generated by instructions for that PE that appear in program order
  before an instruction I1 executed by that PE where all of the following apply:
  - I1 uses the in-scope old translation information.
  - The use of the in-scope old translation information generates a synchronous Data Abort.
  - If I1 did not generate an abort from use of the in-scope old translation information, I1 would generate
      a memory access that RWx would be locally-ordered-before.
\end{lstlisting}

\subsubsection{Edges justification}

We justify existence of edges in \herd{ob} with the following tests:

\paragraph{obs}
\begin{itemize}
  \item \herd{[W] ; trfe ; [T]}  (\TEST{CoTRpte.inv+dsb-isb})
  \item \herd{[W] ; trfe ; [T\_f]} (\TEST{CoTfRpte+dsb-isb})
\end{itemize}

\paragraph{tob}
\begin{itemize}
  \item \herd{[T\_f] ; tfr ; [W]} (\TEST{CoRpteTf.inv+dsb-isb})
  \item \herd{[T] ; iio ; [R|W] ; po ; [W]} (see also \herd{speculative ; [W]}, \TEST{S.T+dmb+po})
  \item \herd{speculative ; trfi} (\TEST{MP.RT.inv+dmb+ctrl-trfi})
\end{itemize}

\paragraph{obtlbi\_translate}
\begin{itemize}
  \item \herd{tcache1} (\TEST{MP.RT.EL1+dsb-tlbiis-dsb+dsb-isb}) \todojp{added .EL1 to test name}
  \item \herd{tcache2 \& (same-trans ; [T \& Stage1] ; trf\^-1 ; wco\^-1)} (\TEST{WDS+dsb-tlbiipa-dsb-po-eret})
  \item \herd{(tcache2 ; wco? ; [TLBI S2]) \& (same-trans ; [T \& Stage1] ; maybe\_TLB\_cached)} (\TEST{WDS+po-dsb-tlbiipa-dsb-tlbiis-dsb-eret})
\end{itemize}

\paragraph{obtlbi}
\begin{itemize}
  \item \herd{obtlbi\_translate} (see previous)
  \item \herd{[R] ; iio\^-1 ; (obtlbi\_translate \& ext)} (\TEST{RBS+dsb-tlbiis-dsb})
  \item \herd{[W] ; iio\^-1 ; (obtlbi\_translate \& ext)} (\TEST{WBM+dsb-tlbiis-dsb})
  \item \herd{[Fault] ; iio\^-1 ; (obtlbi\_translate \& ext)} (see \herd{obfault} for relevant tests)
\end{itemize}

\paragraph{ctxob}

These edges are over-approximate compared to the assumed real semantics of context-synchronizing events.

\begin{itemize}
  \item \herd{speculative ; [MSR]}   (\TEST{MP.RT.inv+dmb+addr-po-msr})
  \item \herd{[CSE] ; instruction-order}   (MP+dmb+ctrl-isb)
  \item \herd{[ContextChange] ; po ; [CSE]} (\TEST{SwitchTable.different-asid+eret})
  \item \herd{speculative ; [ISB]} (MP+dmb+ctrl-isb, MP+dmb+addr-po-isb, \TEST{MP.TR.inv+dmb+isb})
\end{itemize}

\paragraph{obfault}

\begin{itemize}
  \item \herd{[R] ; data ; [Fault\_W]}  (\TEST{S.RTf.inv.EL1+dsb-tlbiis-dsb+data})
  \item \herd{speculative ; [Fault\_W]}  (\TEST{S.RTf.inv.EL1+dsb-tlbiis-dsb+ctrl})
  \item \herd{[dmb] ; po ; [Fault\_R]}  (\TEST{MP.RTf.inv.EL1+dsb-tlbiis-dsb+dmb})
  \item \herd{[dmb] ; po ; [Fault\_W]}  (\TEST{S.RTf.inv.EL1+dsb-tlbiis-dsb+dmb})
  \item \herd{[A|Q] ; po ; [Fault]}  (\TEST{MP.RTf.inv.EL1+dsb-tlbiis-dsb+poap}, \TEST{MP.RTf.inv.EL1+dsb-tlbiis-dsb+poqp}, \TEST{S.RTf.inv.EL1+dsb-tlbiis-dsb+poap}, \TEST{S.RTf.inv.EL1+dsb-tlbiis-dsb+poqp}%
  \item \herd{[R|W] ; po ; [Fault\_L]}  (\TEST{S.RTf.inv.EL1+dsb-tlbiis-dsb+popl}, \TEST{R.Tf.inv.EL1+dsb-tlbiis-dsb+popl})
\end{itemize}

\paragraph{obETS}

\begin{itemize}
  \item \herd{obfault ; [Fault] ; iio\^-1 ; [T\_f]} (\TEST{MP.RTf.inv+dmbs}, \TEST{MP.RTf.inv+dmb+addr})%
\end{itemize}

\paragraph{dob}

These edges ensure that self-satisfying cycles cannot be constructed,
which could otherwise lead to new translation table entries out of thin air.

\begin{itemize}
  \item \herd{addr ; trfi} (\TEST{LB+addr-trfis})
  \item \herd{data ; trfi} (\TEST{LB+data-trfis})
\end{itemize}

Note the lack of \herd{ctrl ; trfi} here does not imply weakness,
as \herd{tob} already covers this.

\paragraph{axioms}
\begin{itemize}
  \item \herd{acyclic (po-pa | trfi)} (\TEST{CoTW1.inv})
\end{itemize}

\clearpage
\subsection{Weak Model}
{
\begin{lstlisting}[language=cat]
"VMSA weak"
include "cos.cat"
include "aarch64_mmu_common.cat"
include "barriers.cat"

let obs = rfe | fr | wco

let dob = addr | data
  | ctrl; [W]
  | (ctrl | (addr; po)); [ISB]
  | addr; po; [W]
  | (addr | data); rfi
  | (addr | ctrl | data); trfi

let aob = rmw
  | [range(rmw)]; rfi; [A | Q]

(* barrier-ordered-before *)
let bob = [R] ; po ; [dmbld]
  | [W] ; po ; [dmbst]
  | [dmbst]; po; [W]
  | [dmbld]; po; [R|W]
  | [L]; po; [A]
  | [A | Q]; po; [R | W]
  | [R | W]; po; [L]
  | [F | C]; po; [dsbsy]
  | [dsb] ; po
  | [CSE] ; instruction-order

(* Ordered-before *)
let ob = (obs | dob | aob | bob | ctxob)^+

(* Internal visibility requirement *)
acyclic po-loc | fr | co | rf as internal

(* External visibility requirement *)
irreflexive ob as external

(* Atomic: Basic LDXR/STXR constraint to forbid intervening writes. *)
empty rmw & (fre; coe) as atomic

(* Writes cannot forward to po-future translations *)
acyclic (po-pa | trfi) as translation-internal

(* break-before-make S1 *)
empty
  ([is_IW | W_invalid] ; co ; [W_valid] ; ob ; [CSE] ; instruction-order ; [T & Stage1])
  & (ob ; [dsbsy] ; po ; ([TLBI-S1] ; po ; [dsbsy] ; ob ; [CSE] ; instruction-order ; [T]) & tlb-affects)
  & trf & loc
  as bbm

(* break S1 *)
empty ([is_IW | W] ; co ; [W_invalid] ; ob ; [dsbsy] ; po
      ; ([TLBI-S1] ; po ; [dsbsy] ; ob ; [M]; iio^-1; [T]) & tlb-affects & ext
      ) & trf & loc
      as brk1

empty ([is_IW | W] ; co ; [W_invalid] ; ob ; [dsbsy] ; po
      ; ([TLBI-S1] ; po ; [dsbsy] ; ob ; [CSE] ; instruction-order ; [T]) & tlb-affects
      ) & trf & loc
      as brk2

(* break-before-make S2 *)
empty
  ([is_IW | W_invalid] ; co ; [W_valid] ; ob ; [CSE] ; instruction-order ; [T & Stage2])
  & (ob ; [dsbsy] ; po
    ; ([TLBI-S2] ; po ; [dsbsy] ; po ;
      ; ([TLBI-S1] ; po ; [dsbsy] ; ob ; [CSE]; instruction-order; [T]) & tlb-affects
      ; iio ; [T]) & tlb-affects)
  & trf & loc
    as bbms2

(* break S2 *)
empty ([is_IW | W] ; co ; [W_invalid] ; ob ; [dsbsy] ; po
      ; ([TLBI-S2] ; po ; [dsbsy] ; po ; ([TLBI-S1] ; po ; [dsbsy] ; ob ; [M]; iio^-1; [T]) & tlb-affects & ext
      ; iio ; [T]) & tlb-affects & ext
      ) & trf & loc
      as brk1s2

empty ([is_IW | W] ; co ; [W_invalid] ; ob ; [dsbsy] ; po
      ; ([TLBI-S2] ; po ; [dsbsy] ; po ; ([TLBI-S1] ; po ; [dsbsy] ; ob ; [CSE]; instruction-order; [T]) & tlb-affects
      ; iio ; [T]) & tlb-affects
      ) & trf & loc
      as brk2s2
\end{lstlisting}
}

\clearpage
\subsection{Break-before-make detection predicate}
{
\begin{lstlisting}[language=smt2]
; Check for break-before-make violations
; This is set of constraints is satisfiable iff there is a BBM violation
(declare-const BBM_Wl0 Event)
(declare-const BBM_Wl1 Event)
(declare-const BBM_Wl2 Event)
(declare-const BBM_Wl3 Event)

(declare-const BBM_Wl0_pa (_ BitVec 64))
(declare-const BBM_Wl1_pa (_ BitVec 64))
(declare-const BBM_Wl2_pa (_ BitVec 64))
(declare-const BBM_Wl3_pa (_ BitVec 64))

(assert (not (= BBM_Wl0_pa BBM_Wl1_pa)))
(assert (not (= BBM_Wl1_pa BBM_Wl2_pa)))
(assert (not (= BBM_Wl2_pa BBM_Wl3_pa)))

(declare-const BBM_Wl0_data (_ BitVec 64))
(declare-const BBM_Wl1_data (_ BitVec 64))
(declare-const BBM_Wl2_data (_ BitVec 64))
(declare-const BBM_Wl3_data (_ BitVec 64))

(declare-const BBM_ia (_ BitVec 36))

(define-fun ia_offset3 ((ia (_ BitVec 36))) (_ BitVec 12)
  (concat ((_ extract 8 0) ia) #b000))

(define-fun ia_offset2 ((ia (_ BitVec 36))) (_ BitVec 12)
  (concat ((_ extract 17 9) ia) #b000))

(define-fun ia_offset1 ((ia (_ BitVec 36))) (_ BitVec 12)
  (concat ((_ extract 26 18) ia) #b000))

(define-fun ia_offset0 ((ia (_ BitVec 36))) (_ BitVec 12)
  (concat ((_ extract 35 27) ia) #b000))

(define-fun page_offset ((pa (_ BitVec 64))) (_ BitVec 12)
  ((_ extract 11 0) pa))

(define-fun table_address ((desc (_ BitVec 64))) (_ BitVec 64)
  (concat #x0000 ((_ extract 47 12) desc) #x000))

(assert (= (page_offset BBM_Wl0_pa) (ia_offset0 BBM_ia)))
(assert (= (page_offset BBM_Wl1_pa) (ia_offset1 BBM_ia)))
(assert (= (page_offset BBM_Wl2_pa) (ia_offset2 BBM_ia)))
(assert (= (page_offset BBM_Wl3_pa) (ia_offset2 BBM_ia)))

(assert (tt_write BBM_Wl0 BBM_Wl0_pa BBM_Wl0_data))

(define-fun valid_desc ((desc (_ BitVec 64))) Bool
   (= (bvand desc #x0000000000000001) #x0000000000000001))

(define-fun valid_table_desc ((desc (_ BitVec 64))) Bool
   (= (bvand desc #x0000000000000011) #x0000000000000011))

; For each level, if it is valid, then its parent must be a valid table entry
(assert
  (and
    (implies (valid_desc BBM_Wl3_data) (valid_table_desc BBM_Wl2_data))
    (implies (valid_desc BBM_Wl2_data) (valid_table_desc BBM_Wl1_data))
    (implies (valid_desc BBM_Wl1_data) (valid_table_desc BBM_Wl0_data))))

; If an entry is pointed to by its parent, then it must be actually
; represented by a valid page table write at the correct location.
; The alternative is if the parent is invalid, in which case anything
; goes
(assert
  (implies (valid_table_desc BBM_Wl0_data)
    (and (tt_write BBM_Wl1 BBM_Wl1_pa BBM_Wl1_data)
         (= (table_address BBM_Wl0_data) (table_address BBM_Wl1_pa)))))

(assert
  (implies (valid_table_desc BBM_Wl1_data)
    (and (tt_write BBM_Wl2 BBM_Wl2_pa BBM_Wl2_data)
         (= (table_address BBM_Wl1_data) (table_address BBM_Wl2_pa)))))

(assert
  (implies (valid_table_desc BBM_Wl2_data)
    (and (tt_write BBM_Wl3 BBM_Wl3_pa BBM_Wl3_data)
         (= (table_address BBM_Wl2_data) (table_address BBM_Wl3_pa)))))

(declare-const BBM_W1 Event)
(declare-const BBM_W1_pa (_ BitVec 64))
(declare-const BBM_W1_data (_ BitVec 64))

(declare-const BBM_W2 Event)

; BBM_W1 and BBM_W2 conflict
(assert (and (tt_write BBM_W1 BBM_W1_pa BBM_W1_data) (valid_desc BBM_W1_data)))
(assert (W_valid BBM_W2))
(assert (not (= ((_ extract 47 12) BBM_W1_data) ((_ extract 47 12) (val_of_64 BBM_W2)))))
(assert (= BBM_W1_pa (addr_of BBM_W2)))

(assert (or
  (and (= BBM_W1 BBM_Wl3) (= BBM_W1_pa BBM_Wl3_pa) (= BBM_W1_data BBM_Wl3_data))
  (and (= BBM_W1 BBM_Wl2) (= BBM_W1_pa BBM_Wl2_pa) (= BBM_W1_data BBM_Wl2_data))
  (and (= BBM_W1 BBM_Wl1) (= BBM_W1_pa BBM_Wl1_pa) (= BBM_W1_data BBM_Wl1_data))
  (and (= BBM_W1 BBM_Wl0) (= BBM_W1_pa BBM_Wl0_pa) (= BBM_W1_data BBM_Wl0_data))))

(assert (co BBM_W1 BBM_W2))

(define-fun BBM_sequence1 ((S_Wp Event) (S_tlbi Event)) Bool
  (and
    (wco BBM_W1 S_Wp)
    (W_invalid S_Wp)
    (implies (= BBM_W1 BBM_Wl3) (or (= S_Wp BBM_Wl3) (= S_Wp BBM_Wl2) (= S_Wp BBM_Wl1) (= S_Wp BBM_Wl0)))
    (implies (= BBM_W1 BBM_Wl2) (or (= S_Wp BBM_Wl2) (= S_Wp BBM_Wl1) (= S_Wp BBM_Wl0)))
    (implies (= BBM_W1 BBM_Wl1) (or (= S_Wp BBM_Wl1) (= S_Wp BBM_Wl0)))
    (implies (= BBM_W1 BBM_Wl0) (= S_Wp BBM_Wl0))
    (wco S_Wp S_tlbi)
    (TLBI-VA S_tlbi)
    (= (tlbi_va (val_of_cache_op S_tlbi)) (concat #x0000 BBM_ia #x000))
    (wco S_tlbi BBM_W2)))

; If there are no valid BBM sequence between BBM_W1 and BBM_W2, we have a BBM violation 
(assert (forall ((BBM_Wp Event) (BBM_tlbi Event))
  (not (BBM_sequence1 BBM_Wp BBM_tlbi))))
\end{lstlisting}
}

\newpage
% !TEX root = top.tex
%\documentclass[11pt,a4paper]{article}

%\usepackage[left=1cm,right=2cm,top=0.5cm,bottom=0.5cm]{geometry}

%\usepackage{amsthm}
%\usepackage{mathtools}

%\newtheorem{lemma}{Lemma}
%\theoremstyle{definition}
%\spnewtheorem{definition}{Definition}{\bfseries}{\itshape}

%\begin{document}

%\title{Model relations and abstraction condition}
%\author{}
%\date{}
%\maketitle

\myappendix{Relationships between models}{app:modelsrelation}

In this appendix, we illustrate how our models support a relatively simple abstraction to higher-level code.
We prove three theorems:
that for static injectively-mapped address spaces, any execution which is consistent in the model with translation, erasing translation events gives an execution that is consistent in the original Armv8-A model without translation
(Theorem~\ref{thm:vaabstraction});
that for any consistent execution in the original Armv8-A model, there is a corresponding consistent execution in our extended model with translations
(Theorem~\ref{thm:vaantiabstraction});
and that our weak model is a sound over-approximation of our full translation model, i.e.,
that for any consistent execution in our full translation model, that same execution is consistent in the weak translation model
(Theorem~\ref{thm:strongrefinesweak}).

%\subsection{Relation between weak and strong models}

%\subsubsection{The weak model refines the strong model}
%
%\begin{lemma}
%  All \verb!ob! edges of the weak model are also \verb!ob! edges of the strong model.
%\end{lemma}
%\begin{proof}
%
%\end{proof}

\subsection{Soundness of the weak model}

\begin{theorem}
\label{thm:strongrefinesweak}
The weak model is a sound over-approximation of the strong model.
\end{theorem}
\begin{proof}
The definition of \verb!ob! in the strong model contains all the clauses of the weak model (and more).
\\
The extra axioms of the weak model are subsumed by those of the strong model:

We label events with identifiers, write given edges solid, and derived edges dashed.
\begin{itemize}
\item bbm
\begin{verbatim}
([a:is_IW | W_invalid] ; co ; [b:W_valid]
 ; ob ; [c:CSE] ; instruction-order ; [d:Tf & Stage1])
&
(ob ; [e:dsb] ; po
 ; (([f:TLBI-S1] ; po ; [g:dsb] ; ob ; [h:CSE] ; instruction-order)
    & tlb_affects))
& trf & loc
\end{verbatim}

\newcommand\xxevent[3]{
\node (#1) at (#3) {$\sf{}{#1}{:}{#2}$};
}

\begin{center}
\begin{tikzpicture}
\xxevent{a}{init | W_i}{0,0}
\node (b) at (3,0) {$\sf{}b{:}W_v$};
\node (c) at (6,0) {$\sf{}c{:}CSE$};
\node (d) at (9,0) {$\sf{}d{:}Tf_{S1}$};
\node (e) at (1,-2) {$\sf{}e{:}dsbsy$};
\node (f) at (3,-2) {$\sf{}f{:}TLBI_{S1}$};
\node (g) at (5,-2) {$\sf{}g{:}dsbsy$};
\node (h) at (7,-2) {$\sf{}h{:}CSE$};
\draw[->,blue] (a) -- node [above] {\sf{}co} (b);
\draw[->,green] (b) -- node [above] {\sf{}ob} (c);
\draw[->] (c) -- node [above] {\sf{}io} (d);
\draw[->,green] (a) -- node [right] {\sf{}ob} (e);
\draw[->] (e) -- node [above] {\sf{}po} (f);
\draw[->] (f) -- node [above] {\sf{}po} (g);
\draw[->,green] (g) -- node [above] {\sf{}ob} (h);
\draw[->] (h) -- node [right] {\sf{}io} (d);
\draw[->] (f) -- node [right] {\sf{}tlb\_affects} (d);
\draw[->,red] (a) to [bend left] node [above] {\sf{}trf} (d);
\draw [-{Straight Barb[left]},blue] (f.100) -- (b.260);
\draw [-{Straight Barb[left]},blue] (b.280) -- node [left] {\sf{}wco} (f.80);
\draw[->,green,dashed] (f) to [bend left] node [above] {\sf{}bob} (g);
\draw[->,green,dashed] (c) to [bend left] node [above] {\sf{}bob} (d);
\draw[->,green,dashed] (h) -- node [above] {\sf{}bob} (d);
\draw[->,green,dashed] (d) to [bend right] node [above] {\sf{}tfr} (b);
\draw[->,green,dashed] (f) to [bend right=50] node [above] {\sf{}ob} (d);
\end{tikzpicture}
\end{center}

From \verb![f] ; po ; [g]!,
we have \verb![f] ; bob ; [g]!,
and from \verb![h] ; io ; [d]!,
we have \verb![h] ; bob ; [d]!
Therefore, together with \verb![g] ; ob ; [h]!,
we have
\verb![f] ; ob ; [d]!.
\\
From \verb![c] ; po ; [d]!,
we \verb![c] ; bob ; [d]!.

\verb!wco! relates \verb!b! and \verb!f!.
\\
\begin{itemize}
\item If \verb![f] ; wco ; [b]!, then this is not a BBM violation.

\begin{center}
\begin{tikzpicture}
\xxevent{a}{W_i}{0,0}
\node (b) at (3,0) {$\sf{}b{:}W_v$};
\node (c) at (6,0) {$\sf{}c{:}CSE$};
\node (d) at (9,0) {$\sf{}d{:}Tf_{S1}$};
\node (e) at (1,-2) {$\sf{}e{:}dsbsy$};
\node (f) at (3,-2) {$\sf{}f{:}TLBI_{S1}$};
\node (g) at (5,-2) {$\sf{}g{:}dsbsy$};
\node (h) at (7,-2) {$\sf{}h{:}CSE$};
\draw[->,blue] (a) -- node [above] {\sf{}co} (b);
\draw[->,green] (b) -- node [above] {\sf{}ob} (c);
\draw[->] (c) -- node [above] {\sf{}io} (d);
\draw[->,green] (a) -- node [right] {\sf{}ob} (e);
\draw[->] (e) -- node [above] {\sf{}po} (f);
\draw[->] (f) -- node [above] {\sf{}po} (g);
\draw[->,green] (g) -- node [above] {\sf{}ob} (h);
\draw[->] (h) -- node [right] {\sf{}io} (d);
\draw[->] (f) -- node [right] {\sf{}tlb\_affects} (d);
\draw[->,red] (a) to [bend left] node [above] {\sf{}trf} (d);
\draw [->,blue] (f) -- node [left] {\sf{}wco} (b);
\draw[->,green,dashed] (f) to [bend left] node [above] {\sf{}bob} (g);
\draw[->,green,dashed] (c) to [bend left] node [above] {\sf{}bob} (d);
\draw[->,green,dashed] (h) -- node [above] {\sf{}bob} (d);
\draw[->,green,dashed] (d) to [bend right] node [above] {\sf{}tfr} (b);
\draw[->,green,dashed] (f) to [bend right=50] node [above] {\sf{}ob} (d);
\end{tikzpicture}
\end{center}

From \verb![d] ; tfr ; [b]!, we have \verb![d] ; ob ; [b]!.

Therefore, we have a cycle in \verb!ob!.

\item If \verb![b] ; wco ; [f]!, then:

\begin{center}
\begin{tikzpicture}
\xxevent{a}{W_i}{0,0}
\node (b) at (3,0) {$\sf{}b{:}W_v$};
\node (c) at (6,0) {$\sf{}c{:}CSE$};
\node (d) at (9,0) {$\sf{}d{:}Tf_{S1}$};
\node (e) at (1,-2) {$\sf{}e{:}dsbsy$};
\node (f) at (3,-2) {$\sf{}f{:}TLBI_{S1}$};
\node (g) at (5,-2) {$\sf{}g{:}dsbsy$};
\node (h) at (7,-2) {$\sf{}h{:}CSE$};
\draw[->,blue] (a) -- node [above] {\sf{}co} (b);
\draw[->,green] (b) -- node [above] {\sf{}ob} (c);
\draw[->] (c) -- node [above] {\sf{}io} (d);
\draw[->,green] (a) -- node [right] {\sf{}ob} (e);
\draw[->] (e) -- node [above] {\sf{}po} (f);
\draw[->] (f) -- node [above] {\sf{}po} (g);
\draw[->,green] (g) -- node [above] {\sf{}ob} (h);
\draw[->] (h) -- node [right] {\sf{}io} (d);
\draw[->] (f) -- node [right] {\sf{}tlb\_affects} (d);
\draw[->,red] (a) to [bend left] node [above] {\sf{}trf} (d);
\draw [->,blue] (b) -- node [left] {\sf{}wco} (f);
\draw[->,green,dashed] (f) to [bend left] node [above] {\sf{}bob} (g);
\draw[->,green,dashed] (c) to [bend left] node [above] {\sf{}bob} (d);
\draw[->,green,dashed] (h) -- node [above] {\sf{}bob} (d);
\draw[->,green,dashed] (d) to [bend right] node [above] {\sf{}tfr} (b);
\draw[->,green,dashed] (d) -- node [left] {\sf{}tlb\_barriered $\subseteq$ ob} (f);
\draw[->,green,dashed] (f) to [bend right=50] node [above] {\sf{}ob} (d);
\end{tikzpicture}
\end{center}

\begin{itemize}
\item If there is another TLBI preventing a BBM violation involving \verb!a! and \verb!b!,
then there is another subgraph of the execution that corresponds to the previous case.
\item If not, then there is also a BBM violation in the strong model,
because they have the same execution candidate,
and use the same BBM check.
\end{itemize}

\end{itemize}

\item brk1
\begin{verbatim}
([a:is_IW | W] ; co ; [b:W_invalid] ; ob ; [c:dsbsy] ; po
 ; ([d:TLBI-S1] ; po ; [e:dsbsy] ; ob ; [f:M]; iio^-1; [g:T])
   & tlb_affects & ext)
& trf & loc
\end{verbatim}

\begin{center}
\begin{tikzpicture}
\xxevent{a}{W}{0,0}
\xxevent{b}{W_i}{2,0}
\xxevent{c}{dsbsy}{4,0}
\xxevent{d}{TLBI_{S1}}{6,0}
\xxevent{e}{dsbsy}{8,0}
\xxevent{f}{M}{10,0}
\xxevent{g}{T}{9,-2}
\draw[->,blue] (a) -- node [above] {\sf{}co} (b);
\draw[->,green] (b) -- node [above] {\sf{}ob} (c);
\draw[->] (c) -- node [above] {\sf{}po} (d);
\draw[->] (d) -- node [above] {\sf{}po} (e);
\draw[->,green] (e) -- node [above] {\sf{}ob} (f);
\draw[->] (g) -- node [left] {\sf{}iio} (f);
\draw[->] (d) -- node [above] {\sf{}tlb\_affects,ext} (g);
\draw[->,red] (a) -- node [below] {\sf{}trf,loc} (g);
\draw[->,green,dashed] (g) -- node [above] {\sf{}tfr} (b);
\draw[->,green,dashed] (g) -- node [left] {\sf{}tlb\_barriered} (d);
\draw[->,green,dashed] (f) to [bend right] node [above] {\sf{}ob} (d);
\draw[->,green,dashed] (d) to [bend right=10] node [above] {\sf{}ob} (f);
\draw [-{Straight Barb[left]},blue] (d.100) to [bend right=18] node [above] {\sf{}wco} (b.80);
\draw [-{Straight Barb[left]},blue] (b.100) to [bend left=20] (d.80);
\draw[->,green,dashed] (b) to [bend left=10] node [above] {\sf{}ob} (d);
\draw[->,green,dashed] (d) to [bend left=10] node [above] {\sf{}ob} (e);
\end{tikzpicture}
\end{center}
From \verb![a] ; trf ; [g]! and \verb![a] ; co ; [b]!, we have \verb![g] ; tfr ; [b]!.
\\
\verb!wco! relates \verb!b! and \verb!d!.
\\
\begin{itemize}
\item
If \verb![d] ; wco ; [b]!.
\\
then we have \verb![d] ; ob ; [b]!
\\
Moreover, from \verb![c] ; po ; [d]!, in the weak model, we have \verb![c] ; ob ; [d]!.
\\
From \verb![b] ; ob ; [c]! and \verb![c] ; ob ; [d]!, we have \verb![b] ; ob ; [d]!.
\\
Therefore, we have a cycle in \verb!ob!.

\item
If \verb![b] ; wco ; [d]!.
\\
From \verb![g] ; tfr ; [b]!, \verb![b] ; wco ; [d]!, and \verb![d] ; tlb-affects^-1 ; [g]!, we have \verb![g] ; tlb_barriered ; [d]!,
\\
and therefore \verb![g] ; obtlbi_translate ; [d]!.
\\
From \verb![g] ; obtlbi_translate ; [d]!,
\verb![g] ; ext ; [d]!,
and
\verb![f] ; iio^-1 ; [g]!,
we have \verb![f] ; obtlbi_translate ; [d]!,
\\
and therefore \verb![f] ; ob ; [d]!.
\\
Moreover,
from \verb![d] ; po ; [e]!,
we have \verb![d] ; bob ; [e]!, and therefore \verb![d] ; ob ; [e]!.
\\
From \verb![d] ; ob ; [e]!
and \verb![e] ; ob ; [f]!,
we have \verb![d] ; ob ; [f]!.
\\
Therefore, we have a cycle in \verb!ob!.
\end{itemize}

\item brk2
\begin{verbatim}
([a:is_IW | W] ; co ; [b:W_invalid] ; ob ; [c:dsbsy] ; po
 ; ([d:TLBI-S1] ; po ; [e:dsbsy] ; ob ; [f:CSE] ; instruction-order ; [g:T])
    & tlb-affects)
& trf & loc
\end{verbatim}

\begin{center}
\begin{tikzpicture}
\xxevent{a}{W}{0,0}
\xxevent{b}{W_i}{2,0}
\xxevent{c}{dsbsy}{4,0}
\xxevent{d}{TLBI_{S1}}{6,0}
\xxevent{e}{dsbsy}{8,0}
\xxevent{f}{CSE}{10,0}
\xxevent{g}{T}{12,0}
\draw[->,blue] (a) -- node [above] {\sf{}co} (b);
\draw[->,green] (b) -- node [above] {\sf{}ob} (c);
\draw[->] (c) -- node [above] {\sf{}po} (d);
\draw[->] (d) -- node [above] {\sf{}po} (e);
\draw[->,green] (e) -- node [above] {\sf{}ob} (f);
\draw[->] (f) -- node [above] {\sf{}io} (g);
\draw[->] (d) to [bend right] node [above] {\sf{}tlb\_affects} (g);
\draw[->,red] (a) to [bend right] node [below] {\sf{}trf,loc} (g);
\draw[->,green,dashed] (c) to [bend left] node [above] {\sf{}bob} (d);
\draw[->,green,dashed] (d) to [bend left] node [above] {\sf{}bob} (e);
\draw[->,green,dashed] (f) to [bend left] node [above] {\sf{}bob} (g);
\draw[->,green,dashed] (g) to [bend left] node [left] {\sf{}tlb\_barriered $\subseteq$ ob} (d);
\end{tikzpicture}
\end{center}

From \verb![c] ; po ; [d]!, we have \verb![c] ; bob ; [d]!,
and therefore \verb![c] ; ob ; [d]!.
\\
Therefore, as before, by examination of \verb!wco!, we have \verb![g] ; ob ; [d]!.
\\
From \verb![d] ; po ; [e]!, we have \verb![d] ; bob ; [e]!,
and therefore \verb![d] ; ob ; [e]!.
\\
Moreover, from \verb![f] ; instruction-order ; [g]!,
we have \verb![f] ; bob ; [g]!.
\\
Therefore, we have a cycle in \verb!ob!.

\item
bbms2
\begin{verbatim}
([a:is_IW | W_invalid] ; co ; [b:W_valid] ; ob
 ; [c:CSE] ; instruction-order ; [d:Tf & Stage2])
& (ob ; [e:dsbsy] ; po
   ; ([f:TLBI-S2] ; po ; [g:dsbsy] ; po ;
      ; ([h:TLBI-S1] ; po ; [i:dsbsy] ; ob ; [j:CSE]; instruction-order; [k:T])
         & tlb-affects
      ; iio) & tlb-affects)
& trf & loc
\end{verbatim}

\begin{center}
\begin{tikzpicture}
\xxevent{a}{init|W_i}{0,0}
\xxevent{b}{W_v}{5,0}
\xxevent{c}{CSE}{10,0}
\xxevent{d}{Tf_{S2}}{15,0}
\xxevent{e}{dsbsy}{1,-2}
\xxevent{f}{TLBI_{S2}}{3,-2}
\xxevent{g}{dsbsy}{5,-2}
\xxevent{h}{TLBI_{S1}}{7,-2}
\xxevent{i}{dsbsy}{9,-2}
\xxevent{j}{CSE}{11,-2}
\xxevent{k}{T}{13,-2}
\draw[->,blue] (a) -- node [above] {\sf{}co} (b);
\draw[->,green] (b) -- node [above] {\sf{}ob} (c);
\draw[->] (c) -- node [above] {\sf{}io} (d);
\draw[->,green] (a) -- node [left] {\sf{}ob} (e);
\draw[->] (e) -- node [above] {\sf{}po} (f);
\draw[->] (f) -- node [above] {\sf{}po} (g);
\draw[->] (g) -- node [above] {\sf{}po} (h);
\draw[->] (h) -- node [above] {\sf{}po} (i);
\draw[->,green] (i) -- node [above] {\sf{}ob} (j);
\draw[->] (j) -- node [above] {\sf{}io} (k);
\draw[->] (k) -- node [above] {\sf{}iio} (d);
\draw[->] (f) -- node [above] {\sf{}tlb\_affects} (d);
\draw[->] (h) to [bend right] node [above] {\sf{}tlb\_affects} (k);
\draw[->,red] (a) to [bend left] node [below] {\sf{}trf,loc} (d);
\draw[->,green,dashed] (c) to [bend left] node [above] {\sf{}bob} (d);
\draw[->,green,dashed] (j) to [bend left] node [above] {\sf{}bob} (k);
\draw[->,green,dashed] (j) to [bend left] node [above] {\sf{}bob} (d);
\draw [-{Straight Barb[left]},blue] (h.100) -- node [right] {\sf{}wco} (b.260);
\draw [-{Straight Barb[left]},blue] (b.280) -- (h.80);
\draw[->,green,dashed] (e) to [bend left] node [above] {\sf{}bob} (f);
\draw[->,green,dashed] (f) to [bend left] node [above] {\sf{}bob} (h);
\draw[->,green,dashed] (h) to [bend left] node [above] {\sf{}bob} (i);
\draw[->,green,dashed] (j) -- node [above] {\sf{}bob} (d);
\draw[->,green,dashed] (d) to [bend right] node [above] {\sf{}tfr} (b);
\draw[->,cyan,dashed] (d) -- node [left] {\sf{}tlb\_barriered $\subseteq$ ob} (f);
\end{tikzpicture}
\end{center}
From \verb![c] ; instruction-order ; [d]!,
we have \verb![c] ; bob ; [d]!,
and therefore \verb![c] ; ob ; [d]!.
\\
From \verb![j] ; instruction-order ; [k]!,
we have \verb![j] ; bob ; [k]!,
and therefore \verb![j] ; ob ; [k]!.
\\
From \verb![e] ; po ; [f]!,
we have \verb![e] ; bob ; [f]!.
\\
From \verb![f] ; po [g:dsbsy] ; po ; [h]!,
we have \verb![f] ; bob ; [h]!.
\\
From \verb![h] ; po ; [i]!,
we have \verb![h] ; bob ; [i]!.
\\
From \verb![a] ; trf ; [d]! and \verb![a] ; co ; [b]!, we have \verb![d] ; tfr ; [b]!.
\\
From \verb![j] ; instruction-order ; [k] ; iio ; [d]!,
we have \verb![j] ; bob ; [d]!.
\\
\verb!wco! relates \verb!b! and \verb!h!.
\begin{itemize}
\item
Assume \verb![h] ; wco ; [b]!. Then this is not a BBM violation.
\\
From \verb![d] ; tfr ; [b]!, we have \verb![d] ; ob ; [b]!.

\item
Assume \verb![b] ; wco ; [h]!.
\begin{itemize}
\item If there is another TLBI preventing a BBM violation involving \verb!a! and \verb!b!,
then there is another subgraph of the execution that corresponds to the previous case.
\item If not, then there is also a BBM violation in the strong model,
because they have the same execution candidate,
and use the same BBM check.
\end{itemize}
\end{itemize}

\item brk1s2
\begin{verbatim}
([a:is_IW | W] ; co ; [b:W_invalid] ; ob ; [c:dsbsy] ; po
 ; ([d:TLBI-S2] ; po ; [e:dsbsy] ; po ;
    ([f:TLBI-S1] ; po ; [g:dsbsy] ; ob ; [h:M]; iio^-1; [i:T & Stage1])
     & tlb-affects & ext
; iio ; [j:T & Stage2]) & tlb-affects & ext
) & trf & loc
\end{verbatim}

\begin{center}
\scalebox{0.9}{
\begin{tikzpicture}
\xxevent{a}{W_v}{0,0}
\xxevent{b}{W_i}{2,0}
\xxevent{c}{dsbsy}{4,0}
\xxevent{d}{TLBI_{S2}}{6,0}
\xxevent{e}{dsbsy}{8,0}
\xxevent{f}{TLBI_{S1}}{10,0}
\xxevent{g}{dsbsy}{12,0}
\xxevent{h}{M}{14,0}
\xxevent{i}{T_{S1}}{13,-2}
\xxevent{j}{T_{S2}}{15,-4}
\xxevent{k}{W}{10,2}
\draw[->,blue] (a) -- node [above] {\sf{}co} (b);
\draw[->,green] (b) -- node [above] {\sf{}ob} (c);
\draw[->] (c) -- node [above] {\sf{}po} (d);
\draw[->] (d) -- node [above] {\sf{}po} (e);
\draw[->] (e) -- node [above] {\sf{}po} (f);
\draw[->] (f) -- node [above] {\sf{}po} (g);
\draw[->,green] (g) -- node [above] {\sf{}ob} (h);
\draw[->] (i) -- node [left] {\sf{}iio} (h);
\draw[->] (i) -- node [left] {\sf{}iio} (j);
\draw[->] (d) -- node [above] {\sf{}tlb\_affects,ext} (j);
\draw[->] (f) -- node [above] {\sf{}tlb\_affects,ext} (i);
\draw[->,red] (a) to [bend right] node [below] {\sf{}trf,loc} (j);
\draw[->,red,dashed] (k) to [bend left] node [right] {\sf{}trf,loc} (i);
\draw[->,green,dashed] (c) to [bend left] node [above] {\sf{}bob} (d);
\draw[->,green,dashed] (d) to [bend left] node [above] {\sf{}bob} (e);
\draw[->,green,dashed] (e) to [bend left] node [above] {\sf{}bob} (f);
\draw[->,green,dashed] (f) to [bend left] node [above] {\sf{}bob} (g);
\draw[->,cyan,dashed] (j) to [bend left] node [above] {\sf{}tlb\_barriered} (d);
\draw [-{Straight Barb[left]},blue] (f.100) -- node [right] {\sf{}wco} (k.260);
\draw [-{Straight Barb[left]},blue] (k.280) -- (f.80);
\end{tikzpicture}
}
\end{center}

Consider \verb!wco!: \\
From \verb![c] ; po ; [d]!, we have \verb![c] ; bob ; [d]!,
and therefore \verb![c] ; ob ; [d]!.
\\
From \verb![d] ; po ; [e]! and \verb![e] ; po ; [f]! we have \verb![d] ; ob ; [f]!.

Therefore \verb!wco! must give us \verb![a] ; wco ; [b] ; wco ; [d] ; wco ; [f]!
otherwise there is a cycle in \verb!ob! and the execution is trivially forbidden.

For \verb![d] ; ob ; [e]!,
\verb![e] ; ob ; [f]!,
and \verb![f] ; ob ; [g]!.

\verb!wco! relates \verb!b! and \verb!d!.
\\
\begin{itemize}
\item
If \verb![d] ; wco ; [b]!, we have a cycle in \verb!ob!.

\item
If \verb![b] ; wco ; [d]!.
\\
From \verb![j] ; tfr ; [b]!, \verb![b] ; wco ; [d]!, and \verb![j] ; tlb-affects^-1 ; [d]!,
we have \verb![j] ; tlb_barriered ; [d]!.
\\
Moreover, there must exist \verb!k! such that \verb![k] ; trf ; [i]!,
and \verb!k! must be related by \verb!wco! to \verb!f!.
\begin{itemize}
\item
If \verb![k] ; wco ; [f]!:

\begin{center}
\scalebox{0.9}{
\begin{tikzpicture}
\xxevent{a}{W_v}{0,0}
\xxevent{b}{W_i}{2,0}
\xxevent{c}{dsbsy}{4,0}
\xxevent{d}{TLBI_{S2}}{6,0}
\xxevent{e}{dsbsy}{8,0}
\xxevent{f}{TLBI_{S1}}{10,0}
\xxevent{g}{dsbsy}{12,0}
\xxevent{h}{M}{14,0}
\xxevent{i}{T_{S1}}{13,-2}
\xxevent{j}{T_{S2}}{15,-4}
\xxevent{k}{W}{10,2}
\draw[->,blue] (a) -- node [above] {\sf{}co} (b);
\draw[->,green] (b) -- node [above] {\sf{}ob} (c);
\draw[->] (c) -- node [above] {\sf{}po} (d);
\draw[->] (d) -- node [above] {\sf{}po} (e);
\draw[->] (e) -- node [above] {\sf{}po} (f);
\draw[->] (f) -- node [above] {\sf{}po} (g);
\draw[->,green] (g) -- node [above] {\sf{}ob} (h);
\draw[->] (i) -- node [left] {\sf{}iio} (h);
\draw[->] (i) -- node [left] {\sf{}iio} (j);
\draw[->] (d) -- node [above] {\sf{}tlb\_affects,ext} (j);
\draw[->] (f) -- node [above] {\sf{}tlb\_affects,ext} (i);
\draw[->,red] (a) to [bend right] node [below] {\sf{}trf,loc} (j);
\draw[->,red,dashed] (k) to [bend left] node [right] {\sf{}trf,loc} (i);
\draw[->,green,dashed] (c) to [bend left] node [above] {\sf{}bob} (d);
\draw[->,green,dashed] (d) to [bend left] node [above] {\sf{}bob} (e);
\draw[->,green,dashed] (e) to [bend left] node [above] {\sf{}bob} (f);
\draw[->,green,dashed] (f) to [bend left] node [above] {\sf{}bob} (g);
\draw[->,cyan,dashed] (j) to [bend left] node [above] {\sf{}tlb\_barriered} (d);
\draw[->,cyan,dashed] (i) to [bend left] node [above] {\sf{}maybe\_TLB\_cached} (f);
\draw [->,blue] (k.280) -- node [right] {\sf{}wco} (f.80);
%\draw[->,cyan,dashed] (i) -- node [left] {\sf{}obtlbi\_translate $\subseteq$ ob} (f);
\draw[->,green,dashed] (h) to [bend right] node [above] {\sf{}obtlbi} (f);
\end{tikzpicture}
}
\end{center}

Then we have \verb![j] ; tlb_barriered ; [d]!
and \verb![i] ; maybe_TLB_cached ; [f]!,
then from third clause of \verb!obltbi_translate!
we get \verb![j] ; obtlbi_translate ; [f]!.
From the second clause of \verb!obtlbi! we have \verb![h] ; obtlbi ; [f]!,
and so \verb![h] ; ob ; [f]!,
and therefore we have a cycle in \verb!ob!.
\item
If \verb![f] ; wco ; [k]!:

\begin{center}
\scalebox{0.9}{
\begin{tikzpicture}
\xxevent{a}{W_v}{0,0}
\xxevent{b}{W_i}{2,0}
\xxevent{c}{dsbsy}{4,0}
\xxevent{d}{TLBI_{S2}}{6,0}
\xxevent{e}{dsbsy}{8,0}
\xxevent{f}{TLBI_{S1}}{10,0}
\xxevent{g}{dsbsy}{12,0}
\xxevent{h}{M}{14,0}
\xxevent{i}{T_{S1}}{13,-2}
\xxevent{j}{T_{S2}}{15,-4}
\xxevent{k}{W}{10,2}
\draw[->,blue] (a) -- node [above] {\sf{}co} (b);
\draw[->,green] (b) -- node [above] {\sf{}ob} (c);
\draw[->] (c) -- node [above] {\sf{}po} (d);
\draw[->] (d) -- node [above] {\sf{}po} (e);
\draw[->] (e) -- node [above] {\sf{}po} (f);
\draw[->] (f) -- node [above] {\sf{}po} (g);
\draw[->,green] (g) -- node [above] {\sf{}ob} (h);
\draw[->] (i) -- node [left] {\sf{}iio} (h);
\draw[->] (i) -- node [left] {\sf{}iio} (j);
\draw[->] (d) -- node [above] {\sf{}tlb\_affects,ext} (j);
\draw[->] (f) -- node [above] {\sf{}tlb\_affects,ext} (i);
\draw[->,red] (a) to [bend right] node [below] {\sf{}trf,loc} (j);
\draw[->,red,dashed] (k) to [bend left] node [right] {\sf{}trf,loc} (i);
\draw[->,green,dashed] (c) to [bend left] node [above] {\sf{}bob} (d);
\draw[->,green,dashed] (d) to [bend left] node [above] {\sf{}bob} (e);
\draw[->,green,dashed] (e) to [bend left] node [above] {\sf{}bob} (f);
\draw[->,green,dashed] (f) to [bend left] node [above] {\sf{}bob} (g);
\draw[->,cyan,dashed] (j) to [bend left] node [above] {\sf{}tlb\_barriered} (d);
\draw [->,blue] (f.100) -- node [right] {\sf{}wco} (k.260);
\end{tikzpicture}
}
\end{center}

Then \verb![d] ; wco ; [k]!,
and \verb![j] ; tlb_barriered ; [d]!.
Then from the second clause of \verb!obtlbi_translate!
we have \verb![j] ; obtlbi_translate ; [d]!
From the second clause of \verb!obtlbi! we have \verb![h] ; obtlbi ; [d]!,
Which implies \verb![h] ; ob ; [d]!,
but \verb![d] ; ob ; [h]! by \verb!bob! so we have a cycle in \verb!ob!.
\end{itemize}

\end{itemize}

\item brk2s2
\begin{verbatim}
([a:is_IW | W] ; co ; [b:W_invalid] ; ob ; [c:dsbsy] ; po
 ; ([d:TLBI-S2] ; po ; [e:dsbsy] ; po ;
    ([f:TLBI-S1] ; po ; [g:dsbsy] ; ob ; [h:CSE] ; instruction-order; [i:T & Stage1])
     & tlb-affects
    ; iio ; [j:T & Stage2]) & tlb-affects)
& trf & loc
\end{verbatim}

\begin{center}
\scalebox{0.7}{
\begin{tikzpicture}
\xxevent{a}{W_v}{0,0}
\xxevent{b}{W_i}{2,0}
\xxevent{c}{dsbsy}{4,0}
\xxevent{d}{TLBI_{S2}}{6,0}
\xxevent{e}{dsbsy}{8,0}
\xxevent{f}{TLBI_{S1}}{10,0}
\xxevent{g}{dsbsy}{12,0}
\xxevent{h}{CSE}{14,0}
\xxevent{i}{T_{S1}}{16,0}
\xxevent{j}{T_{S2}}{18,0}
\xxevent{k}{W}{10,2}
\draw[->,blue] (a) -- node [above] {\sf{}co} (b);
\draw[->,green] (b) -- node [above] {\sf{}ob} (c);
\draw[->] (c) -- node [above] {\sf{}po} (d);
\draw[->] (d) -- node [above] {\sf{}po} (e);
\draw[->] (e) -- node [above] {\sf{}po} (f);
\draw[->] (f) -- node [above] {\sf{}po} (g);
\draw[->,green] (g) -- node [above] {\sf{}ob} (h);
\draw[->] (h) -- node [above] {\sf{}io} (i);
\draw[->] (i) -- node [above] {\sf{}iio} (j);
\draw[->] (d) to [bend right] node [above] {\sf{}tlb\_affects,ext} (j);
\draw[->] (f) to [bend right] node [above] {\sf{}tlb\_affects,ext} (i);
\draw[->,red] (a) to [bend right] node [below] {\sf{}trf,loc} (j);
\draw[->,red,dashed] (k) to [bend left] node [above] {\sf{}trf,loc} (i);
\draw[->,green,dashed] (c) to [bend left] node [above] {\sf{}bob} (d);
\draw[->,green,dashed] (d) to [bend left] node [above] {\sf{}bob} (e);
\draw[->,green,dashed] (e) to [bend left] node [above] {\sf{}bob} (f);
\draw[->,green,dashed] (f) to [bend left] node [above] {\sf{}bob} (g);
\draw[->,green,dashed] (h) to [bend left] node [above] {\sf{}bob} (i);
\draw[->,green,dashed] (h) to [bend left] node [above] {\sf{}bob} (j);
\draw[->,cyan,dashed] (j) to [bend left] node [below] {\sf{}tlb\_barriered} (d);
\draw[->,green,dashed] (j) to [bend right] node [above] {\sf{}obtlbi\_translate $\subseteq$ ob} (f);
\draw [-{Straight Barb[left]},blue] (f.100) -- node [right] {\sf{}wco} (k.260);
\draw [-{Straight Barb[left]},blue] (k.280) -- (f.80);
\end{tikzpicture}
}
\end{center}

Similar to the previous case,
we have \verb![c] ; ob ; [d]!,
\verb![d] ; ob ; [e]!,
\verb![e] ; ob ; [f]!,
\verb![f] ; ob ; [g]!,
\verb![h] ; ob ; [i]!,
and \verb![h] ; ob ; [j]!.

\verb!wco! relates \verb!b! and \verb!d!.
\\
\begin{itemize}
\item
If \verb![d] ; wco ; [b]!, we have a cycle in \verb!ob!.

\item
If \verb![b] ; wco ; [d]!.
\\
From \verb![j] ; tfr ; [b]!, \verb![b] ; wco ; [d]!, and \verb![j] ; tlb-affects^-1 ; [d]!,
we have \verb![j] ; tlb_barriered ; [d]!.
\\
Moreover, there must exist \verb!k! such that \verb![k] ; trf ; [i]!,
and \verb!k! must be related by \verb!wco! to \verb!f!.
\begin{itemize}
\item
If \verb![k] ; wco ; [f]!:

\begin{center}
\scalebox{0.7}{
\begin{tikzpicture}
\xxevent{a}{W_v}{0,0}
\xxevent{b}{W_i}{2,0}
\xxevent{c}{dsbsy}{4,0}
\xxevent{d}{TLBI_{S2}}{6,0}
\xxevent{e}{dsbsy}{8,0}
\xxevent{f}{TLBI_{S1}}{10,0}
\xxevent{g}{dsbsy}{12,0}
\xxevent{h}{CSE}{14,0}
\xxevent{i}{T_{S1}}{16,0}
\xxevent{j}{T_{S2}}{18,0}
\xxevent{k}{W}{10,2}
\draw[->,blue] (a) -- node [above] {\sf{}co} (b);
\draw[->,green] (b) -- node [above] {\sf{}ob} (c);
\draw[->] (c) -- node [above] {\sf{}po} (d);
\draw[->] (d) -- node [above] {\sf{}po} (e);
\draw[->] (e) -- node [above] {\sf{}po} (f);
\draw[->] (f) -- node [above] {\sf{}po} (g);
\draw[->,green] (g) -- node [above] {\sf{}ob} (h);
\draw[->] (h) -- node [above] {\sf{}io} (i);
\draw[->] (i) -- node [above] {\sf{}iio} (j);
\draw[->] (d) to [bend right] node [above] {\sf{}tlb\_affects,ext} (j);
\draw[->] (f) to [bend right] node [above] {\sf{}tlb\_affects,ext} (i);
\draw[->,red] (a) to [bend right] node [below] {\sf{}trf,loc} (j);
\draw[->,red,dashed] (k) to [bend left] node [above] {\sf{}trf,loc} (i);
\draw[->,green,dashed] (c) to [bend left] node [above] {\sf{}bob} (d);
\draw[->,green,dashed] (d) to [bend left] node [above] {\sf{}bob} (e);
\draw[->,green,dashed] (e) to [bend left] node [above] {\sf{}bob} (f);
\draw[->,green,dashed] (f) to [bend left] node [above] {\sf{}bob} (g);
\draw[->,green,dashed] (h) to [bend left] node [above] {\sf{}bob} (i);
\draw[->,green,dashed] (h) to [bend left] node [above] {\sf{}bob} (j);
\draw[->,cyan,dashed] (j) to [bend left] node [below] {\sf{}tlb\_barriered} (d);
\draw[->,green,dashed] (j) to [bend right] node [above] {\sf{}obtlbi\_translate $\subseteq$ ob} (f);
\draw [->,blue] (k.280) -- node [right] {wco} (f.80);
\draw[->,green,dashed] (j) to [bend right=50] node [above] {\sf{}obtlbi\_translate $\subseteq$ ob} (d);
\end{tikzpicture}
}
\end{center}

then we have \verb![i] ; maybe_TLB_cached ; [f]!,
and therefore \verb![j] ; obtlbi_translate ; [f]!,
and therefore \verb![j] ; ob ; [f]!
so there is a cycle in \verb!ob!.

\item
If \verb![f] ; wco ; [k]!:

\begin{center}
\scalebox{0.7}{
\begin{tikzpicture}
\xxevent{a}{W_v}{0,0}
\xxevent{b}{W_i}{2,0}
\xxevent{c}{dsbsy}{4,0}
\xxevent{d}{TLBI_{S2}}{6,0}
\xxevent{e}{dsbsy}{8,0}
\xxevent{f}{TLBI_{S1}}{10,0}
\xxevent{g}{dsbsy}{12,0}
\xxevent{h}{CSE}{14,0}
\xxevent{i}{T_{S1}}{16,0}
\xxevent{j}{T_{S2}}{18,0}
\xxevent{k}{W}{10,2}
\draw[->,blue] (a) -- node [above] {\sf{}co} (b);
\draw[->,green] (b) -- node [above] {\sf{}ob} (c);
\draw[->] (c) -- node [above] {\sf{}po} (d);
\draw[->] (d) -- node [above] {\sf{}po} (e);
\draw[->] (e) -- node [above] {\sf{}po} (f);
\draw[->] (f) -- node [above] {\sf{}po} (g);
\draw[->,green] (g) -- node [above] {\sf{}ob} (h);
\draw[->] (h) -- node [above] {\sf{}io} (i);
\draw[->] (i) -- node [above] {\sf{}iio} (j);
\draw[->] (d) to [bend right] node [above] {\sf{}tlb\_affects,ext} (j);
\draw[->] (f) to [bend right] node [above] {\sf{}tlb\_affects,ext} (i);
\draw[->,red] (a) to [bend right] node [below] {\sf{}trf,loc} (j);
\draw[->,red,dashed] (k) to [bend left] node [above] {\sf{}trf,loc} (i);
\draw[->,green,dashed] (c) to [bend left] node [above] {\sf{}bob} (d);
\draw[->,green,dashed] (d) to [bend left] node [above] {\sf{}bob} (e);
\draw[->,green,dashed] (e) to [bend left] node [above] {\sf{}bob} (f);
\draw[->,green,dashed] (f) to [bend left] node [above] {\sf{}bob} (g);
\draw[->,green,dashed] (h) to [bend left] node [above] {\sf{}bob} (i);
\draw[->,green,dashed] (h) to [bend left] node [above] {\sf{}bob} (j);
\draw[->,cyan,dashed] (j) to [bend left] node [below] {\sf{}tlb\_barriered} (d);
\draw[->,green,dashed] (j) to [bend right] node [above] {\sf{}obtlbi\_translate $\subseteq$ ob} (d);
\draw [->,blue] (f.100) -- node [right] {\sf{}wco} (k.260);
\end{tikzpicture}
}
\end{center}

then we have \verb![j] ; obtlbi_translate ; [d]!,
and therefore \verb![j] ; ob ; [d]!,
so there is a cycle in \verb!ob!.
\end{itemize}
\end{itemize}
\end{itemize}
\end{proof}

%\newpage

\subsection{Virtual address abstraction and anti-abstraction}

We consider a simple case when the virtual address abstraction ought to hold:
under some conditions, the model with translation and the original model without
translations coincide. Here, we only consider the consistency of the
pre-executions, but not how these pre-executions arise.

\subsubsection{Abstraction}

\begin{definition}[VA abstraction subcondition]
  \verb!G! satisfies the \emph{VA abstraction subcondition} when it has no
  page-table-affecting instructions: no TLBI, no context-changing operations
  (for example via writing to registers, for example via MSR TTBR), etc.
\end{definition}

\begin{definition}[VA abstraction condition]
  \verb!Gtr! satisfies the \emph{VA abstraction condition} when it satisfies the
  VA abstraction subcondition, and has a static injective page table.
\end{definition}

\begin{theorem}[VA abstraction]
\label{thm:vaabstraction}
  For all
  (\verb!Gtr : concrete execution!) \\
  if \verb!Gtr! is consistent wrt. the model with translation \\
  and respects the VA abstraction condition, then \\
  let \verb!Gabs! = erase \verb!Gtr! in \\
  \verb!Gabs! is consistent wrt. the model without translation.
\end{theorem}
\begin{proof}
  First, the builtin \verb!addr! of the abstract model is assumed to coincide with
  the derived \verb!addr! of the concrete model by the erasure. Showing that the two
  definitions of pre-executions do relate in this way is outside of our scope.
  Given that the definitions \verb!addr! coincide, the definitions of all the other
  derived relations of the abstract model, including \verb!ob! in the translation
  model, are syntactically supersets of their definition in the concrete model,
  so a cycle in \verb!ob! in the abstract model is also a cycle in \verb!ob! in the
  concrete model.
\end{proof}

\subsubsection{Anti-abstraction}

For this direction, we need to be able to put the translation table somewhere.

\paragraph{Step 1: Building the candidate execution in the translation model}

\begin{definition}[translation extension condition]
  The \emph{translation extension condition} is the data of \\
  (\verb!Gabs : execution!) \\
  such that \verb!Gabs! is consistent wrt. the model without translation \\
  and has no TLBI, and no MSR TTBR \\
  and \\
  (\verb!va_space : va_address -> bool!) \\
  such that all the memory accesses of \verb!Gabs! are in \verb!va_space! \\
  and \\
  (\verb!pt_pa_space : pa_address -> bool!) \\
  (\verb!pt_initial_state : pa_address -> option (list byte)!), \\
  such that the domains of \verb!pt_pa_space! and \verb!pt_initial_state! coincide \\
  and \\
  (\verb!tr_ctxt : translation_context!), \\
  such that \verb!id_map_lifted va_space! and \verb!pt_pa_space! are disjoint address spaces \\
  and \\
  (\verb!translate : translation_function!), \\
  such that translating \verb!abstract_va_space! translate-reads from within \verb!pt_pa_space! and gives the injective map.
\end{definition}

\begin{definition}[translation extension]
  Given the translation extension condition, the \emph{translation extension} \verb!Gtr! of
  \verb!Gabs! is constructed by:
  \begin{itemize}
  \item adding all the initial writes for the page tables,
  \item adding all the translate reads obtained by running the \verb!translate! function
    with the \verb!tr_ctxt!,
  \item adding the translate reads in \verb!iio! between the fetch and the explicit event,
  \item adding \verb!tdata! to match \verb!addr!,
  \item adding \verb!trf! from the corresponding initial writes to the translates.
  \end{itemize}
\end{definition}

\begin{definition}[VA anti abstraction condition]
  \verb!Gtr! satisfies the \emph{VA anti-abstraction condition} when it is derived from a
  consistent execution which satisfies the VA abstraction subcondition by the
  translation extension.
\end{definition}

\begin{lemma}[VA abstraction condition for translation extension]
  If \verb!Gtr! satisfies the \emph{VA anti-abstraction condition}, then \verb!Gtr! satisfies the
  VA abstraction condition.
\end{lemma}
\begin{proof}
  The translation extension does not add any extra instructions, and sets up
  static injective page tables.
\end{proof}

\begin{lemma}[obtlbi-empty]
  If \verb!Gtr! satisfies the VA anti-abstraction condition, then \verb!obtlbi! is empty.
\end{lemma}
\begin{proof}
  \verb!obtlbi! has
  \begin{itemize}
  \item \verb!obtlbi_translate!
    which has
    \begin{itemize}
    \item \verb!tcache1! \\
      which is
      \verb![T & Stage1] ; tfr ; tseq1!
      \\
      the latter is
      \\
      \verb![W] ; (maybe_TLB_barriered_by_va & ob) ; [TLBI VA]!
      \\
      which requires a TLBI, so it is empty
    \item \verb!tcache2 & ...! \\
      which requires a TLBI, so it is empty
    \item \verb!(tcache2 ; ...) & ...! \\
      which requires a TLBI, so it is empty
   \end{itemize}
  \item \verb![M] ; iio^-1 ; obtlbi_translate!
  \\
    to which the same reasoning applies
    \end{itemize}
\end{proof}

\paragraph{Step 2: Consistency}

\begin{lemma}
  If \verb!Gtr! satisfies the VA anti-abstraction condition, then
  translation-internal is acylic.
\end{lemma}
\begin{proof}
  \verb!po-pa; [W]; trf! is empty
  \\
  because by the VA anti-abstraction condition there are no non-initial writes to page tables.
\end{proof}

So we only need to show \verb!external! is acyclic.

\begin{lemma}[ob-to-T]
  If \verb!G! satisfies the VA anti-abstraction condition, then, for all $n \geq 1$,
\begin{verbatim}
  imm(ob)^n ; [T] ==
    iio
    | imm(ob)^(n-1) ; trfe
    | imm(ob)^(n-1) ; [T] ; iio ; [T]
    | imm(ob)^(n-1) ; [CSE] ; instruction-order
    | imm(ob)^(n-1) ; po ; [ERET] ; instruction-order ; [T]
\end{verbatim}
\end{lemma}
\begin{proof}
\begin{itemize}
\item  The \verb!addr! clause
  \\
  \verb!| tdata ; [T_f]!
  \\
  is empty because there are no translation failures.

\item  \verb!tob! does not contribute: there are no faults, and no non-initial writes to
  page table entries.

\item  The first clause of \verb!ctxob! is empty because there are no \verb!MSR TTBR!.
  The third and fourth are also empty, because they do not end in a \verb![T]!.

\item  Given a static injective mapping, the new \verb!| (addr | data | ctrl) ; trfi!
  clause of \verb!dob! is empty.
\end{itemize}
\end{proof}

\begin{lemma}[no-cycle-ob-to-init]
  If \verb!Gtr! is well-formed and consistent (in either model),
  then there is cycle in \verb!ob! via the initial writes.
\end{lemma}
\begin{proof}
  By well-formedness, \verb!wco ; [INIT] = [INIT] ; wco ; [INIT]!, and \verb!wco! is acyclic.
\\
  By examination of the other edges.
\end{proof}

\begin{lemma}[ob-from-T]
  If \verb!Gtr! satisfies the VA anti-abstraction condition, then
\begin{verbatim}
[T] ; imm(ob) ==
  iio
  | [T] ; iio ; [M] ; po ; [W]
\end{verbatim}
\end{lemma}
\begin{proof}
  By examination of the edges.
\end{proof}

\begin{lemma}[instruction-order-compress]{\ }\\
  \verb!instruction-order ; [T] ; iio ; [M] ; po! $\subseteq$ \verb!instruction-order!
\end{lemma}
\begin{proof}
  If we unfold the definitions of \verb!instruction-order! and \verb!po!, we have
  \\
  \verb!iio^-1 ; fpo ; iio ; [T] ; iio ; [M] ; [M|F|C] ; iio^-1 ; fpo ; iio ; [M|F|C]!
  \\
  which we can simplify into
  \\
  \verb!iio^-1 ; fpo ; fpo ; iio ; [M|F|C]!
  \\
  which means we have
  \\
  \verb!instruction-order!.
\end{proof}

\begin{lemma}[instruction-order-compress-iio]
  \verb!instruction-order ; iio ; po! $\subseteq$ \verb!instruction-order!
\end{lemma}
\begin{proof}
  \verb!iio! is transitive, and is the RHS of \verb!instruction-order!.
\end{proof}

\begin{lemma}[ob-acyclic-preserved]
  If \verb!G! satisfies the VA anti-abstraction condition,
  if there is a cycle in translate-\verb!ob!,
  then there is a cycle in plain-\verb!ob!.
\end{lemma}
\begin{proof}

  Consider a minimal cycle in translate-\verb!imm(ob)! (that is, the transitive
  closure of the \verb!ob! of the model with translation). Let $n$ be its length.
  \\
  We show that there is a cycle in plain-\verb!ob!.
  \\
  Assume, for contradiction, that the cycle contains an edge that is not in
  plain-\verb!ob! (that is, the \verb!ob! of the model without translation):
  \begin{itemize}

  \item iio
  \\
    by case split:
    \begin{itemize}

\item    \verb![T] ; iio; [M]!:
      by Lemma ob-to-T, the \verb!ob! edge to the left has to be either
      \begin{itemize}
      \item \verb!iio!
        in which case, by transitivity of \verb!iio!, there is a shorter cycle, so we
        have a contradiction.
        \\
        Let us call this Case IIOtrans.
      \item \verb!trfe!, which is from an initial write by the VA abstraction condition,
      \\
        but by Lemma no-cycle-ob-to-init, the cycle cannot exist.
      \item \verb!imm(ob)^(n-2); [T]; iio; [T]; iio; [M]!
      \\
        then we have
        \verb!imm(ob)^(n-2); [T]; iio; [M]!, which involves one fewer translate,
        \\
        so we have a contradiction.
      \item \verb!imm(ob)^(n-2) ; [CSE] ; instruction-order!
      \\
        This is similar to IIOtrans.
      \item \verb!imm(ob)^(n-2) ; po ; [ERET] ; instruction-order ; [T]!
      \\
        This is similar to IIOtrans.
      \end{itemize}

\item    \verb![T] ; iio ; [T]!:
\\
      So the whole cycle looks like
      \verb!imm(ob)^(n-1) ; [T] ; iio ; [T]!

      By Lemma ob-to-T, we have either
      \begin{itemize}
      \item \verb!imm(ob)^(n-2) ; iio ; [T] ; iio ; [T]!
      \\
        See Case IIOtrans.
      \item \verb!imm(ob)^(n-2) ; trfe!
      \\
        the \verb!trfe! is from an initial write by the VA abstraction condition,
        \\
        and by Lemma no-cycle-ob-to-init, the cycle cannot exist.
      \item \verb!imm(ob)^(n-2); [T]; iio; [T]!
      \\
        but we already have \verb!iio! to the second T,
        \\
        so we have a cycle involving one fewer translate,
        \\
        so we have a contradiction.
      \item \verb!imm(ob)^(n-2) ; [CSE] ; instruction-order!
      \\
        This is similar to IIOtrans.
      \item \verb!imm(ob)^(n-2) ; po ; [ERET] ; instruction-order ; [T]!
      \\
        This is similar to IIOtrans.
      \end{itemize}
\end{itemize}

\item \verb!tob! has
     \begin{itemize}
    \item \verb![T_f] ; tfr!
    \\
      which has a fault, so we have a contradiction.

    \item \verb!([T_f] ; tfri) & (po ; [dsb.sy] ; instruction-order)^-1!
    \\
      which has a fault, so we have a contradiction.

    \item \verb!speculative ; trfi!
      which is empty, because of the static page table.
   \end{itemize}

\item \verb!obtlbi!, which is empty by Lemma obtlbi-empty.

\item \verb!ctxob! has
  \begin{itemize}
  \item \verb!speculative ; [MSR TTBR]!
  \\
      by the VA abstraction condition, there is no MSR TTBR
    \item \verb![CSE] ; instruction-order!
      \\
      So the whole cycle looks like
      \\
      \verb![CSE] ; instruction-order ; imm(ob)^(n-1)!
      \\
      Because \verb!instruction-order! is acyclic, $n \geq 1$, so we have
      \\
      \verb![CSE] ; instruction-order ; imm(ob) ; imm(ob)^(n-2)!
      \\
      By Lemma ob-from-T, we have either:
      \begin{itemize}
      \item \verb![CSE] ; instruction-order ; iio ; imm(ob)^(n-2)!
      \\
        which means that by Lemma instruction-order-compress, we have
        \\
        \verb![CSE] ; instruction-order ; imm(ob)^(n-2)!
        \\
        so we have a cycle involving one edge fewer, so we have a contradiction.
      \item \verb![CSE] ; instruction-order ; [T] ; iio ; [M] ; po ; [W] ; imm(ob)^(n-2)!
      \\
        which means that by Lemma instruction-order-compress, we have
        \\
        \verb![CSE] ; instruction-order ; imm(ob)^(n-2)!
        \\
        so we have a cycle involving one edge fewer, so we have a contradiction.
       \end{itemize}

     \item \verb![ContextChange] ; po ; [CSE]!
     \\
      by the VA abstraction condition, there is no \verb!ContextChange!.

    \item \verb!speculative ; [CSE]!
    \\
      The \verb!CSE! has to be an \verb!ISB!, because there are no exceptions,
      and the \verb!speculative! is either
      in \verb!dob! in the plain model, so we have a contradiction,
      or in \verb![T]; instruction-order!.
      \\
      So the whole cycle looks like
      \verb!imm(ob)^(n-1) ; [T] ; iio ; [M] ; po ; [ISB]!

      Because \verb!po | iio! is acyclic, $n-1$ has to be $\geq 1$, so by Lemma ob-to-T, we have either
      \begin{itemize}
      \item \verb!imm(ob)^(n-2); iio; [T]; iio; [M]; po; [ISB]!
      \\
        See Case IIOtrans.
      \item \verb!trfe!, which is from an initial write by the VA abstraction condition,
      \\
        but by Lemma no-cycle-ob-to-init, the cycle cannot exist
      \item \verb!imm(ob)^(n-2); [T]; iio; [T]; iio; [M]; po; [ISB]!
      \\
        but we already have \verb!iio! to the second T,
        \\
        so we have a cycle involving one fewer translate,
        \\
        so we have a contradiction.
      \item \verb!imm(ob)^(n-2); [CSE] ; instruction-order ; [T] ; iio ; [M] ; po ; [ISB]!
      \\
        which means that by Lemma instruction-order-compress, we have
        \\
        \verb!imm(ob)^(n-2); [CSE] ; instruction-order!
        \\
        so we have a cycle involving one edge fewer,
        \\
        so we have a contradiction.
      \item \verb!imm(ob)^(n-2) ; po ; [ERET] ; instruction-order ; [T] ; iio ; [M] ; po ; [ISB]!
      \\
        is similar
      \end{itemize}
      
  \item \verb!po ; [ERET] ; instruction-order ; [T]!
  \\
      So the whole cycle looks like
      \\
      \verb!po ; [ERET] ; instruction-order ; [T] ; imm(ob)^(n-1)!
      \\
      Because instruction-order is acyclic, $n \geq 1$, so we have
      \\
      \verb!po ; [ERET] ; instruction-order ; [T] ; imm(ob) ; imm(ob)^(n-2)!
      \\
      By Lemma ob-from-T, we have either:
      \begin{itemize}
      \item \verb!po ; [ERET] ; instruction-order ; [T] ; iio ; imm(ob)^(n-2)!
      \\
        which means that by Lemma instruction-order-compress-iio, we have
        \\
        \verb!po ; [ERET] ; instruction-order ; imm(ob)^(n-2)!
        \\
        so we have a cycle involving one edge fewer, so we have a contradiction.
      \item \verb!po ; [ERET] ; instruction-order ; [T] ; ([T] ; iio ; [M]; po ;! \verb! [W]) ; imm(ob)^(n-2)!
      \\
        which means that by Lemma instruction-order-compress, we have
        \\
        \verb!po ; [ERET] ; instruction-order ; imm(ob)^(n-2)!
        \\
        so we have a cycle involving one edge fewer, so we have a contradiction.
      \end{itemize}

    \end{itemize}
  \item extended \verb!dob!:
    \begin{itemize}
    \item involving \verb!trfi! from non-initial writes,
    which contradicts our assumption about static translation.
    \item or \verb![T] ; instruction-order ; [W]!,
    \\
    so \verb![T] ; iio ; [M] ; po ; [W]!
\\
      So the whole cycle looks like
      \verb!imm(ob)^(n-1) ; [T] ; iio ; [M] ; po ; [W]!

      Because \verb!po | iio! is acyclic, $n-1$ has to be $\geq 1$, so by Lemma ob-to-T, we have either
      \begin{itemize}
      \item \verb!imm(ob)^(n-2); iio; [T]; iio; [M]; po; [W]!
      \\
        See Case IIOtrans.
      \item \verb!trfe!, which is from an initial write by the VA abstraction condition,
      \\
        but by Lemma no-cycle-ob-to-init, the cycle cannot exist
      \item \verb!imm(ob)^(n-2); [T]; iio; [T]; iio; [M]; po; [W]!
      \\
        but we already have \verb!iio! to the second T,
        \\
        so we have a cycle involving one fewer translate,
        \\
        so we have a contradiction.
      \item \verb!imm(ob)^(n-2); [CSE] ; instruction-order ; [T] ; iio ; [M] ; po ; [W]!
      \\
        which means that by Lemma instruction-order-compress, we have
        \\
        \verb!imm(ob)^(n-2); [CSE] ; instruction-order!
        \\
        so we have a cycle involving one edge fewer,
        \\
        so we have a contradiction.
      \item \verb!imm(ob)^(n-2) ; po ; [ERET] ; instruction-order ; [T] ; iio ; [M] ; po ; [W]!
      \\
        is similar
      \end{itemize}
      \end{itemize}

  \item extended \verb!bob!, but only involving TLBI,
    which contradicts our assumption of no TLBI.

 \item extended \verb!obs!, but only involving \verb!trfe!,
    by the VA abstraction condition, the only writes to page tables are from
    initial writes, and by Lemma no-cycle-ob-to-init, there are no \verb!ob! cycles via
    initial writes, so there is no cycle.
  \item \verb!obfault!, which involves a fault, which contradicts our assumptions.
  \item \verb!obets!, which involves a fault or a TLBI, which contradicts our assumptions.
  \end{itemize}

  All the other edges are in plain-\verb!ob! by definition.
\end{proof}

\begin{theorem}[VA anti-abstraction]
\label{thm:vaantiabstraction}
  If the translation extension condition holds,
  then there exists a \verb!Gtr! that satisfies the VA anti-abstraction condition such that
  \verb!Gtr! is a stitching of \verb!Gabs! with the \verb!pt_initial_state! according to \verb!translate! in \verb!tr_ctxt! and
  \verb!Gtr! is consistent wrt. the model with translation.
\end{theorem}
\begin{proof}
  \verb!Gtr! exists by the translation extension construction,
  \\
  and it is consistent by Lemma ob-acyclic-preserved.
\end{proof}

%\end{document}

\newpage
\myappendix{Test results}{app:results}%

\subsection{Isla model results}

Here $\checkmark$ and $\times$ indicate whether or not the model allows an execution with satisfying the final-state constraint given in the test.  All these are as intended. 

\begin{longtable}{p{.6\textwidth}@{}c@{}rc@{}r}

\end{longtable}

\newpage
\noindent Below are the pKVM tests. 
The two tests without results timed out after 4 hours.
\begin{longtable}{p{.7\textwidth}@{}cr}
&  \multicolumn{2}{c}{Strong model } \\
Test Name & allow? & \hfill time \\
pKVM.create\_hyp\_mappings.inv.l2\dotfill  & \hfill $\times$ & \hfill 2493ms \\
pKVM.create\_hyp\_mappings.inv.l3\dotfill  & \hfill $\times$ & \hfill 807ms \\
pKVM.host\_handle\_trap.free\_table.toml\dotfill  & \hfill $\times$ & \hfill 15718265ms \\
\begin{tabular}{@{}l@{}}pKVM.host\_handle\_trap \\ \hspace{.2cm} .stage2\_idmap.change\_block\_size \\ \end{tabular}\dotfill  & \hfill $\times$ & \hfill 7010271ms \\
\begin{tabular}{@{}l@{}}pKVM.host\_handle\_trap \\ \hspace{.2cm} .stage2\_idmap.change\_block\_size.change\_permissions \\ \end{tabular}\dotfill  & \hfill $\times$ & \hfill 2370057ms \\
pKVM.host\_handle\_trap.stage2\_idmap.l3\dotfill  & \hfill $\times$ & \hfill 166475ms \\
\begin{tabular}{@{}l@{}}pKVM.host\_handle\_trap \\ \hspace{.2cm} .stage2\_idmap.l3.already\_exists.concurrent \\ \end{tabular}\dotfill  & \hfill --- & \hfill --- \\
\begin{tabular}{@{}l@{}}pKVM.host\_handle\_trap \\ \hspace{.2cm} .stage2\_idmap.l3.already\_exists \\ \end{tabular}\dotfill  & \hfill --- & \hfill --- \\
pKVM.host\_handle\_trap\_twice.stage2\_idmap.l3\dotfill  & \hfill $\times$ & \hfill 149172ms \\
pKVM.switch\_to\_new\_table\dotfill  & \hfill $\times$ & \hfill 1405ms \\
pKVM.vcpu\_run\dotfill  & \hfill $\times$ & \hfill 956ms \\
pKVM.vcpu\_run.same\_vm\dotfill  & \hfill $\times$ & \hfill 3107ms \\
pKVM.vcpu\_run.update\_vmid\dotfill  & \hfill $\times$ & \hfill 1817ms \\
pKVM.vcpu\_run.update\_vmid.concurrent\dotfill  & \hfill $\times$ & \hfill 8604ms \\

\end{longtable}
\newpage
\subsection{Hardware results}

%\newpage
%\myappendix{Hardware Results}\label{app:hw}%
%\mysetheader{Appendix~\ref{app:hw}: Hardware results}

Below is a table of our results from running our hand-written
hardware tests on the various machines we have available:
a Raspberry Pi 4;
a Raspberry Pi 3B+;
and an AWS \cc{m6g-metal} instance (claiming to be an A76).
Our hardware test harness uses a different form of test to our Isla tooling; tests with the same name have manually-checked correspondence.
%\newpage
%\thispagestyle{empty}

{
\tiny
%\vspace*{-20mm}
\hspace*{-25mm}
\scalebox{0.7}{
\begin{minipage}{1.2\textwidth}

\end{minipage}
}
%\TODO{BS: does not fit on page... maybe mess with page size here?}
}

\newpage

\mysetheader{References}
\fancyhead[EOL]{\nouppercase{\leftmark}}

\bibliographystyle{plain}
\bibliography{bib}

\begin{thebibliography}{10}

\bibitem{Power2.07}
{\em Power ISA$^{\mbox{TM}}$ Version 2.07}.
\newblock IBM, 2013.

\bibitem{pKVM-src}
{pKVM} source.
\newblock
  \url{https://android-kvm.googlesource.com/linux/+/refs/heads/pkvm/arch/arm64/kvm/hyp/nvhe/},
  2021.
\newblock Accessed 2021-07-06.

\bibitem{Adir:2003}
Allon Adir, Hagit Attiya, and Gil Shurek.
\newblock Information-flow models for shared memory with an application to the
  {PowerPC} architecture.
\newblock {\em IEEE Trans. Parallel Distrib. Syst.}, 14(5):502--515, 2003.

\bibitem{Adve:1990:WON:325164.325100}
Sarita~V. Adve and Mark~D. Hill.
\newblock Weak ordering --- a new definition.
\newblock In {\em Proceedings of the 17th Annual International Symposium on
  Computer Architecture}, ISCA '90, pages 2--14, New York, NY, USA, 1990. ACM.

\bibitem{damp09}
Jade Alglave, Anthony Fox, Samin Ishtiaq, Magnus~O. Myreen, Susmit Sarkar,
  Peter Sewell, and Francesco Zappa~Nardelli.
\newblock The semantics of {Power} and {ARM} multiprocessor machine code.
\newblock In {\em Proc.~DAMP 2009}, January 2009.

\bibitem{herd7}
Jade Alglave and Luc Maranget.
\newblock The \texttt{herd7} tool.
\newblock \url{http://diy.inria.fr/doc/herd.html/}, 2019.
\newblock Accessed 2019-07-08.

\bibitem{cav2010}
Jade Alglave, Luc Maranget, Susmit Sarkar, and Peter Sewell.
\newblock Fences in weak memory models.
\newblock In {\em Proc. CAV}, 2010.

\bibitem{Alglave:2011:LRT:1987389.1987395}
Jade Alglave, Luc Maranget, Susmit Sarkar, and Peter Sewell.
\newblock Litmus: running tests against hardware.
\newblock In {\em Proceedings of TACAS 2011: the 17th international conference
  on Tools and Algorithms for the Construction and Analysis of Systems}, pages
  41--44, Berlin, Heidelberg, 2011. Springer-Verlag.

\bibitem{alglave:herd}
Jade Alglave, Luc Maranget, and Michael Tautschnig.
\newblock {Herding Cats: Modelling, Simulation, Testing, and Data Mining for
  Weak Memory}.
\newblock {\em ACM TOPLAS}, 36(2):7:1--7:74, July 2014.

\bibitem{DBLP:conf/fmcad/AlkassarCHKP10}
Eyad Alkassar, Ernie Cohen, Mark~A. Hillebrand, Mikhail Kovalev, and
  Wolfgang~J. Paul.
\newblock Verifying shadow page table algorithms.
\newblock In Roderick Bloem and Natasha Sharygina, editors, {\em Proceedings of
  10th International Conference on Formal Methods in Computer-Aided Design,
  {FMCAD} 2010, Lugano, Switzerland, October 20-23}, pages 267--270. {IEEE},
  2010.

\bibitem{DBLP:conf/vstte/AlkassarCKP12}
Eyad Alkassar, Ernie Cohen, Mikhail Kovalev, and Wolfgang~J. Paul.
\newblock Verification of {TLB} virtualization implemented in {C}.
\newblock In Rajeev Joshi, Peter M{\"{u}}ller, and Andreas Podelski, editors,
  {\em Verified Software: Theories, Tools, Experiments - 4th International
  Conference, {VSTTE} 2012, Philadelphia, PA, USA, January 28-29, 2012.
  Proceedings}, volume 7152 of {\em Lecture Notes in Computer Science}, pages
  209--224. Springer, 2012.

\bibitem{B.a}
{ARM Limited}.
\newblock {ARM} architecture reference manual. {ARMv8}, for {ARMv8-A}
  architecture profile.
\newblock \url{https://developer.arm.com/documentation/ddi0487/latest/}, March
  2017.
\newblock B.a Armv8.1 EAC, v8.2 Beta. ARM DDI 0487B.a (ID0331117). 6354pp.

\bibitem{G.a}
{Arm Limited}.
\newblock {Arm} architecture reference manual. {Armv8}, for {Armv8-A}
  architecture profile.
\newblock \url{https://developer.arm.com/documentation/ddi0487/latest/},
  January 2021.
\newblock G.a Armv8.7 EAC. ARM DDI 0487G.a (ID011921). 8538pp.

\bibitem{sail-popl2019}
Alasdair Armstrong, Thomas Bauereiss, Brian Campbell, Alastair Reid, Kathryn~E.
  Gray, Robert~M. Norton, Prashanth Mundkur, Mark Wassell, Jon French,
  Christopher Pulte, Shaked Flur, Ian Stark, Neel Krishnaswami, and Peter
  Sewell.
\newblock {ISA} semantics for {ARMv8-A, RISC-V, and CHERI-MIPS}.
\newblock In {\em Proc. 46th ACM SIGPLAN Symposium on Principles of Programming
  Languages}, January 2019.
\newblock Proc. ACM Program. Lang. 3, POPL, Article 71.

\bibitem{isla-cav}
Alasdair Armstrong, Brian Campbell, Ben Simner, Christopher Pulte, and Peter
  Sewell.
\newblock Isla: Integrating full-scale {ISA} semantics and axiomatic
  concurrency models.
\newblock In {\em In Proc. 33rd International Conference on Computer-Aided
  Verification}, July 2021.
\newblock Extended version available at
  \url{https://www.cl.cam.ac.uk/~pes20/isla/isla-cav2021-extended.pdf}.

\bibitem{DBLP:conf/types/BartheBCCL13}
Gilles Barthe, Gustavo Betarte, Juan~Diego Campo, Jes{\'{u}}s~Mauricio
  Chimento, and Carlos Luna.
\newblock Formally verified implementation of an idealized model of
  virtualization.
\newblock In Ralph Matthes and Aleksy Schubert, editors, {\em 19th
  International Conference on Types for Proofs and Programs, {TYPES} 2013,
  April 22-26, 2013, Toulouse, France}, volume~26 of {\em LIPIcs}, pages
  45--63. Schloss Dagstuhl - Leibniz-Zentrum f{\"{u}}r Informatik, 2013.

\bibitem{DBLP:conf/fm/BartheBCL11}
Gilles Barthe, Gustavo Betarte, Juan~Diego Campo, and Carlos Luna.
\newblock Formally verifying isolation and availability in an idealized model
  of virtualization.
\newblock In Michael~J. Butler and Wolfram Schulte, editors, {\em {FM} 2011:
  Formal Methods - 17th International Symposium on Formal Methods, Limerick,
  Ireland, June 20-24, 2011. Proceedings}, volume 6664 of {\em Lecture Notes in
  Computer Science}, pages 231--245. Springer, 2011.

\bibitem{DBLP:conf/csfw/BartheBCL12}
Gilles Barthe, Gustavo Betarte, Juan~Diego Campo, and Carlos Luna.
\newblock Cache-leakage resilient {OS} isolation in an idealized model of
  virtualization.
\newblock In Stephen Chong, editor, {\em 25th {IEEE} Computer Security
  Foundations Symposium, {CSF} 2012, Cambridge, MA, USA, June 25-27, 2012},
  pages 186--197. {IEEE} Computer Society, 2012.

\bibitem{DBLP:conf/aplas/BartheKS08}
Gilles Barthe, C{\'{e}}sar Kunz, and Jorge~Luis Sacchini.
\newblock Certified reasoning in memory hierarchies.
\newblock In G.~Ramalingam, editor, {\em Programming Languages and Systems, 6th
  Asian Symposium, {APLAS} 2008, Bangalore, India, December 9-11, 2008.
  Proceedings}, volume 5356 of {\em Lecture Notes in Computer Science}, pages
  75--90. Springer, 2008.

\bibitem{cpp-popl}
M.~Batty, S.~Owens, S.~Sarkar, P.~Sewell, and T.~Weber.
\newblock Mathematizing {C++} concurrency.
\newblock In {\em Proc.~POPL}, 2011.

\bibitem{popl2012}
Mark Batty, Kayvan Memarian, Scott Owens, Susmit Sarkar, and Peter Sewell.
\newblock {Clarifying and Compiling C/C++ Concurrency: from C++11 to POWER}.
\newblock In {\em Proceedings of POPL 2012: The 39th ACM SIGPLAN-SIGACT
  Symposium on Principles of Programming Languages (Philadelphia)}, pages
  509--520, 2012.

\bibitem{bohemadvec}
H.-J. Boehm and S.~Adve.
\newblock Foundations of the {C++} concurrency memory model.
\newblock In {\em Proc.~PLDI}, 2008.

\bibitem{DBLP:conf/pldi/BornholtT17}
James Bornholt and Emina Torlak.
\newblock Synthesizing memory models from framework sketches and litmus tests.
\newblock In Albert Cohen and Martin~T. Vechev, editors, {\em Proceedings of
  the 38th {ACM} {SIGPLAN} Conference on Programming Language Design and
  Implementation, {PLDI} 2017, Barcelona, Spain, June 18-23, 2017}, pages
  467--481. {ACM}, 2017.

\bibitem{DBLP:books/daglib/0073498}
William~W. Collier.
\newblock {\em Reasoning about parallel architectures}.
\newblock Prentice Hall, 1992.

\bibitem{seL4-the-proof}
Data61/CSIRO.
\newblock Frequently asked questions on {seL4}: The proof.
\newblock \url{http://sel4.systems/Info/FAQ/proof.pml}, accessed 2019-07-01,
  2019.

\bibitem{deacon-cat}
Will Deacon.
\newblock The {ARMv8} application level memory model.
\newblock
  \url{https://github.com/herd/herdtools7/blob/master/herd/libdir/aarch64.cat}
  (accessed 2019-07-01), 2016.

\bibitem{pKVM-linux-talk}
Will Deacon.
\newblock Virtualization for the masses: Exposing {KVM} on {Android}.
\newblock \url{https://www.youtube.com/watch?v=wY-u6n75iXc}, November 2020.
\newblock KVM Forum Talk.

\bibitem{DBLP:phd/dnb/Degenbaev12}
Ulan Degenbaev.
\newblock {\em Formal specification of the x86 instruction set architecture}.
\newblock PhD thesis, Saarland University, 2012.

\bibitem{DBLP:conf/birthday/DegenbaevPS09}
Ulan Degenbaev, Wolfgang~J. Paul, and Norbert Schirmer.
\newblock Pervasive theory of memory.
\newblock In Susanne Albers, Helmut Alt, and Stefan N{\"{a}}her, editors, {\em
  Efficient Algorithms, Essays Dedicated to Kurt Mehlhorn on the Occasion of
  His 60th Birthday}, volume 5760 of {\em Lecture Notes in Computer Science},
  pages 74--98. Springer, 2009.

\bibitem{pKVM-lwn}
Jake Edge.
\newblock {KVM} for {Android}.
\newblock \url{https://lwn.net/Articles/836693/}, November 2020.

\bibitem{FGP16}
Shaked Flur, Kathryn~E. Gray, Christopher Pulte, Susmit Sarkar, Ali Sezgin, Luc
  Maranget, Will Deacon, and Peter Sewell.
\newblock Modelling the {ARMv8} architecture, operationally: Concurrency and
  {ISA}.
\newblock In {\em Proceedings of POPL: the 43rd ACM SIGPLAN-SIGACT Symposium on
  Principles of Programming Languages}, 2016.

\bibitem{mixed17}
Shaked Flur, Susmit Sarkar, Christopher Pulte, Kyndylan Nienhuis, Luc Maranget,
  Kathryn~E. Gray, Ali Sezgin, Mark Batty, and Peter Sewell.
\newblock Mixed-size concurrency: {ARM}, {POWER}, {C/C++11}, and {SC}.
\newblock In {\em The 44st Annual {ACM} {SIGPLAN-SIGACT} Symposium on
  Principles of Programming Languages, Paris, France}, pages 429--442, January
  2017.

\bibitem{DBLP:journals/jpdc/GharachorlooAGHH92}
Kourosh Gharachorloo, Sarita~V. Adve, Anoop Gupta, John~L. Hennessy, and
  Mark~D. Hill.
\newblock Programming for different memory consistency models.
\newblock {\em J. Parallel Distributed Comput.}, 15(4):399--407, 1992.

\bibitem{GoelPhD}
Shilpi Goel.
\newblock {\em Formal Verification of Application and System Programs Based on
  a Validated x86 ISA Model}.
\newblock PhD thesis, University of Texas at Austin, 2016.
\newblock \url{https://repositories.lib.utexas.edu/handle/2152/46437}.

\bibitem{DBLP:books/sp/17/GoelHK17}
Shilpi Goel, Warren A.~Hunt Jr., and Matt Kaufmann.
\newblock Engineering a formal, executable x86 {ISA} simulator for software
  verification.
\newblock In {\em Provably Correct Systems}, pages 173--209. 2017.

\bibitem{micro2015}
Kathryn~E. Gray, Gabriel Kerneis, Dominic Mulligan, Christopher Pulte, Susmit
  Sarkar, and Peter Sewell.
\newblock An integrated concurrency and core-{ISA} architectural envelope
  definition, and test oracle, for {IBM POWER} multiprocessors.
\newblock In {\em Proc.~MICRO-48, the 48th Annual IEEE/ACM International
  Symposium on Microarchitecture}, December 2015.

\bibitem{DBLP:conf/osdi/GuSCWKSC16}
Ronghui Gu, Zhong Shao, Hao Chen, Xiongnan~(Newman) Wu, Jieung Kim, Vilhelm
  Sj{\"{o}}berg, and David Costanzo.
\newblock {CertiKOS}: An extensible architecture for building certified
  concurrent {OS} kernels.
\newblock In {\em 12th {USENIX} Symposium on Operating Systems Design and
  Implementation, {OSDI} 2016, Savannah, GA, USA, November 2-4, 2016.}, pages
  653--669, 2016.

\bibitem{DBLP:journals/jcs/GuancialeNDB16}
Roberto Guanciale, Hamed Nemati, Mads Dam, and Christoph Baumann.
\newblock Provably secure memory isolation for linux on {ARM}.
\newblock {\em J. Comput. Secur.}, 24(6):793--837, 2016.

\bibitem{DBLP:conf/isca/HossainTM20}
Naorin Hossain, Caroline Trippel, and Margaret Martonosi.
\newblock Transform: Formally specifying transistency models and synthesizing
  enhanced litmus tests.
\newblock In {\em 47th {ACM/IEEE} Annual International Symposium on Computer
  Architecture, {ISCA} 2020, Valencia, Spain, May 30 - June 3, 2020}, pages
  874--887. {IEEE}, 2020.

\bibitem{Klein_AEMSKH_14}
Gerwin Klein, June Andronick, Kevin Elphinstone, Toby Murray, Thomas Sewell,
  Rafal Kolanski, and Gernot Heiser.
\newblock Comprehensive formal verification of an {OS} microkernel.
\newblock {\em ACM Transactions on Computer Systems}, 32(1):2:1--2:70, February
  2014.

\bibitem{DBLP:phd/basesearch/Kolanski11}
Rafal Kolanski.
\newblock {\em Verification of programs in virtual memory using separation
  logic}.
\newblock PhD thesis, University of New South Wales, Sydney, Australia, 2011.

\bibitem{DBLP:phd/dnb/Kovalev13}
Mikhail Kovalev.
\newblock {\em {TLB} virtualization in the context of hypervisor verification}.
\newblock PhD thesis, Saarland University, 2013.

\bibitem{DBLP:conf/uss/LiLGNH21}
Shih{-}Wei Li, Xupeng Li, Ronghui Gu, Jason Nieh, and John~Zhuang Hui.
\newblock Formally verified memory protection for a commodity multiprocessor
  hypervisor.
\newblock In Michael Bailey and Rachel Greenstadt, editors, {\em 30th {USENIX}
  Security Symposium, {USENIX} Security 2021, August 11-13, 2021}, pages
  3953--3970. {USENIX} Association, 2021.

\bibitem{Li2021}
Shih-Wei Li, Xupeng Li, Ronghui Gu, Jason Nieh, and John~Zhuang Hui.
\newblock A secure and formally verified {Linux KVM} hypervisor.
\newblock In {\em 2021 IEEE Symposium on Security and Privacy (SP)}, pages
  839--856, Los Alamitos, CA, USA, may 2021. IEEE Computer Society.

\bibitem{arm-memory-model-tool}
Arm Limited.
\newblock
  \url{https://developer.arm.com/architectures/cpu-architecture/a-profile/memory-model-tool},
  2022.
\newblock Accessed 2022-02-23.

\bibitem{tut}
Luc Maranget, Susmit Sarkar, and Peter Sewell.
\newblock A tutorial introduction to the {ARM} and {POWER} relaxed memory
  models.
\newblock Draft available from
  \url{http://www.cl.cam.ac.uk/~pes20/ppc-supplemental/test7.pdf}, 2012.

\bibitem{tphols09}
Scott Owens, Susmit Sarkar, and Peter Sewell.
\newblock A better x86 memory model: {x86-TSO}.
\newblock In {\em Proceedings of TPHOLs 2009: Theorem Proving in Higher Order
  Logics, LNCS 5674}, pages 391--407, 2009.

\bibitem{pultethesis}
Christopher Pulte.
\newblock {\em The Semantics of Multicopy Atomic {ARMv8} and {RISC-V}}.
\newblock PhD thesis, University of Cambridge, 2019.
\newblock \url{https://doi.org/10.17863/CAM.39379}.

\bibitem{armv8-mca}
Christopher Pulte, Shaked Flur, Will Deacon, Jon French, Susmit Sarkar, and
  Peter Sewell.
\newblock {Simplifying ARM Concurrency: Multicopy-atomic Axiomatic and
  Operational Models for ARMv8}.
\newblock In {\em Proceedings of the 45th ACM SIGPLAN Symposium on Principles
  of Programming Languages}, January 2018.

\bibitem{DBLP:conf/pldi/PultePKLH19}
Christopher Pulte, Jean Pichon{-}Pharabod, Jeehoon Kang, Sung~Hwan Lee, and
  Chung{-}Kil Hur.
\newblock Promising-{ARM}/{RISC-V}: a simpler and faster operational
  concurrency model.
\newblock In Kathryn~S. McKinley and Kathleen Fisher, editors, {\em Proceedings
  of the 40th {ACM} {SIGPLAN} Conference on Programming Language Design and
  Implementation, {PLDI} 2019, Phoenix, AZ, USA, June 22-26, 2019}, pages
  1--15. {ACM}, 2019.

\bibitem{persistence-tso}
Azalea Raad and Viktor Vafeiadis.
\newblock Persistence semantics for weak memory: Integrating epoch persistency
  with the tso memory model.
\newblock {\em Proc. ACM Program. Lang.}, 2({OOPSLA}), oct 2018.

\bibitem{pldi2012}
Susmit Sarkar, Kayvan Memarian, Scott Owens, Mark Batty, Peter Sewell, Luc
  Maranget, Jade Alglave, and Derek Williams.
\newblock Synchronising {C/C++} and {POWER}.
\newblock In {\em Proceedings of {PLDI 2012}, the 33rd {ACM SIGPLAN conference
  on Programming Language Design and Implementation (Beijing)}}, pages
  311--322, 2012.

\bibitem{pldi105}
Susmit Sarkar, Peter Sewell, Jade Alglave, Luc Maranget, and Derek Williams.
\newblock Understanding {POWER} multiprocessors.
\newblock In {\em Proceedings of PLDI 2011: the 32nd ACM SIGPLAN conference on
  Programming Language Design and Implementation}, pages 175--186, 2011.

\bibitem{x86popl}
Susmit Sarkar, Peter Sewell, Francesco Zappa~Nardelli, Scott Owens, Tom Ridge,
  Thomas Braibant, Magnus Myreen, and Jade Alglave.
\newblock The semantics of {x86-CC} multiprocessor machine code.
\newblock In {\em Proceedings of POPL 2009: the 36th annual ACM SIGPLAN-SIGACT
  symposium on Principles of Programming Languages}, pages 379--391, January
  2009.

\bibitem{cacm}
Peter Sewell, Susmit Sarkar, Scott Owens, Francesco Zappa~Nardelli, and
  Magnus~O. Myreen.
\newblock {x86-TSO}: A rigorous and usable programmer's model for x86
  multiprocessors.
\newblock {\em Communications of the ACM}, 53(7):89--97, July 2010.
\newblock (Research Highlights).

\bibitem{esop2022}
Ben Simner, Alasdair Armstrong, Jean Pichon{-}Pharabod, Christopher Pulte,
  Richard Grisenthwaite, and Peter Sewell.
\newblock Relaxed virtual memory in {Armv8-A}.
\newblock In {\em Proceedings of ESOP 2022: 31st European Symposium on
  Programming, held as part of the European Joint Conferences on Theory and
  Practice of Software, {ETAPS} 2022, Munich, Germany}, April 2022.

\bibitem{iflat-esop2020-extended}
Ben Simner, Shaked Flur, Christopher Pulte, Alasdair Armstrong, Jean
  Pichon-Pharabod, Luc Maranget, and Peter Sewell.
\newblock {ARMv8-A} system semantics: instruction fetch in relaxed
  architectures (extended version).
\newblock In {\em Proceedings of the 29th European Symposium on Programming},
  April 2020.

\bibitem{DBLP:conf/lpar/SyedaK17}
Hira Syeda and Gerwin Klein.
\newblock Reasoning about translation lookaside buffers.
\newblock In {\em LPAR-21, 21st International Conference on Logic for
  Programming, Artificial Intelligence and Reasoning, Maun, Botswana, May 7-12,
  2017}, pages 490--508, 2017.

\bibitem{DBLP:phd/basesearch/Syeda19}
Hira~Taqdees Syeda.
\newblock {\em Low-level program verification under cached address
  translation}.
\newblock PhD thesis, University of New South Wales, Sydney, Australia, 2019.

\bibitem{DBLP:conf/itp/SyedaK18}
Hira~Taqdees Syeda and Gerwin Klein.
\newblock Program verification in the presence of cached address translation.
\newblock In {\em Interactive Theorem Proving - 9th International Conference,
  {ITP} 2018, Held as Part of the Federated Logic Conference, FloC 2018,
  Oxford, UK, July 9-12, 2018, Proceedings}, pages 542--559, 2018.

\bibitem{DBLP:journals/jar/SyedaK20}
Hira~Taqdees Syeda and Gerwin Klein.
\newblock Formal reasoning under cached address translation.
\newblock {\em J. Autom. Reason.}, 64(5):911--945, 2020.

\bibitem{Gu2021}
Runzhou Tao, Jianan Yao, Xupeng Li, Shih-Wei Li, Jason Nieh, and Ronghui Gu.
\newblock Formal verification of a multiprocessor hypervisor on arm relaxed
  memory hardware.
\newblock In {\em SOSP 2021: Proceedings of the 28th ACM Symposium on Operating
  Systems Principles}, October 2021.

\bibitem{DBLP:journals/jar/TewsVW09}
Hendrik Tews, Marcus V{\"{o}}lp, and Tjark Weber.
\newblock Formal memory models for the verification of low-level
  operating-system code.
\newblock {\em J. Autom. Reason.}, 42(2-4):189--227, 2009.

\bibitem{DBLP:conf/asplos/TrippelMLPM17}
Caroline Trippel, Yatin~A. Manerkar, Daniel Lustig, Michael Pellauer, and
  Margaret Martonosi.
\newblock Tricheck: Memory model verification at the trisection of software,
  hardware, and {ISA}.
\newblock In Yunji Chen, Olivier Temam, and John Carter, editors, {\em
  Proceedings of the Twenty-Second International Conference on Architectural
  Support for Programming Languages and Operating Systems, {ASPLOS} 2017,
  Xi'an, China, April 8-12, 2017}, pages 119--133. {ACM}, 2017.

\bibitem{riscv-20181221-Public-Review-draft}
Andrew Waterman and Krste Asanovi{\'{c}}, editors.
\newblock {\em The {RISC-V} Instruction Set Manual {Volume I: Unprivileged
  ISA}}.
\newblock December 2018.
\newblock Document Version {20181221-Public-Review-draft. Contributors: Arvind,
  Krste Asanovi\'{c}, Rimas Avi\v{z}ienis, Jacob Bachmeyer, Christopher F.
  Batten, Allen J. Baum, Alex Bradbury, Scott Beamer, Preston Briggs,
  Christopher Celio, Chuanhua Chang, David Chisnall, Paul Clayton, Palmer
  Dabbelt, Roger Espasa, Shaked Flur, Stefan Freudenberger, Jan Gray, Michael
  Hamburg, John Hauser, David Horner, Bruce Hoult, Alexandre Joannou, Olof
  Johansson, Ben Keller, Yunsup Lee, Paul Loewenstein, Daniel Lustig, Yatin
  Manerkar, Luc Maranget, Margaret Martonosi, Joseph Myers, Vijayanand
  Nagarajan, Rishiyur Nikhil, Jonas Oberhauser, Stefan O'Rear, Albert Ou, John
  Ousterhout, David Patterson, Christopher Pulte, Jose Renau, Colin Schmidt,
  Peter Sewell, Susmit Sarkar, Michael Taylor, Wesley Terpstra, Matt Thomas,
  Tommy Thorn, Caroline Trippel, Ray VanDeWalker, Muralidaran Vijayaraghavan,
  Megan Wachs, Andrew Waterman, Robert Watson, Derek Williams, Andrew Wright,
  Reinoud Zandijk, and Sizhuo Zhang}.

\bibitem{DBLP:conf/popl/WickersonBSC17}
John Wickerson, Mark Batty, Tyler Sorensen, and George~A. Constantinides.
\newblock Automatically comparing memory consistency models.
\newblock In Giuseppe Castagna and Andrew~D. Gordon, editors, {\em Proceedings
  of the 44th {ACM} {SIGPLAN} Symposium on Principles of Programming Languages,
  {POPL} 2017, Paris, France, January 18-20, 2017}, pages 190--204. {ACM},
  2017.

\bibitem{DBLP:conf/ipps/YangGLS04}
Yue Yang, Ganesh Gopalakrishnan, Gary Lindstrom, and Konrad Slind.
\newblock Nemos: {A} framework for axiomatic and executable specifications of
  memory consistency models.
\newblock In {\em 18th International Parallel and Distributed Processing
  Symposium {(IPDPS} 2004), {CD-ROM} / Abstracts Proceedings, 26-30 April 2004,
  Santa Fe, New Mexico, {USA}}, 2004.

\end{thebibliography}

\end{document}